\documentclass[11pt]{article}
\usepackage[margin=1in]{geometry}
\usepackage{amsmath}
\usepackage{amsthm}
\usepackage{amssymb}
\usepackage{color}
\usepackage[dvipsnames]{xcolor}
\usepackage{hyperref}
\usepackage{cleveref}
\usepackage[linesnumbered,ruled,vlined,longend]{algorithm2e}
\usepackage{tikz}
\usepackage{capt-of}
\usetikzlibrary{shapes.geometric}
\hypersetup{colorlinks=true,citecolor=OliveGreen,linkcolor=blue,urlcolor=black}
\usepackage{boldline}

\newcommand{\wt}{\widetilde}
\newcommand{\wh}{\widehat}
\newcommand{\init}{\mathrm{init}}
\newcommand{\final}{\mathrm{final}}

\newcommand{\start}{\mathrm{start}}
\newcommand{\tot}{\mathrm{tot}}

\newcommand{\cnt}{\mathrm{cnt}}
\newcommand{\polylogn}{\mathrm{poly}\log(n)}

\DeclareMathOperator*{\R}{{\mathbb{R}}}
\DeclareMathOperator*{\E}{{\mathbb{E}}}

\DeclareMathOperator*{\poly}{\mathrm{poly}}

\DeclareMathOperator*{\vol}{\mathrm{vol}}

\DeclareMathOperator*{\ext}{\mathrm{ext}}
\DeclareMathOperator*{\old}{\mathrm{old}}
\DeclareMathOperator*{\nnz}{\mathrm{nnz}}
\DeclareMathOperator*{\DS}{\mathrm{DS}}
\DeclareMathOperator*{\valid}{\mathrm{valid}}
\DeclareMathOperator*{\normrm}{\mathrm{norm}}

\newtheorem{theorem}{Theorem}[section]
\newtheorem{lemma}[theorem]{Lemma}
\newtheorem{definition}[theorem]{Definition}
\newtheorem{corollary}[theorem]{Corollary}

\newtheorem{fact}[theorem]{Fact}

\newcommand{\ov}{\overline}
\newcommand{\hpv}{h^{\perp v}}
\newcommand{\hv}{h^{v}}
\newcommand{\gpv}{g^{\perp v}}
\newcommand{\gv}{g^{v}}

\newcommand{\imbal}{\mathrm{im}}

\global\long\def\norm#1{\|#1\|}%

\newcommand{\mvar}[1]{\mathbf{#1}}
\global\long\def\balloss{\beta_{\eta}}%
\global\long\def\epsad{\epsilon_{\mathrm{add}}}%
    
\global\long\def\ma{\mvar A}%
\global\long\def\mb{\mvar B}%
\global\long\def\md{\mvar D}%
\global\long\def\mg{\mvar G}%

\global\long\def\mh{\mvar H}%

\global\long\def\mI{\mvar I}%

\global\long\def\mj{\mvar J}%
\global\long\def\mL{\mvar L}%

\global\long\def\mm{\mvar M}%
 
\global\long\def\mn{\mvar N}%

\global\long\def\mr{\mvar R}%
\global\long\def\ms{\mvar S}%
\global\long\def\mt{\mvar T}%
\global\long\def\mv{\mvar V}%

\global\long\def\mw{\mvar W}%
\global\long\def\mz{\mvar Z}%

\global\long\def\mzero{\mvar 0}%
\global\long\def\mntilde{\widetilde{\mn}_G}%

\global\long\def\onesVec{\vec{1}}%

\global\long\def\defeq{\stackrel{\mathrm{{\scriptscriptstyle def}}}{=}}%

\global\long\def\gbar{\overline{G}}%
\global\long\def\polylogn{\mathrm{polylog}(n)}%
\global\long\def\polylog{\mathrm{polylog}}%

\global\long\def\mdiag{\mvar{diag}}%

\global\long\def\indicVec#1{\onesVec_{#1}}%

\DeclareMathAlphabet{\bbold}{U}{bbold}{m}{n}%
\global\long\def\allone{\bbold{1}}%
\global\long\def\allzero{\bbold{0}}%

{\makeatletter
	\gdef\xxxmark{%
		\expandafter\ifx\csname @mpargs\endcsname\relax %
		\expandafter\ifx\csname @captype\endcsname\relax %
		\marginpar{xxx}%
		\else
		xxx %
		\fi
		\else
		xxx %
		\fi}
	\gdef\xxx{\@ifnextchar[\xxx@lab\xxx@nolab}
	\long\gdef\xxx@lab[#1]#2{{\bf \color{blue} [\xxxmark #2 ---{\sc #1}]}}
	\long\gdef\xxx@nolab#1{{\bf \color{blue} [\xxxmark #1]}}
}

\begin{document}
\title{Generalized Flow in Nearly-linear Time\\ on Moderately Dense Graphs}
\date{}
\author{Shunhua Jiang\footnote{\url{sj3005@columbia.edu}. Columbia University.} \qquad Michael Kapralov\footnote{\url{michael.kapralov@epfl.ch}. \'Ecole Polytechnique F\'ed\'erale de Lausanne} \qquad Lawrence Li\footnote{\url{lawrenceli@cs.toronto.edu}. University of Toronto.} \qquad Aaron Sidford\footnote{\url{sidford@stanford.edu}. Stanford University.}}

\maketitle

\begin{abstract}
In this paper we consider generalized flow problems where there is an $m$-edge $n$-node directed graph $G = (V,E)$ and each edge $e \in E$ has a loss factor $\gamma_e >0$ governing whether the flow is increased or decreased as it crosses edge $e$. We provide a randomized $\wt{O}( (m + n^{1.5}) \cdot \polylog(\frac{W}{\delta}))$ time algorithm for solving the generalized maximum flow and generalized minimum cost flow problems in this setting where $\delta$ is the target accuracy and $W$ is the maximum of all costs, capacities, and loss factors and their inverses. %
This improves upon the previous state-of-the-art $\wt{O}(m \sqrt{n} \cdot \log^2(\frac{W}{\delta}) )$ time algorithm, obtained by combining the algorithm of \cite{ds08} with techniques from \cite{ls14}.
To obtain this result we provide new dynamic data structures and spectral results regarding the matrices associated to generalized flows and apply them through the interior point method framework of \cite{bll+21}. %
\end{abstract}

\thispagestyle{empty}
\newpage
\thispagestyle{empty}
\tableofcontents

\newpage
\setcounter{page}{1}
\section{Introduction}

In this paper we consider \emph{generalized flow problems} where there is a directed graph $G = (V,E)$ with $n$-nodes $V$, $m$-edges $E$, together with positive loss factors $\gamma \in \R^{E}_{> 0}$ for edges, and positive edge capacities $u \in \R^{E}_{> 0}$. We call any $f \in \R^{E}_{\geq 0}$ a \emph{flow} and define the \emph{imbalance} of $f$ at $a \in V$, as
\[
\mathrm{im}_{G}(f)_a \defeq \sum_{e = (b,a) \in E} \gamma_e f_e -\sum_{e = (a,b) \in E} f_e.
\]
The imbalance of $f$ at $a$ denotes the net amount of flow into $a$. 
Whereas in standard flow problems increasing the flow on edge $e = (a,b)$ by $\alpha$ decreases the imbalance at $a$ by $\alpha$ and increases it at $b$ by $\alpha$, here the increase at $b$ is instead $\gamma_e \alpha$. Consequently, $\gamma_e < 1$ can be viewed as flow leaving the graph, i.e., being \emph{lost}, as flow crosses edge $e$, and we refer to such problems with $\gamma_e < 1$ as \emph{lossy flow problems}.
Conversely, $\gamma_e > 1$ can be viewed as flow entering the graph as the flow crosses $e$.

In this paper we consider the problems of solving\footnote{Unless specified otherwise, we use ``solving'' to refer to producing a high-accuracy approximately-feasible solution.} canonical optimization problems for graphs with loss factors. In particular we consider the following:
\begin{itemize}
    \item \emph{Generalized Maximum Flow}: In this problem there is a specified $s,t \in V$ and the goal is to find flow $f \in \R^{E}_{\geq 0}$ with $f \leq u$ entrywise and $\imbal_G(f)_a = 0$ for all $a \notin \{s,t\}$ such that $\imbal_G(f)_t$ is maximized. 
    This corresponds to maximizing the flow into $t$ given an unlimited supply at $s$ while respecting the capacity constraints.
    \item \emph{Generalized Minimum Cost Flow}: In this problem there is a specified $d \in \R^V$ and $c \in \R^E$ and the goal is to find  $f \in \R^{E}_{\geq 0}$ with $f \leq u$ entrywise and $\imbal_G(f) = d$ such that $c^\top f$ is minimized.
\end{itemize}

In the special case when $\gamma_e = 1$ for all $e \in E$, the generalized maxflow problem is the standard maxflow problem and the generalized min-cost flow problem is the standard min-cost flow problem. %
A line of work over the past decade has obtained numerous improvements to the running times for solving these standard problems \cite{ckmst11,m16,ls20,blss20,bln+20,bll+21,glp22,bgs22}. This line of work culminated in \cite{cklpps23,bcklppss23, bcklmms24}, which give deterministic almost linear time $O(m^{1+o(1)})$ algorithms for solving maximum flow and minimum cost flow to high-accuracy, and a randomized algorithm \cite{bll+21} that solves these problems in nearly linear time $\wt{O}( m+n^{1.5})$ on moderately dense graphs where $m \geq n^{1.5}$.\footnote{$\widetilde{O}(m)$ hides $\polylog(m)$ factors. Almost linear time refers to $O(m^{1 + o(1)})$ running times, and nearly linear time refers to $\wt{O}(m)$ running times. In just the introduction, we assume that the problem magnitude parameter $W$ and the error parameter $1/\delta$ are both bounded by $\poly(n)$, and hide $\polylog(W/\delta)$ factors in $\wt{O}(\cdot)$. 
} 
Interestingly, the almost linear time algorithms even apply to more general convex flow problems where the goal is to minimize sums of convex cost functions applied to the edges while routing  specified demands $d$: 
\[
\min_{\substack{f \in \R^{E}:~ \mathrm{im}(f)_v = d_v ~\forall v \in V \\ \allzero \leq f \leq u }} \sum_{e \in E} \phi_e(f_e),
\text{ where each $\phi_e : \R \rightarrow \R$ is a convex.}
\]
However, there is no known almost linear time reduction from generalized flow problems to min-cost flow, or even to this more general convex flow problem.

Despite these advances and multiple works studying generalized flows \cite{gpt88,v89,gjo97,fw02,ds08}, the state-of-the-art running times for solving generalized maxflow and generalized min-cost flow are substantially slower than those for maximum flow and minimum cost flow. %
These state-of-the-art running times include the $\wt{O}(m^{1.5})$ time algorithm of \cite{ds08},\footnote{The algorithm of \cite{ds08} was stated for lossy maximum flow, for which it produces \emph{exact feasible} solutions. Careful inspection shows that their algorithm can also be applied to generalized maximum flow and generalized min-cost flow to produce \emph{approximate feasible} solutions. 
Our algorithm also outputs approximate feasible solutions for generalized maximum flow and generalized min-cost flow, and the precise guarantees are stated in \Cref{thm:generlize_min_cost_flow}. For lossy maximum flow, however, we can obtain exact feasible solutions using the same techniques as \cite{ds08}.
} %
as well as a $\wt{O}(m \sqrt{n})$ time algorithm, obtained by combining the algorithm of \cite{ds08} with techniques from \cite{ls14}.

The central goal of this paper is to close this gap. We consider the problem of designing faster algorithms for solving the generalized flow problem and more broadly, developing new algorithmic tools for reasoning about linear programs with at most two variables per inequality. Our main result is a nearly linear time algorithm for solving generalized maximum flow and generalized minimum cost flow to high accuracy on sufficiently dense graphs. To obtain this result, we provide new tools to reason about the linear algebraic structure of matrices associated with generalized flows and new dynamic data structures for maintaining information about them. We then apply these data structures in the interior point method (IPM) optimization framework of \cite{bll+21}. Beyond faster running times, we hope that the spectral graph theory tools we develop for generalized flows will have broader applications and facilitate faster algorithms for a wider range of problems.%

\subsection{Our Results}

Let $\mb_G \in \R^{E \times V}$ denote the edge-incidence matrix of the lossy graph $G = (V,E,\gamma)$, where for every edge $e = (a, b)$, the corresponding row $e$ of $\mb_G$ has $\gamma_{e}$ on entry $b$ and $-1$ on entry $a$, and therefore $(\mb_G^{\top} f)_v = \mathrm{im}_G(f)_v$ for all $v \in V$. Let $\mb_{G \backslash \{s,t\}} \in \R^{E \times V\backslash \{s,t\}}$ denote the submatrix of $\mb_G$ obtained by deleting the columns corresponding to vertices $s$ and $t$. The generalized maximum flow problem can be formulated as the following linear program:%
\[
\min_{\substack{\mb_{G \backslash \{s,t\}}^{\top} f = 0 \\ \allzero \leq f \leq u}} (\mb_G^{\top} f)_t.
\]
The generalized minimum cost flow problem can be formulated as the following linear program:
\[
\min_{\substack{\mb_G^{\top} f = d \\ \allzero \leq f \leq u}} c^{\top} f.
\]

To obtain our generalized flow algorithms, we consider the more general problem of solving the following linear program where every row of $\ma \in \R^{m \times n}$ has at most two non-zero entries: %
\begin{equation}\label{eq:LP_two_sided_constraints}
\min_{\substack{\ma^{\top} x=b\\
        \ell \le x \le u
    }
}c^{\top}x.
\end{equation}
This problem is the dual of a well-studied problem called the two variable per inequality (2VPI) linear program \cite{m83,h04}.
Additionally, it clearly encompasses the generalized maximum flow and generalized minimum cost flow problems.\footnote{We write $\ma^{\top} x = b$ (instead of $\ma x = b$) in \eqref{eq:LP_two_sided_constraints} so that generalized flow problems are special cases with $\ma = \mb_G$.}

We refer to the linear program \eqref{eq:LP_two_sided_constraints} as a \emph{two-sparse LP}. Additionally, we define matrices with at most two non-zero entries per row as \emph{two-sparse matrices} and measure their magnitude as follows. 
\begin{definition}[Two-sparse matrix and its magnitude parameter]\label{def:W_A}
We say a matrix $\ma$ is a \emph{two-sparse matrix} if every row of $\ma$ has at most two non-zero entries. For any matrix $\ma$, we define
\[
W_{\ma} \defeq \max_{i,j:\ma_{i,j} \neq 0}\left(\max(|\ma_{i,j}|, \frac{1}{|\ma_{i,j}|})\right)\,.
\]
\end{definition}

Our main result is that we can solve two-sparse LPs to $\delta$-accuracy %
in $\wt{O}((m + n^{1.5})\cdot \poly\log{(W/\delta)})$ time, where $\delta$-accuracy is defined in the theorem below. %
As a corollary, we can also solve the generalized maxflow and generalized min-cost flow problems to $\delta$-accuracy in $\wt{O}((m + n^{1.5})\cdot \poly\log{(W/\delta)})$ time, where $W$ is an upper bound on the magnitude of any integer used to describe the problem instance, and $\delta$ is an additive error parameter.

\begin{theorem}[Two-sparse LP]\label{thm:two_sparse_LP}
Let $\ma\in\R^{m\times n}$ be a two-sparse matrix, $c,\ell,u\in\R^{m}$, and $b\in\R^{n}$. Let $W \defeq \max\left(W_{\ma}, \|c\|_{\infty},\|b\|_{\infty},\|u\|_{\infty},\|\ell\|_{\infty},\frac{\max_{i}(u_{i}-\ell_{i})}{\min_{i}(u_{i}-\ell_{i})}\right)$. Assume that there is a point $x$ satisfying $\ma^{\top} x = b$ and $\ell \leq x \leq u$. There exists an algorithm that given any such $\ma$ and any $\delta>0$, runs in time $\wt{O}( (m + n^{1.5}) \cdot \poly\log(\frac{W}{\delta}))$ and w.h.p.\footnote{We use ``w.h.p.''~as an abbreviation for ``with high probability'', which means for any arbitrarily large constant specified in advance, the event happens with probability at least $1-1/m^c$, where $m$ is the input size. %
}
outputs a vector $x^{(\final)}$ satisfying %
\begin{align*}
\|\ma^{\top} x^{(\final)}-b\|_{\infty}\le\delta\enspace\text{ and }\enspace\ell \le x^{(\final)} \le u \enspace\text{ and }\enspace c^{\top}x^{(\final)}\le\min_{\substack{\ma^{\top} x=b\\
        \ell \le x \le u
    }
}c^{\top}x+\delta.
\end{align*}
\end{theorem}
Applying \Cref{thm:two_sparse_LP}, we obtain the following results for the generalized maxflow  (\Cref{thm:generlize_maxflow}) and generalized min-cost flow (\Cref{thm:generlize_min_cost_flow}) problems. 
In the special case of lossy maxflow, we further apply the technique of \cite{ds08} to obtain an exactly feasible flow.

\begin{theorem}[Generalized maxflow]\label{thm:generlize_maxflow}
    There exists an algorithm that, given any lossy graph $G = (V,E,\gamma)$, source $s \in V$, sink $t \in V$, capacities $u\in\R_{\geq 0}^{E}$, and $\delta>0$, for which there exists $f \in \R_{\geq 0}^E$ such that $f \leq u$ and $\mathrm{im}_{G}(f)_v = 0$ for all $v\in V \backslash \{s,t\}$, runs in time 
    \[
    \wt{O}\Big( (m + n^{1.5}) \cdot \poly\log(W/\delta)\Big)
    \text{ where }
    W \defeq \max\left(\|\gamma\|_{\infty}, \|\gamma^{-1}\|_{\infty}, \|u\|_{\infty},\frac{\max_e u_e}{\min_e u_e}\right)
    \]
    and w.h.p.~outputs $f \in \R_{\geq 0}^E$ satisfying $f \leq u$, 
\begin{align*}
\left|\mathrm{im}_{G}(f)_v \right| \leq \delta ~ \forall v\in V \backslash \{s,t\}, \enspace\text{ and }\enspace \mathrm{im}_{G}(f)_t \le\min_{\substack{\mathrm{im}_{G}(f')_v = 0\, \forall v\in V\backslash \{s,t\} \\
        \allzero \le f'_{e}\le u_{e}\forall e \in E
    }
}\mathrm{im}_{G}(f')_t + \delta.
\end{align*}
In the special case of lossy maxflow (i.e., when $\gamma \leq \allone_V$), the algorithm outputs an exactly feasible flow, i.e., it has the additional property that $\mathrm{im}_{G}(f)_v = 0 ~ \forall v\in V \backslash \{s,t\}$.
\end{theorem}

\begin{theorem}[Generalized min-cost flow]\label{thm:generlize_min_cost_flow}
    There exists an algorithm that, given any lossy graph $G = (V,E,\gamma)$, costs $c \in \R^E$, capacities $u\in\R_{\geq 0}^{E}$, demands $d\in\R^{V}$, and $\delta>0$, for which there exists $f \in \R_{\geq 0}^E$ such that $f \leq u$ and $\mathrm{im}_{G}(f)_v = d_v $ for all $v\in V$, runs in time 
    \[
    \wt{O}\Big( (m + n^{1.5}) \cdot \poly\log(W/\delta)\Big)
    \text{ where }
    W \defeq \max\left(\|\gamma\|_{\infty}, \|\gamma^{-1}\|_{\infty}, \|c\|_{\infty},\|u\|_{\infty},\frac{\max_e u_e}{\min_e u_e}\right)
    \]
    and w.h.p.~outputs $f \in \R_{\geq 0}^E$ satisfying $f \leq u$, 
\begin{align*}
\left|\mathrm{im}_{G}(f)_v - d_v\right| \leq \delta ~ \forall v\in V, \enspace\text{ and }\enspace c^{\top} f \le\min_{\substack{\mathrm{im}_{G}(f')_v = d_v \forall v\in V \\
        \allzero \le f'_{e}\le u_{e}\forall e \in E
    }
}c^{\top} f' + \delta.
\end{align*}
\end{theorem}

To obtain these results, our key technical contribution is a \emph{dynamic heavy hitter data structure} that can efficiently find all heavy entries $|\ma h|_i \geq \epsilon$ for any two-sparse matrix $\ma$ and any query vector $h$. We apply a reduction from \cite{bln+20,bll+21} to show that to prove \Cref{thm:two_sparse_LP}  it suffices to design this data structure (along with additional data structures we provide).
Our heavy hitter data structure relies on spectral graph theory tools that we develop to analyze the lossy Laplacian, which is an analog of the standard Laplacian that we introduce for lossy graphs. %
We believe both the data structures and these spectral results may be of independent interest.

\subsection{Related Work}

\paragraph{Generalized flow.} 
The generalized maxflow and generalized min-cost flow problems are well-studied and both combinatorial and continuous optimization based algorithms have been developed for them. We provide a brief summary of known weakly polynomial time algorithms for solving these problems in \Cref{tab:results}. For further references, we refer readers to the discussions in these papers and Chapter~15 of \cite{amo93}. Prior to our result, the state-of-the-art algorithm for solving generalized flow was given by \cite{ds08}, which can be improved to $\wt{O}(m \sqrt{n})$ time by combining techniques from \cite{ls14}. %
A key technical contribution of \cite{ds08} is an M-matrix scaling lemma that reduces solving linear systems in M-matrices to solving Laplacian linear systems, allowing the use of Laplacian solvers of \cite{st04} to implement each IPM iteration in nearly linear time. 
We refer readers to \cite{ajss19} for further discussion and improvements to solving M-matrices. %

\begin{table}[ht!]
    \centering
    \begin{tabular}{|c|c|c|c|c|}
        \hline
        {\bf Year} & {\bf Authors} & {\bf Problem} & {\bf Exact or approx}  & {\bf Time} \\
        \hlineB{3}
        1988    & Goldberg, Plotkin, Tardos \cite{gpt88} & Max & Exact & $m^2n^2$\\
        \hline
        1989    & Vaidya~\cite{v89} & Min-Cost & Exact & $m^{1.5}n^2$\\
        \hline
        1997   & Goldfarb, Jin, Orlin \cite{gjo97}& Max & Exact & $m^3$\\
        \hline
        2002   & Wayne~\cite{w02}& Min-Cost & Exact & $m^2n^3$\\
        \hline
        2002   & Fleischer, Wayne~\cite{fw02} & Min-Cost & $(1+\delta)$-mult.~approx & $m^2\delta^{-2}$\\
        \hline
        2008   & Daitch, Spielman~\cite{ds08} & Min-Cost & $\delta$-additive error & $m^{1.5}$   \\
        \hline
        2014   & Lee, Sidford~\cite{ls14} & Min-Cost & $\delta$-additive error & $m\sqrt{n}$\\
        \hline
        2025   & Our Work & Min-Cost & $\delta$-additive error & $m + n^{1.5}$\\
        \hline
    \end{tabular}
    \caption{A summary of weakly polynomial %
    algorithms for solving the generalized flow problems. All time complexities hide $n^{o(1)}$ and $\poly\log(W)$ factors, and the time complexities for approximate algorithms hide $\polylog(\delta^{-1})$ factors, where $W$ and $\delta$ are defined as in \Cref{thm:generlize_min_cost_flow}.}
    \label{tab:results}.
\end{table}

\paragraph{Strongly polynomial algorithms for solving 2VPI.}
Recently,  \cite{dknov24} provided the first strongly polynomial algorithm for solving generalized min-cost flow, %
and hence for linear programs with two variables per inequality (2VPI). Prior work had extensively studied and provided strongly polynomial algorithms for solving the primal and dual feasibility variants of the problem  \cite{m83,cm91,hn94,v14,k22}. For further references we refer the readers to the discussion in \cite{dknov24}.
Note that this line of work often considers the LP formulation of 2VPI with one-sided constraints $x \geq 0$ instead of two-sided constraints $\ell \leq x \leq u$. %
Such a restriction is without loss of generality because the 2VPI formulation with two-sided constraints can be reduced to a formulation with one-sided constraints at the cost of increasing the number of constraints from $n$ to $m+n$. Since, in this paper we seek running times which depend on both $m$ and the ratio $m/n$, we explicitly consider the 2VPI formulation with two-sided constraints, as in \eqref{eq:LP_two_sided_constraints}.

\paragraph{Maximum flow and minimum cost flow.}
Before the breakthrough almost linear time maxflow and min-cost flow algorithm of \cite{cklpps23}, multiple advances had been made to improve the time complexity of IPM algorithms for solving these flow problems %
including \cite{ckmst11,m13,ls14,m16,cmsv17,ls20,kls20,amv20,amv21,bgjllps22}. There is also a recent, related line of work on improving combinatorial algorithms for solving flow problems  \cite{bbst24,ck24a}.

\subsection{Organization}

In \Cref{sec:technical_overview}, we give an overview of our approach. %
In \Cref{sec:prelim} we cover general notation and known technical tools.
In \Cref{sec:v_uniformity} and \Cref{sec:spectral_approximation}, we prove some spectral theorems about our lossy Laplacian. In \Cref{sec:heavy_hitters_balanced_expanders} and \Cref{sec:general_heavy_hitters}, we use these spectral theorems to construct heavy hitter and sampler algorithms on lossy graph incidence matrices and 2-sparse matrices. %
Finally, in \Cref{sec:IPM_using_data_structs}, we show how the heavy hitter and sampler algorithms fit into an IPM that can solve the generalized flow and 2-sparse LP problems.

\section{Technical Overview}\label{sec:technical_overview}

In this section we give an overview of our approach. We solve two-sparse LPs by applying the  general interior point method (IPM) framework of \cite{bln+20,bll+21} (see \Cref{thm:path_following}), which applies to linear programs of the form \eqref{eq:LP_two_sided_constraints} with any constraint matrix $\ma$. The framework essentially reduces solving the LP to designing three data structures for the constraint matrix $\ma$: (1) a \emph{heavy hitter data structure} that finds the entries of $\ma h$ with large absolute value, %
(2) a \emph{sampler data structure} that samples entries according to a distribution defined by $\ma h$, and (3) an \emph{inverse maintenance data structure} that solves linear equations in $\ma^{\top} \mv \ma$, where $\mv$ is a non-negative diagonal matrix. The IPM has $\wt{O}(\sqrt{n} \log(\frac{W}{\delta}))$ iterations and runs in time which depends on the time of implementing these data structures.

We obtain our results by providing efficient data structures for each of these three data structure problems for two-sparse $\ma$. To motivate our approach and explain our results, in this overview we focus primarily on designing the heavy hitter data structure. This task captures the main difficulty of the problem and illustrates our key ideas for designing our data structures. 
More specifically, the heavy hitter problem is to find the 
 \emph{heavy hitters} of $\ma h$, the  entries of $\ma h$ that are greater than $\epsilon$ in absolute value, where $\ma$ is a two-sparse matrix that undergoes a bounded number of updates. Given any $\ma$ and $\epsilon$, it is easy to see there are at most $\epsilon^{-2}\|\ma h\|_2^2$ heavy hitters, and this bound is our target time complexity for answering heavy hitter queries. Our main theorem regarding a heavy hitter data structures for two-sparse matrices is given informally below.

\begin{theorem}[Heavy hitter for two-sparse matrices, informal version of \Cref{thm:heavy_hitter_two_sparse_general}]\label{thm:heavy_hitter_two_sparse_general_informal}
There is a data structure \textsc{HeavyHitter} that can be initialized with any two-sparse matrix $\ma \in \R^{m \times n}$ in $\wt{O}(m) \cdot \polylog(\frac{W}{\delta})$ time, and supports inserting or deleting any row to $\ma$ in $\wt{O}(1) \cdot \polylog(\frac{W}{\delta})$ time, and querying for all heavy hitters $i$ with $|(\ma h)_i| \ge \epsilon$ in $\wt{O}\left( \epsilon^{-2} \| \ma h \|_2^2 + n\right) \cdot \polylog(\frac{W}{\delta})$ time. %
\end{theorem}

\Cref{thm:heavy_hitter_two_sparse_general} allows us to support all heavy hitter queries of the IPM in $\widetilde{O}(m+n^{1.5}) \cdot \polylog(\frac{W}{\delta})$ time. %
When we apply the associated algorithm in each iteration of the IPM, the IPM guarantees that the sum of these $\| \ma h \|_2^2$ across all iterations is $\widetilde{O}(m) \cdot \polylog(\frac{W}{\delta})$. If we support each heavy hitter query in $\wt{O}\left(\epsilon^{-2}\|\ma h\|_2^2 + n\right) \cdot \polylog(\frac{W}{\delta})$ time and each row update of $\ma$ in $\wt{O}(1) \cdot \polylog(\frac{W}{\delta})$ time, %
then we can obtain our desired runtime of $\widetilde{O}(m + n^{1.5}) \cdot \polylog(\frac{W}{\delta})$ across all iterations.

A key ingredient of our proof of the \Cref{thm:heavy_hitter_two_sparse_general} is a heavy hitter data structure for the special case of incidence matrices $\mb_G$ of what we call \emph{balanced lossy expanders}. These are unweighted lossy graphs $G$ whose associated non-lossy graphs $\gbar$ are expanders, and have flow multipliers $(1- \beta) \allone_{E} \leq \gamma \leq \allone_{E}$ for a sufficiently small $\beta$. %
These graph can be viewed as a natural analogue of unweighted expanders for standard graphs. Perhaps the main technical contribution of this work is developing sufficient spectral graph theory tools, to obtain this heavy hitter data structure.

In \Cref{sec:tech_expander} we first review how \cite{bln+20} solves the heavy hitter problem for graphs, and introduce an approach based on \emph{balanced lossy expanders}. %
In \Cref{sec:tech_heavy_hitter_balanced_expander}, we discuss our approach for solving the heavy hitter problem on balanced lossy expanders. We first discuss the error requirements for our approaches to work, and then present the key spectral results. In \Cref{sec:tech_heavy_hitter_two_sparse} we present a proof sketch for \Cref{thm:heavy_hitter_two_sparse_general_informal} by using heavy hitters on balanced lossy expanders. Finally in \Cref{sec:tech_combine} we introduce our other two data structures, the sampler and inverse maintenance, which together with the heavy hitter yield our $\wt{O}(m + n^{1.5}) \cdot \polylog(\frac{W}{\delta})$ time algorithm for solving two-sparse LPs  (\Cref{thm:two_sparse_LP}). %

\subsection{Expander Decomposition Based Approaches}\label{sec:tech_expander}
\paragraph{Expander Decompositions for Graphs and Prior Approach.}
We first discuss the strategy used in \cite{bln+20,bll+21} for solving the heavy hitter problem on the edge-vertex incidence matrix $\mb_{\gbar} \in \R^{E \times V}$ of a graph $\gbar$. %
They observe that the heavy hitter problem is decomposable in that if there exists an edge decomposition $\gbar = \gbar_1 \cup \gbar_2$, to find the heavy hitters of $\mb_{\gbar} h = \begin{bmatrix} \mb_{\gbar_1} \\ \mb_{\gbar_2}
\end{bmatrix} h$, it suffices to find all heavy hitters of $\mb_{\gbar_1} h$ and $\mb_{\gbar_2} h$ separately. %
They use standard tools \cite{sw21} to decompose $\gbar$ into expanders with $1/\polylogn$ expansion and solve the problem on the expanders. %

To solve the heavy hitter problem when $\gbar$ is an expander, \cite{bln+20} uses  that the Laplacian $\mL_{\gbar} = \mb_{\gbar}^{\top} \mb_{\gbar}$ can be spectrally approximated by a diagonal matrix minus a rank-1 matrix. For simplicity of presentation, in this section we assume the graph is unweighted %
and $d$-regular, in which case this spectral approximation is %
\begin{equation}\label{eq:Laplacian_approx}
\mL_{\gbar} \approx_{\polylogn} \md - \frac{d}{n}\allone_V \allone_V^\top,
\end{equation}
where $\md = d \cdot \mI$ is the degree matrix of $\gbar$ and $\allone_V$ is the all-ones vector.  This spectral approximation implies that for any vector $h$ orthogonal to $\allone_V$, $h^\top \md h \approx_{\polylogn} h^\top \mL_{\gbar} h = \|\mb_{\gbar} h\|_2^2$. 

Leveraging this spectral approximation fact, \cite{bln+20} provides the following simple algorithm that finds the heavy hitters in $\widetilde{O}(\|\mb_{\gbar} h\|_2^2\epsilon^{-2})$ time. %
Given any query vector $h$, \cite{bln+20} first subtracts a constant from all entries of $h$ to obtain a shifted vector $h'$ that is orthogonal to the all-ones vector. Such constant shifts do not affect the heavy hitter problem, since $\mb_{\gbar} \allone_V = 0$. (More broadly, $\ker(\mb_{\gbar}) = \mathrm{span}(\allone_V)$.) They then solve the heavy hitter problem with the shifted vector $h'$ by using that a row of  $\mb_{\gbar}$ , e.g.,  $\indicVec{u} - \indicVec{v}$, can only be a heavy hitter if either $h'_u$ or $h'_v$ has absolute value greater than $\epsilon/2$. To find all heavy hitters, the algorithm of \cite{bln+20} checks the adjacent edges of all vertices for which $|h'_v| \geq \epsilon/2$. Checking the adjacent edges of one vertex can be done in $O(d)$ time, and in total this step takes $O(\sum_v d (h'_v)^2/\epsilon^{2}) = O(\epsilon^{-2} \cdot (h')^\top \md h') = \widetilde{O}(\|\mb_{\gbar} h\|_2^2\epsilon^{-2})$ time, as desired (where in the last equality we used \eqref{eq:Laplacian_approx}).

\paragraph{Expander Decompositions for Lossy Graphs.} Now let's consider applying an analogous approach to %
unweighted lossy graphs. With a slight notational overload, %
we consider the equivalent formulation where instead of working with loss factors $\gamma$ we work with %
\emph{flow multipliers} $\eta$, where $\eta \defeq \gamma^{-1} \geq \allone_{E}$, and denote the lossy graph as $G = (V,E,\eta)$. %
We define $\mb_{G}\in\R^{E\times V}$ as the incidence matrix of $G$ which has  $\indicVec{b_{e}} - \eta_{e} \indicVec{a_{e}}$ as row $e = (a_e, b_e)$ of $\mb_{G}$ where $\indicVec{a_{e}}$ and $\indicVec{b_{e}}$ are indicator vectors with $1$ on $a_{e}$ and $b_{e}$ respectively. These two formulations are equivalent, and we choose to use flow multipliers $\eta_e \geq \allone_{E}$ since it makes it cleaner to describe balanced lossy graphs where $\eta_e \leq 1 + \balloss$.
For a lossy graph $G = (V,E,\eta)$, we call $\mL_G = \mb_G^\top \mb_G$ the \emph{lossy Laplacian} of $G$. We denote by $\gbar=(V,E)$ the {\em smoothed graph} of $G$, which is a graph with no flow multipliers and the same vertices and edges as $G$ and let $\mL_{\gbar} \defeq  \mb_{\gbar}^\top \mb_{\gbar}$ denote the Laplacian of $\gbar$.

We first observe that the heavy hitter problem for a lossy graph $\mb_{G}$ is still decomposable: Given any edge decomposition $G = G_1 \cup G_2$, it suffices to find all heavy hitters of $\mb_{G_1} h$ and $\mb_{G_2} h$ separately. Our first attempt is to compute an expander decomposition on the smoothed graph $\gbar$, and use the same decomposition on our lossy graph $G$. However, even if the underlying smoothed graph $\gbar$ is an expander, the corresponding lossy Laplacian $\mL_G$ does not necessarily have a large enough gap between the least and second least eigenvalues $\lambda_{1}(\mL_G)$ and $\lambda_2(\mL_G)$. For example, if the incidence matrix of $G$ is $\mb_G = \begin{pmatrix}
1 & -\frac{1}{x} & 0\\
0 & 1 & -\frac{1}{x}\\
-\frac{1}{x} &  0 & 1\\
\end{pmatrix}$, then its underlying smoothed graph $\gbar$ has $\lambda_{1}(\mL_{\gbar}) = 0$ and $\lambda_2(\mL_{\gbar}) = 3$. However, the lossy Laplacian $\mL_G$ has $\lambda_{1}(\mL_G)$ and $\lambda_2(\mL_G)$ that both tend to $1$ as $x$ goes to infinity, and consequently a gap between $\lambda_{1}(\mL_G)$ and $\lambda_2(\mL_G)$ that goes to 0.

To overcome this obstacle, our first spectral result for lossy graphs is the following lemma that shows that if we have the extra condition where the flow multipliers $\eta_e$ are all relatively close to one, then the lossy graph indeed has a large enough spectral gap.
\begin{lemma}[Informal version of \Cref{lem:eta_to_not,lem:N_approx_Ivvt}]\label{lem:N_approx_Ivvt_informal}
If lossy graph $G = (V,E,\eta)$ satisfies $\allone_E\leq \eta \leq (1 + \balloss) \allone_E$ where $\balloss \leq 1/\polylogn$, and the underlying smoothed graph $\gbar$ is a $d$-regular expander with conductance at least $20 \balloss$, then the least eigenvalue $\lambda_{\min}$ of $\mL_G$ and its corresponding unit eigenvector $v_{\min}$ satisfy %
\begin{equation}
\label{eq:informal_balanced_approx}
\mL_G \approx_{\polylogn}  \md - d (1-\lambda_{\min}) v_{\min} v_{\min}^{\top}.
\end{equation}
\end{lemma}
We use the term  \emph{balanced lossy expanders} to refer to  such unweighted lossy graphs that have both a small flow multiplier, as well as an underlying smoothed graph that is an expander. %
We show that any lossy graph can be decomposed into a collection of such balanced lossy expanders by partitioning its edges, applying a suitable vertex scaling, and performing expander decomposition to the underlying smoothed graph.
The total number of vertices across all subgraphs in this decomposition is $O(n \cdot \polylog(nW))$ %
(more formal statements can be found in \Cref{sec:general_heavy_hitters}, where edge partition and vertex scaling are shown in \Cref{sec:heavy_hitter_general}, and expander decomposition is shown in \Cref{sec:heavy_hitter_balanced}). 
As such, from now on we focus on balanced lossy expanders. 

Leveraging \eqref{eq:informal_balanced_approx}, if we could somehow explicitly compute the exact eigenvector $v_{\min}$ and $\mb_G v_{\min}$ in each iteration, with the latter given in sorted order, then the following natural generalization of the algorithm in \cite{bll+21,bln+20} would solve the heavy hitter problem for $\mb_G$.
Given a query vector $h$, the algorithm projects $h$ onto the direction of $v_{\min}$ rather than the all-ones vector as before, i.e., $h^{v_{\min}} = v_{\min}v_{\min}^\top h$, and the algorithm also computes the component orthogonal to $v_{\min}$, i.e., $h^{\perp v_{\min}} = (\mI -v_{\min}v_{\min}^\top)h$. The algorithm then finds the heavy hitters of these two vectors separately. In the direction perpendicular to $v_{\min}$, similar as before, the algorithm checks the adjacent edges of all vertices $i$ for which $|(h^{\perp v_{\min}})_i| \geq \epsilon/2$. This step takes time $O((h^{\perp v_{\min}})^{\top} \md h^{\perp v_{\min}})$, which can be bounded by the spectral approximation in \Cref{lem:N_approx_Ivvt_informal}: $(h^{\perp v_{\min}})^{\top} \md h^{\perp v_{\min}} \approx_{\polylogn} (h^{\perp v_{\min}})^\top \mL_G h^{\perp v_{\min}} = \|\mb_{G} h^{\perp v_{\min}}\|_2^2 \leq \|\mb_G h \|_2^2$. In the direction of $v_{\min}$, since an appropriately sorted $\mb_G v_{\min}$ has already been computed, finding the heavy hitters of $\mb_G h^{v_{\min}}$ is straightforward, as $\mb_G h^{v_{\min}}$ is simply a scalar multiple of $\mb_G v_{\min}$.%

\paragraph{Difficulties Generalizing.} Unfortunately, this approach encounters two immediate obstacles.

First, the smallest eigenvector $v_{\min}$ is not known. For lossy graphs, the smallest eigenvector is not necessarily the all-ones vector as it is for graphs. The eigenvector $v_{\min}$ can be computed to $\epsilon$ accuracy using the power method and lossy Laplacian solvers in $\wt{O}(m \log(1/\epsilon))$ time (see \Cref{thm:power_method}). However, this runtime is prohibitively expensive for each heavy hitter query. An alternative would be to efficiently maintain a coarser approximation of $v_{\min}$, however this would require care regarding the level of approximation. For example, it is no longer necessarily true that $\|\mb_{G} h^{\perp v}\|_2^2 \leq O(\|\mb_G h \|_2^2)$, when $v$ is only a coordinate-wise multiplicative $(1\pm 1/\polylog)$ approximation to the smallest eigenvector $v_{\min}$. 

Second, the vector $\mb_G v_{\min}$ is also not known. Even if $v_{\min}$ were given explicitly, computing $\mb_G v_{\min}$ would still take $O(m)$ time per iteration, as the eigenvector changes from iteration to iteration. %
One way of getting around this problem is to maintain an approximation to $\mb_G v_{\min}$ explicitly, but this would require some type of stability on $v_{\min}$ that allows us to either sparsely update our approximation $v$ in each iteration, or only update our entire approximation infrequently.

\subsection{Heavy Hitters on Balanced Lossy Expanders}\label{sec:tech_heavy_hitter_balanced_expander}
In this section we show how to overcome the obstacles mentioned above to solve the heavy hitter problem on balanced lossy expanders. %

\paragraph{Approximation Requirement of $v$.}%
We first address the accuracy required in our approximation $v$ of the smallest eigenvector $v_{\min}$ for our generalization of the algorithm in \cite{bln+20} to work. Recall that this algorithm finds the heavy hitters of $\mb_G \hpv$ and $\mb_G\hv$ for $h^{v} = vv^\top h$ and $h^{\perp v} = (\mI -vv^\top)h$. This algorithm is able to correctly output all heavy hitters of $\mb_G h$, because if $|(\mb_G h)_i| \geq \epsilon$, we must have that either $|(\mb_G h^{\perp v})_i| \geq \epsilon/2$ or $|(\mb_G h^{v})_i| \geq \epsilon/2$. So we focus on bounding the time complexity of this algorithm.

When finding the heavy hitters of $\mb_G h^{\perp v}$, by a similar argument as before, checking all neighbours of each vertex $i$ where $h^{\perp v}_i \geq \epsilon/2$ requires $O(({\hpv})^\top \md \hpv\epsilon^{-2})$ time. We therefore seek a $v$ for which $O(({\hpv})^\top \md \hpv\epsilon^{-2})$ is bounded by $\widetilde{O}(\|\mb_G h\|_2^2 \epsilon^{-2})$, the budget that we have for one heavy hitter query. 
In particular, we would like to find a unit vector $v$ satisfying 
\[
(h^{\perp v})^\top \md h^{\perp v}
\leq \polylogn \cdot h^\top  \mL_G h
\text{ for all }
h,
\]
or equivalently that

\begin{equation}\label{eq:tech_overview_spectral_approx}
    \md - dvv^\top \preceq {\polylogn}\cdot\mL_G.
\end{equation}

When $\lambda_2(\mL_G) \geq \frac{d}{\polylogn}$, we have that $v = v_{\min}(\mL_G)$ satisfies the requirement in \Cref{eq:tech_overview_spectral_approx}. %
Our next spectral result shows that perhaps surprisingly, \Cref{eq:tech_overview_spectral_approx} holds whenever $v$ has a reasonably small Rayleigh quotient $v^{\top} \mL_G v \leq \polylogn \cdot \lambda_{\min}(\mL_G)$.  %

\begin{theorem}[Spectral approximation from low Rayleigh quotient, informal version of \Cref{thm:spectral_sparsification}]\label{thm:spectral_sparsification_informal}
If $G$ is a balanced lossy expander and unit vector $v$ satisfies $v^\top \mL_G v \leq \polylogn \cdot \lambda_{\min}(\mL_G)$, then
\[
    \mL_G \approx_{\polylogn} \md - (1 - \lambda)dvv^\top
    \text{ for }
    \lambda = v^\top \mL_G v\,.
\]
\end{theorem}
\Cref{thm:spectral_sparsification_informal} shows that in the heavy hitter algorithm it suffices to use a unit vector $v$ whose Rayleigh quotient $v^\top \mL_G v$ is within a $\polylogn$ factor from the least eigenvalue $\lambda_{\min}$. Another equivalent statement of this requirement is that $v$ needs to be very aligned with $v_{\min}$, more precisely, $1-(v^{\top} v_{\min})^2 \leq \polylogn \cdot \lambda_{\min}$. %
The relationship of these two requirements are shown in \Cref{lem:bound_on_v_topv} and \Cref{lem:bound_on_v_topv_other_direction}.

\Cref{thm:spectral_sparsification_informal} suggests a natural approach to maintain the approximate eigenvector $v$: simply use the same $v$ until the Rayleigh quotient $v^\top \mL_G v$ is no longer smaller than $\polylogn \cdot \lambda_{\min}$. Note that in this approach it is also easy to find the heavy hitters in $\mb_G \hv$ because we can explicitly maintain $\mb_G v$ together with $v$. The number of heavy hitters in $\mb_G\hv$, which also determines the running time of this step, is $\|\mb_G\hv\|_2^2\epsilon^{-2}$, and it can be bounded by our budget $\wt{O}(\|\mb_G h\|_2^2 \epsilon^{-2})$ using the triangle inequality $\|\mb_G\hv\|_2^2 \leq O(\|\mb_Gh\|_2^2 + \|\mb_G\hpv\|_2^2)$, and \Cref{thm:spectral_sparsification_informal} that $\|\mb_G \hpv\|_2^2 \leq \wt{O}\big((\hpv)^{\top} (\md - (1-\lambda) d v v^{\top}) \hpv\big) = \wt{O}\big(h^{\top} (\md - d v v^{\top}) h\big) \leq \wt{O}(\|\mb_G h\|_2^2)$.

\paragraph{Structure of Eigenvector $v_{\min}$.}%

To effectively follow this approach, we need to handle that whenever a sufficient number of vertex deletions occur our algorithms need to perform a type of re-normalization. This re-normalization can potentially increase the Rayleigh quotient $v^{\top} \mL_G v$ for the re-normalized approximate smallest eigenvector $v$. We show however that this increase is fairly limited and to do this we prove another another structural result about $v_{\min}$ of potential independent interest. In particular, we show that $v_{\min}$  of a balanced lossy expanders is coordinate-wise close to the all-ones vector.

\begin{theorem}[Uniformity of $v_{\min}$, informal version of \Cref{thm:unifomity_v}]\label{thm:unifomity_v_informal}
If $G$ is a balanced lossy expander, then
\begin{align*}
     \frac{1}{\sqrt{n}} \cdot \left(1-1/\polylogn\right) \leq (v_{\min}(\mL_G))_i \leq \frac{1}{\sqrt{n}} \cdot (1+1/\polylogn)
     \text{ for all } 
     i \in [n]\,.
\end{align*}
\end{theorem}

\Cref{thm:unifomity_v_informal} shows that the least eigenvector $v_{\min}$ of a balanced lossy expander behaves similarly to that of a graph, which is always the all-ones vector. One might wonder if this coordinate-wise approximation to $\frac{\allone_V}{\sqrt{n}}$ would allow us to apply the same heavy hitter algorithm used for graphs to balanced lossy expanders. Unfortunately, this is not the case; \Cref{thm:unifomity_v_informal} only guarantees that $1 - (v_{\min}^\top \frac{\allone_V}{\sqrt{n}})^2 \leq 1/\polylogn$, whereas our algorithm requires an approximate vector $v$ satisfying  $1 - (v_{\min}^\top v)^2 \leq \polylogn \cdot \lambda_{\min}$ (as in \Cref{thm:spectral_sparsification_informal}).

\paragraph{Heavy Hitters on Balanced Lossy Expanders.} 
We now have all the spectral results needed to design and analyze our heavy hitter algorithm on balanced lossy expanders. %
Our algorithm uses a reweighting technique together with dynamic expander decomposition to maintain balanced lossy expanders as the graph evolves. This decomposition algorithm routinely deletes edges and vertices, and our goal is to support heavy hitter queries of $\mb_G$ under such deletions. Our heavy hitter algorithm maintains an approximate eigenvector $v \in \R^n$ of the Laplacian $\mL_G$ throughout all updates, and we begin by explaining how to maintain this vector $v$.

Instead of finding all heavy entries of $\mb_G h$, we introduce a modified problem $\mb'$ to ensure that the smallest eigenvalue of the associated Laplacian $\mL' \defeq \mb'^{\top} \mb'$ is bounded away from zero. Concretely, we define
\[
\mb' \defeq \begin{bmatrix}\mb_G \\ \epsad \mI\end{bmatrix}, \text{ where }\epsad \defeq \poly(\delta/(mW)),
\]
 and consider the heavy hitter problem of $\mb' h$. The key advantage of this modified problem is that the modified Laplacian $\mL'$ now satisfies $\lambda_{\min}(\mL') \geq \epsad$. This lower bound guarantees that $\lambda_{\min}(\mL')$ can decrease by a constant factor for at most $O(\log(\epsad^{-1}))$ times, which is crucial for bounding the number of times to recompute the approximate eigenvector $v$. Since $\epsad$ is chosen sufficiently small, replacing $\mb_G$ with $\mb'$ only adds an additional $\delta/10$ error to the final LP solution. 
 
Our algorithm only recomputes $v$ when the condition, $v^\top \mL' v \leq \polylogn \cdot \lambda_{\min}(\mL')$ from \Cref{thm:spectral_sparsification_informal}, no longer holds. 
Let $\mL'^{\old}$ denote the Laplacian at the last time $v$ was updated, where $v = v_{\min}(\mL'^{\old})$, and note that $v^\top \mL'^{\old} v = \lambda_{\min}(\mL'^{\old})$. After each edge deletion, the Rayleigh quotient can only decrease. After each vertex deletion, our algorithm restricts $v$ onto the new vertex support, and re-normalizes it so that it remains a unit vector. This re-normalization does not increase the Rayleigh quotient by more than a constant factor as a result of \Cref{thm:unifomity_v}. Our algorithm restarts whenever more than half of the edges are deleted, and \Cref{thm:unifomity_v} ensures that $v$ is close to the all-ones vector, so no set of deleted vertices can be too large in $v$, i.e., $\sum_{i \in S} v_i^2$ of the deleted vertices $S$ is bounded by a constant. %
 More formally, we show that $v^\top \mL' v \leq 3 \cdot v^\top \mL'^{\old} v$. 
 Therefore, whenever the algorithm recomputes $v$, we must have
\[
\lambda_{\min}(\mL') < \frac{1}{\polylogn} \cdot v^{\top} \mL' v \leq \frac{3}{\polylogn} \cdot v^\top \mL'^{\old} v = \frac{3}{\polylogn} \cdot \lambda_{\min}(\mL'^{\old}). 
\]
This inequality can occur at most $O(\log(\epsad^{-1}))$ times since $\lambda_{\min}(\mL') \geq \epsad$.

One last remaining issue is that it is computationally expensive to directly compute the least eigenvalue in order to check if $v^\top \mL' v \leq \polylogn \cdot \lambda_{\min}(\mL')$ still holds. We instead test whether the norm $\|\md' \hpv\|_2$ exceeds our budget $\|\mb' \hpv\|_2$ by using a Johnson-Lindenstrauss sketch $\mj\mb'$. By \Cref{thm:spectral_sparsification_informal}, if $\|\md' \hpv\|_2 > \polylogn \cdot \|\mb' \hpv\|_2$, then the spectral approximation $\mL_G \approx_{\polylogn} \md - (1 - \lambda)dvv^\top$ no longer holds, and so the Rayleigh quotient must have decreased significantly.

Having established how to maintain the approximate eigenvector $v$, next we explain how to use it in the heavy hitter queries. Given a query vector $h$, we decompose it as $\hv = v v^{\top} h$ and $\hpv = (\mI - v v^{\top}) h$, and compute the heavy hitters of both $\mb_G \hpv$ and $\mb_G \hv$.
\begin{itemize}
\item For $\mb_G \hpv$, we check all neighbours of each vertex $i$ where $\hpv_i \geq \epsilon/2$, and this step takes $O(\|\md \hpv\|_2^2 \epsilon^{-2}) \leq \wt{O}(\|\mb_G h\|_2^2 \epsilon^{-2})$ time. 
\item For $\mb_G \hv$, we compute $\mb_G v$ explicitly whenever $v$ is updated and use this precomputed vector to find the heavy entries of $\mb_G \hv$. This step takes $O(\|\mb_G \hv\|_2^2 \epsilon^{-2}) \leq \wt{O}(\|\mb_G h\|_2^2 \epsilon^{-2})$ time.
\end{itemize}
This covers the central ideas of our heavy hitter data structure for balanced lossy expanders. Our spectral results and full algorithm provided in \Cref{sec:v_uniformity}, \ref{sec:spectral_approximation}, and \ref{sec:heavy_hitters_balanced_expanders} are slightly more complex in how they handle arbitrary degrees, and work with the normalized Laplacian instead of the Laplacian.

\subsection{Heavy Hitters on Two-Sparse Matrices}\label{sec:tech_heavy_hitter_two_sparse}

In this section, we discuss in greater detail how to build upon our heavy hitter data structure for balanced lossy Laplacians to obtain a heavy hitter data structure for general two-sparse matrices, using a sequence of reductions shown in \Cref{fig:chain_reductions}.

We first reduce the problem of heavy hitters on general two-sparse matrices to the problem of heavy hitters on two-sparse matrices with a positive and negative entry in each row, i.e., incidence matrix of general weighted lossy graphs. This reduction is standard and is similar to the one in \cite{h04}; the formal result is given in \Cref{lem:reduction_two_sparse_to_lossy_graph}. 

Next we reduce the heavy hitter problem on general weighted lossy graphs to that on balanced lossy graphs, which are unweighted lossy graphs with bounded flow multipliers (\Cref{sec:heavy_hitter_general}). To do this, we first partition the edges of the weighted lossy graph so that, in each subgraph, all edge weights are within a constant factor of each other. Each subgraph can then be treated as an unweighted lossy graph, with a single scalar scaling factor applied later. We further decompose each unweighted subgraph into bipartite components with all edges oriented from left to right, bucket the edges based on their flow multipliers, and rescale the vertices so that all flow multipliers lie in the range $\allone_{E} \leq \eta \leq (1+ \beta) \allone_{E}$ for a sufficiently small $\beta$. 

Finally, through the use of the dynamic expander maintenance data structures in \cite{sw21}, we reduce the heavy hitter problem on balanced lossy graphs, to the heavy hitter problem on balanced lossy expanders (\Cref{sec:heavy_hitter_balanced}). The heavy hitter algorithm on balanced lossy expanders requires that the vertex degrees remain approximately regular. To maintain this property, we delete vertices once their degrees fall below a $\polylogn$ fraction of their initial degrees. We show that this degree preservation process, combined with the expander pruning, results in at most $\polylogn$ overhead.

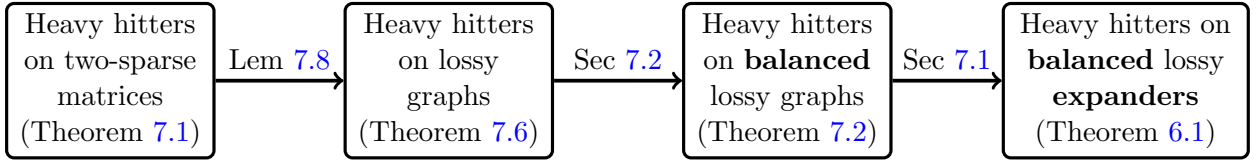
\begin{figure}[!ht]
\centering
\begin{tikzpicture}
    \node (A) [draw, rectangle,rounded corners=1mm, minimum width=2cm, very thick, node distance=4.5cm, minimum height=1cm,text width=2.5cm,align=center] {Heavy hitters on two-sparse matrices\\ (\Cref{thm:heavy_hitter_two_sparse_general})};
    \node (B) [draw, rectangle,rounded corners=1mm,right of=A, minimum width=2cm, very thick, node distance=4.5cm, minimum height=1cm,text width=2.5cm,align=center] {Heavy hitters on lossy graphs \\ (\Cref{thm:heavy_hitter_general})};
    \node (C) [draw, rectangle,rounded corners=1mm, right of=B,very thick,node distance=4.5cm, text width=2.5cm,align=center] {Heavy hitters on \textbf{balanced} lossy graphs \\(\Cref{thm:heavy_hitter_balanced})};
    \node (D) [draw, rectangle,rounded corners=1mm, right of=C,very thick, node distance=4.5cm, text width=3cm,align=center] {Heavy hitters on \\\textbf{balanced} lossy \textbf{expanders} \\(\Cref{thm:heavy_hitter_expander})};

    \draw[->, line width=0.5mm] (A) -- (B) node[midway, above] {Lem~\ref{lem:reduction_two_sparse_to_lossy_graph}};
    \draw[->, line width=0.5mm] (B) -- (C) node[midway, above] {Sec~\ref{sec:heavy_hitter_general}};
    \draw[->, line width=0.5mm] (C) -- (D) node[midway, above] {Sec~\ref{sec:heavy_hitter_balanced}};
\end{tikzpicture}
\caption{An illustration of the chain of reductions for the heavy hitter data structures.}\label{fig:chain_reductions}
\end{figure}

\subsection{Other Data Structures}\label{sec:tech_combine}

We conclude the overview by describing how we implement the other two data structures, i.e., the sampler and inverse maintenance, required by the IPM framework of \cite{bll+21}.  (The formal definition of inverse maintenance and our result can be found in \Cref{def:inverse_maintenance} and \Cref{thm:inverse_maintenance}.)

\paragraph{Sampler.} This data structure for a two-sparse matrix $\ma \in \R^{m\times n}$ is required to support, for any query vector $h$, sampling each $e \in [m]$ independently with probability
\[
p_e \geq \frac{m}{\sqrt{n}} \frac{(\ma h)_e^2}{\|\ma h\|_2^2}
\]
in $\wt{O}(\frac{m}{\sqrt{n}}) \cdot \polylog(\frac{W}{\delta})$ time. %
Similar to prior results \cite{bln+20,bll+21}, we implement the sampler using the same approach as the heavy hitter data structure.  
We support the sampling operation using a similar chain of reductions as the heavy hitter operation, so that it reduces to sampling an edge of a balanced lossy expander $G$ with probability $p_e' \geq C \cdot (\mb_{G} h)_e^2$, and estimating the $\ell_2$ norm $\|\mb_{G} h\|_2$. 
Using the approximate eigenvector $v$ that is maintained as in \Cref{sec:tech_heavy_hitter_balanced_expander} that satisfies $(\md - (1-\lambda) v v^{\top}) \approx_{\polylogn} \mL_{G}$, we again decompose the query vector $h = \hv + \hpv$ where $\hv$ is in the direction of $v$, and $\hpv$ is orthogonal to $v$. 

To sample an edge $e = (i,j)$, our spectral approximation results ensures that 
\[
(\mb_{G} h)_e^2 \approx (\mb_{G} \hv)_e^2 + \frac{(\hpv_i)^2}{d_i} + \frac{(\hpv_j)^2}{d_j}.
\]
Since the algorithm already maintains $\mb_{G} v$, we can sample from the first term $(\mb_{G} \hv)_e^2$ efficiently. For the second term $\frac{(\hpv_i)^2}{d_i} + \frac{(\hpv_j)^2}{d_j}$, we sample uniformly each incident edge of every vertex $i$ with probability $\frac{(\hpv_i)^2}{d_i}$. The total time is proportional to the expected number of edges sampled.

To estimate the $\ell_2$ norm, the spectral approximation also implies that %
\[
\|\mb_{G} h\|_2^2 \approx \|\mb_{G} \hv\|_2^2 + (\hpv)^{\top} \md \hpv\,.
\]
We can compute the first term using the maintained value of $\|\mb_{G} v\|_2^2$, and compute the second term in $O(n)$ time.

The formal definition of the sampler and our result regarding an efficient sampler can be found in \Cref{def:heavy_sampler} and \Cref{cor:sampler}.

\paragraph{Inverse Maintenance.} This data structure supports approximately solving a linear system with $\ma^{\top} \mv \ma$ in $\wt{O}(n) \cdot \polylog(\frac{W}{\delta})$ time. We follow the same idea as \cite{bll+21} that uses leverage score sampling to maintain a $\wt{O}(n)$-sparse diagonal matrix $\ms$ such that $\ma^{\top} \ms \ma \approx \ma^{\top} \mv \ma$. We then use the M-matrix solver of \cite{ds08} (see \Cref{thm:fastmsolver}) to solve this lossy Laplacian system in $\wt{O}(n) \cdot \polylog(\frac{W}{\delta})$ time.

\section{Preliminaries}\label{sec:prelim}

\paragraph{Vectors.}
We use $\indicVec{i}$ to denote the $i$-th standard unit vector, and we use $\allone_n$ and $\allzero_n$ to denote all-ones and all-zeros vectors in $\R^{n}$.
When the dimension of the vector is clear from context, omit the subscript.

Given any $a,b\in \R^{n}$, we use $a \cdot b$ to denote the vector obtained from coordinate wise multiplication. Similarly, we also use other scalar operations to denote the corresponding entrywise vector operations when clear from context.

\paragraph{Matrices.}
We use $\mzero_{m,n}$ to denote the all-zeros matrix of dimension $m \times n$, and we use $\mI_n$ to denote the identity matrix of dimension $n \times n$. When the dimension of the matrix is clear from context, we also use $\mzero$ and $\mI$ without subscripts to denote all-zeros and identity matrices. 

For any $\ma \in \R^{m \times n}$, $I \subseteq [m]$, and $J \subseteq [n]$, we use $\ma_{I,J}$ to denote the submatrix of $\ma$ with rows in $I$ and columns in $J$, and we use $\ma_{I,:}$ to denote the submatrix of $\ma$ with rows in $I$, and $\ma_{:,J}$ to denote the submatrix of $\ma$ with columns in $J$.

We use $\mm \succeq \mzero$ %
to denote that the matrix $\mm$ is positive semidefinite (PSD). We use $\ma \succeq \mb$ to denote $\ma - \mb \succeq \mzero$. 
We use $\ma\approx_{c}\mb$
to denote that $c^{-1} \mb\preceq\ma\preceq c\mb$.
For any PSD matrix $\mm\in\R^{n\times n}$ we let $\lambda_{1}(\mm)\leq\ldots\leq\lambda_{n}(\mm)$
denote the eigenvalues of $\mm$. Let $v_{1}(\mm),\ldots,v_{n}(\mm)$
denote a corresponding orthonormal basis of eigenvectors, and let $\lambda_{\min}(\mm) \defeq \lambda_1(\mm)$ and $\lambda_{\max}(\mm) \defeq \lambda_n(\mm)$. When $\lambda_{\min}(\mm)$ is a simple eigenvalue, we also define $v_{\min}(\mm)$ to be the unit eigenvector corresponding to $\lambda_{\min}(\mm)$ such that its first non-zero coordinate is positive. %
When the matrix is clear from context, we also use the shorthands $\lambda_{\min}$, $\lambda_{\max}$, and $v_{\min}$. %

For any $\mm \in \R^{n \times n}$, we use $\mdiag(\mm)$ to denote the diagonal matrix consisting of the diagonal entries of $\mm$, and for any vector $v \in \R^n$, we also use $\mdiag(v)$ to denote the diagonal matrix with $v$ on the diagonal. 
Additionally, for any vector denoted in lowercase, we also refer to the diagonal matrix with the vector being on the diagonal with the same letter in uppercase, e.g., $\md = \mdiag(d)$. 

For any non-degenerate (i.e., full column rank) matrix $\ma \in \R^{m \times n}$, we denote the leverage scores of $\ma$ as $\sigma(\ma) \in \R^{m}$ where $\sigma(\ma)_i = (\ma (\ma^{\top} \ma)^{-1} \ma^{\top})_{i,i}$.

\paragraph{Norms.} For any vector $w\in \R^n_{>0}$, we define $\|v\|_w \defeq (\sum_{i=1}^n w_i v_i^2)^{1/2}$. For any norm $\|\cdot \|$, we define its induced matrix norm as $\|\mm\| \defeq \sup_{\|x\|=1} \|\mm x\|$. 

\paragraph{Sets.} We let $[n] \defeq \{1,2,\cdots,n\}$. For any $v \in \R^n$ and $\alpha \in \R$, we define $S_{\geq \alpha}(v) \defeq \{i \in [n] \mid v_i \geq \alpha\}$ and $S_{\leq \alpha}(v) \defeq \{i \in [n] \mid v_i \leq \alpha\}$.  %

\paragraph{Graphs.}
In this paper we consider undirected graphs that allow multi-edges but no self-loops, and we assign an orientation  to each edge. 
For any graph $G$, we use $V(G)$ and $E(G)$ to denote its vertices and edges. For any $V_1, V_2 \subseteq V(G)$, we use $E_G(V_1, V_2)$ to denote the set of edges between vertices in $V_1$ and vertices in $V_2$ (in either direction). With an abuse of notation, for any vertex $v \in V(G)$, we also denote $E_G(v, V_1) = E_G(\{v\}, V_1)$. For any vertex $v \in V(G)$, we use $\mathcal{N}_G(v)$ to denote the neighboring vertices of $v$ in $G$. We omit the subscript of $E_G(\cdot)$ and $\mathcal{N}_G(\cdot)$ when the graph is clear from context.

For any weighted graph $G = (V,E,w)$, for any $E' \subseteq E$, we define $w(E') \defeq \sum_{e \in E'} w_e$, and for any vertex $v \in V$, we define its weighted degree $\deg_w(v)$ as the sum of the weights of all edges incident to $v$, and for any $S \subseteq V$, we define $\vol_G(S) \defeq \sum_{v \in S} \deg_w(v)$. We omit the subscript of $\vol_G(\cdot)$ when the graph and its weights are clear from context. %
We define $\mb_{G} \in \R^{E \times V}$ as the incidence matrix of $G$ where row $e$ of $\mb_{G}$ is $\indicVec{b_e} -\indicVec{a_e}$. We define the (weighted) Laplacian of $G$ as $\mL_{G} \defeq \mb_{G}^{\top} \mw \mb_{G}$, and let $\md_{G} \defeq \mdiag(\mL_{G})$ denote its diagonal. We define the normalized Laplacian of $G$ as $\mn_{G} \defeq \md_{G}^{-1/2} \mL_{G} \md_{G}^{-1/2}$, and we let $d_{G}\in \R^V$ denote the vector of diagonal entries of $\md_{G}$. 
We define the conductance $\phi(G)$ as follows:%
\begin{align*}
    \phi(G) \defeq \min_{S \subseteq V}\frac{w(E_{G}(S,V\backslash S))}{\min\{\vol_G(S),\vol_G(V\backslash S)\}}.
\end{align*}
We say a graph $G$ has expansion $\phi$ if $\phi(G) \geq \phi$, and we also call such $G$ a \emph{$\phi$-expander}.

\paragraph{Lossy Flow and Lossy Graph Matrices.}
We call $G=(V,E,\eta,w)$ a \emph{weighted lossy flow graph }if $E$
is a multiset of pairs of distinct vertices in $V$, i.e., for each $e \in E$, $e=(a_{e},b_{e})$ 
with $a_{e},b_{e}\in V$ and $a_{e}\neq b_{e}$, $\eta\in\R_{\geq 1}^{E}$ are the flow multipliers, and $w \in \R_{\geq 0}^E$ are the edge weights. 

We define $V(G)$, $E(G)$, and $E_G(V_1,V_2)$ similarly as for graphs. 
We define $\mb_{G}\in\R^{E\times V}$
as the incidence matrix of $G$ where row $e$ of $\mb_{G}$
is $\indicVec{b_{e}} - \eta_{e} \indicVec{a_{e}}$.%

We define the (weighted) Laplacian of $G$ as $\mL_{G}\defeq\mb_{G}^{\top} \mw \mb_{G}$, %
and let $\md_{G}\defeq\mdiag(\mL_{G})$ denote
its diagonal. We also refer to the Laplacian of a lossy graph as a lossy Laplacian, to distinguish it from the standard graph Laplacian. %
We further define the normalized Laplacian of $G$ as $\mn_{G}\defeq\md_{G}^{-1/2}\mL_{G}\md_{G}^{-1/2}$. We also let $d_G \in \R^{V}$ denote the vector of diagonal entries of $\md_G$. We say that the graph is unweighted when $w$ is all one, and with an abuse of notation we use the shorthand $G = (V,E,\eta)$ to denote $G=(V,E,\eta,\allone_m)$.

We define $\gbar=(V,E,w)$ as the \emph{smoothed graph} associated with $G=(V,E,\eta,w)$, where $\gbar$ has the same set of edges as $G$ and its edges are non-lossy. We say a lossy graph $G$ is connected if $\gbar$ is connected. 
We call a lossy graph $G=(V,E,\eta,w)$ \emph{$\balloss$-balanced} if $\eta_e \leq 1 + \balloss$ for all $e\in E$. We say a lossy graph $G$ has expansion $\phi$, or equivalently we call $G$ a \emph{$\phi$-expander}, if the underlying smoothed graph $\ov{G}$ has conductance $\phi(\ov{G}) \geq \phi$. %

\paragraph{Data Structure Guarantees.}  
In all data structures provided in this paper, amortized bounds are taken over the entire sequence of all operations, unless an operation is explicitly stated to have a worst-case time bound. 
Only limited effort is made to optimize the polylogarithmic terms throughout this paper.

\section{Spectral Approximation of Lossy Graphs}\label{sec:spectral_approximation}

In this section we first prove \Cref{thm:spectral_sparsification}, the formal version of \Cref{thm:spectral_sparsification_informal}. Towards the end of this section, we also use the spectral results developed to prove \Cref{thm:spectral_sparsification} to show that the standard power iteration can be used to approximately compute $v_{\min}$ of the lossy Laplacian.

\Cref{thm:spectral_sparsification} shows that in order to obtain a constant spectral approximation of the form $\mI - (1 - \lambda)vv^\top$ to $\mn_G$, we only require $v^\top \mn_G v \leq c \lambda_1(\mn_G)$, for some constant $c$. To facilitate our later algorithmic development, we prove a more general version of this fact that allows more flexibility in the normalizing diagonal $\md$, as well as some additive error, $\epsad \mI$,.

\begin{theorem}[Spectral approximation from low Rayleigh quotient]\label{thm:spectral_sparsification} %
    Let $c_1, c_2 \geq 1$, let $0 \leq \epsad \leq 1$, let $G = (V,E,\eta,w)$ be a $\balloss$-balanced lossy flow graph with $\lambda_2(\mn_{\gbar}) \geq 20\balloss$, let $d\in \R_{\geq 0}^V$ be a vector that satisfies $d_G \leq d \leq c_1 d_G$. 
    Furthermore, let $v \in \R^V$ be a unit vector such that $v^\top (\md^{-1/2}\mL_G \md^{-1/2} + \epsad \mI)v \leq c_2 \lambda_1(\md^{-1/2}\mL_G \md^{-1/2} + \epsad \mI)$ where $\md = \mdiag(d)$. 
    Then, for $\lambda \defeq  v^\top (\md^{-1/2}\mL_G \md^{-1/2} + \epsad \mI)v$, %
    \begin{align*}
        \frac{(\lambda_2(\mn_{\gbar}))^2}{12c_1^2c_2^2}\left(\mI - (1 - \lambda)vv^\top\right) \preceq \md^{-1/2}\mL_G\md^{-1/2} + \epsad\mI \preceq \frac{24c_1c_2}{\lambda_2(\mn_{\gbar})}\left(\mI - (1 - \lambda)vv^\top\right).
    \end{align*}
\end{theorem}

Applying \Cref{thm:spectral_sparsification} with $\md = \md_G$, and $\epsad = 0$ immediately yields the following corollary:

\begin{corollary}
    Let $G = (V,E,\eta,w)$ be a $\balloss$-balanced lossy flow graph with $\lambda_2(\mn_{\gbar}) \geq 20\balloss$. Let $v$ be a unit vector such that $v^\top \mn_Gv \leq c \lambda_1(\mn_G)$ for some $c \geq 1$. 
    Then for $\lambda \defeq v^\top \mn_Gv$,
    \begin{align*}
        \frac{(\lambda_2(\mn_{\gbar}))^2}{12c^2}\left(\mI - (1 - \lambda)vv^\top\right) \preceq \mn_G \preceq \frac{24c}{\lambda_2(\mn_{\gbar})}\left(\mI - (1 - \lambda)vv^\top\right).
    \end{align*}
\end{corollary}

Before proving \Cref{thm:spectral_sparsification}, we prove some structural lemmas. First, we show that the normalized graph Laplacian and a sufficiently balanced normalized lossy Laplacian are close spectrally.

\begin{lemma}
\label{lem:eta_to_not} 
If $G=(V,E,\eta,w)$ \emph{is a $\balloss$-balanced lossy flow graph} for $\balloss \leq 1$, then,
\[
	\md_{\gbar} \preceq \md_G \preceq (1 + 3 \balloss)\md_{\gbar}
	\text{ and }
	\norm{\mn_{G}-\mn_{\gbar}}_{2}\leq 10 \balloss\,.
\]
\end{lemma}
\begin{proof}
Recall that each edge $e\in E$ is denoted as $(a_e, b_e)$, and row $e$ of $\mb_G$ and $\mb_{\gbar}$ are $\indicVec{b_{e}} - \eta_{e} \indicVec{a_{e}}$ and $\indicVec{b_{e}} - \indicVec{a_{e}}$ respectively. We have that
\begin{align*}
(\mb_G)_{e,:}^{\top} (\mb_G)_{e,:} - (\mb_{\gbar})_{e,:}^{\top} (\mb_{\gbar})_{e,:} = &~ (\indicVec{b_{e}} - \eta_{e} \indicVec{a_{e}}) (\indicVec{b_{e}} - \eta_{e} \indicVec{a_{e}})^{\top} - (\indicVec{b_{e}} - \indicVec{a_{e}}) (\indicVec{b_{e}} - \indicVec{a_{e}})^{\top} \\
= &~ (\eta_e^2-1)\indicVec{a_{e}}\indicVec{a_{e}}^{\top}-(\eta_e-1)\left(\indicVec{a_{e}}\indicVec{b_{e}}^{\top}+\indicVec{b_{e}}\indicVec{a_{e}}^{\top}\right).
\end{align*}
Using this equation, and by the definition of the Laplacian, %
\begin{align*}
\mL_{G}-\mL_{\gbar} & =\mb_{G}^{\top}\mw\mb_{G}-\mb_{\gbar}^{\top}\mw\mb_{\gbar}
=\sum_{e\in E}w_e\left[(\eta_e^2-1)\indicVec{a_{e}}\indicVec{a_{e}}^{\top}-(\eta_e-1)\left(\indicVec{a_{e}}\indicVec{b_{e}}^{\top}+\indicVec{b_{e}}\indicVec{a_{e}}^{\top}\right)\right].
\end{align*}
Let $\md^{\Delta}\defeq\mdiag(\mL_{G}-\mL_{\gbar})$
and let $\ma^{\Delta}\defeq \mL_{\gbar}-\mL_{G}+\md^{\Delta}$.
Recall that $\eta_e \geq 1$, $\balloss \leq 1$, and $\eta_e \leq 1 + \balloss$. Since $\eta_e^2-1 \leq 3 (\eta_e-1) \leq 3 \balloss$,
we have that entrywise
\begin{align*}
\mzero\leq\md^{\Delta} & = \sum_{e\in E}w_e(\eta_e^2-1)\indicVec{a_{e}}\indicVec{a_{e}}^{\top}
\leq \sum_{e\in E} 3 \balloss w_e\indicVec{a_{e}}\indicVec{a_{e}}^{\top}
\leq 3 \balloss\md_{G}\text{ and }\\
\mzero\leq\ma^{\Delta} & = \sum_{e \in  E} w_e (\eta_e - 1) \left(\indicVec{a_{e}}\indicVec{b_{e}}^{\top}+\indicVec{b_{e}}\indicVec{a_{e}}^{\top}\right) \leq\sum_{e\in E}\balloss w_e\left(\indicVec{a_{e}}\indicVec{b_{e}}^{\top}+\indicVec{b_{e}}\indicVec{a_{e}}^{\top}\right)\,.
\end{align*}
Since the entries of $\md_{G}^{-1/2}\md^{\Delta}\md_{G}^{-1/2}$ are non-negative and the matrix is diagonal,
\[
\norm{\md_{G}^{-1/2}\md^{\Delta}\md_{G}^{-1/2}}_{2}\leq\norm{\md_{G}^{-1/2}
\left(3 \balloss\md_{G}\right)\md_{G}^{-1/2}}_{2}
\leq 3 \balloss\,.
\]
This proves the first claim that $\md_{\gbar} \preceq \md_G \preceq (1 + 3 \balloss)\md_{\gbar}$.

Similarly, since the entries of $\md_{G}^{-1/2}\ma^{\Delta}\md_{G}^{-1/2}$ are non-negative, by the Perron–Frobenius theorem the largest eigenvalue of this matrix is equals to its spectral radius, and the corresponding eigenvector is non-negative, so 
\[
\norm{\md_{G}^{-1/2}\ma^{\Delta}\md_{G}^{-1/2}}_{2}
\leq \Big\|\md_{G}^{-1/2} \sum_{e\in E}\balloss w_e\left(\indicVec{a_{e}}\indicVec{b_{e}}^{\top}+\indicVec{b_{e}}\indicVec{a_{e}}^{\top}\right) \md_{G}^{-1/2}\Big\|_2
\leq\balloss\,,
\]
where in the last step we used that the matrix $\sum_{e\in E}w_e(\indicVec{a_{e}}\indicVec{b_{e}}^{\top}+\indicVec{b_{e}}\indicVec{a_{e}}^{\top})$ is the adjacency
matrix of a weighted graph where every vertex $a\in V$ has weighted degree at most
$[\md_{G}]_{aa}$. 

Consequently, we have that 
\[
\norm{\md_{G}^{-1/2}(\mL_{G}-\mL_{\gbar})\md_{G}^{-1/2}}_{2}
=\norm{\md_{G}^{-1/2}(\md^{\Delta}-\ma^{\Delta})\md_{G}^{-1/2}}_{2}\leq
4\balloss\,.
\]
Additionally, note that $1 \leq \eta_e^2 \leq (1 + \balloss)^2\leq 1 + 3\balloss$, and hence $\md_{\gbar}\preceq\md_{G}\preceq(1+3\balloss)\md_{\gbar}$. This then implies that
\begin{align*}
    \|\mI - \md^{-1/2}_G \md^{1/2}_{\gbar} \|_2 
    \leq 1 - \frac{1}{\sqrt{1 + 3\balloss}}
    \leq 3 \balloss.  %
\end{align*}
Since
\[
\norm{\md_{G}^{-1/2}\mL_{\gbar}\md_{\gbar}^{-1/2}}_{2} \leq \norm{\md_{\gbar}^{-1/2}\mL_{\gbar}\md_{\gbar}^{-1/2}}_{2}\leq 1,
\]
we have that:
\begin{align*}
&~ \|\md_{\gbar}^{-1/2}\mL_{\gbar}\md_{\gbar}^{-1/2} - \md_G^{-1/2}\mL_{\gbar} \md_G^{-1/2} \|_2 \\
\leq &~ \|(\mI - \md^{-1/2}_G \md^{1/2}_{\gbar})\md_{\gbar}^{-1/2}\mL_{\gbar} \md_{\gbar}^{-1/2}\|_2 + \|\md_{G}^{-1/2}\mL_{\gbar} \md^{-1/2}_{\gbar}(\mI - \md^{-1/2}_{G} \md^{1/2}_{\gbar})\|_2
\leq 6 \balloss,
\end{align*}
and hence:
\begin{equation*}
    \norm{\mn_{G}-\mn_{\gbar}}_{2}\leq \norm{\md_{G}^{-1/2}(\mL_{G}-\mL_{\gbar})\md_{G}^{-1/2}}_{2} + \norm{\md_{\gbar}^{-1/2}\mL_{\gbar}\md_{\gbar}^{-1/2} - \md_G^{-1/2}\mL_{\gbar} \md_G^{-1/2}}_2
    \leq 10\balloss. \qedhere
\end{equation*}
\end{proof}

A consequence of \Cref{lem:eta_to_not} is that so long as $\lambda_2(\mn_{\gbar})$ is greater than $20\balloss$, we can obtain a nontrivial lower bound on $\lambda_2(\mn_G)$. %

Next we prove that given the exact least eigenvalue $\lambda_1$ and the corresponding unit eigenvector $v_{\min}$ of $\md^{-1/2}\mL_G\md^{-1/2} + \epsad \mI$, the matrix $\mI - (1 - \lambda_1)v_{\min}v_{\min}^\top$ is a good spectral approximation of $\md^{-1/2}\mL_G\md^{-1/2} + \epsad \mI$. 
\begin{lemma}\label{lem:N_approx_Ivvt}
    Let $G = (V,E,\eta)$ be a $\balloss$-balanced lossy flow graph with $20\balloss \leq \lambda_2(\mn_{\gbar}) \leq 1$. Let $d\in \R_{\geq 0}^V$ satisfy $d_G \leq d \leq c \cdot d_G$ for some $c \geq 1$.  %
    Let $\lambda_1 = \lambda_1(\md^{-1/2}\mL_G\md^{-1/2} + \epsad \mI)$ for some $0 \leq \epsad \leq 1$, and let $v_{\min}$ be a corresponding unit eigenvector. Then,
    \begin{align*}
        \frac{\lambda_2(\mn_{\gbar})}{2c} \Big(\mI - (1 - \lambda_1)v_{\min}v_{\min}^\top\Big) \preceq \md^{-1/2}\mL_G\md^{-1/2}  + \epsad\mI \preceq 4 \Big(\mI - (1 - \lambda_1)v_{\min}v_{\min}^\top\Big).
    \end{align*}
    In particular, $\lambda_2(\md^{-1/2}\mL_G\md^{-1/2}  + \epsad\mI) \geq \frac{\lambda_2(\mn_{\gbar})}{2c}$ and $\lambda_n(\md^{-1/2}\mL_G\md^{-1/2}  + \epsad\mI) \leq 4.$
\end{lemma}
To prove this lemma, we use the following fact. %
\begin{fact}\label{fact:eigenvalue_scaling}
	Let $\mm,\md\in\R^{n\times n}$ be positive semi-definite matrices. If $\frac{1}{c_1} \mI \preceq \md \preceq c_2 \mI$ for $c_1,c_2 \geq 1$, then $\mm' \defeq \md \mm \md$ satisfies that
\[
\frac{\lambda_i(\mm)}{c_1^2} \leq \lambda_i(\mm') \leq c_2^2\lambda_i(\mm)
\text{ for all }i\in[n]\,.
\]
\end{fact}
\begin{proof}
The Courant-Fischer min-max theorem implies that
\begin{align*}
\lambda_i(\mm') = \min_{\dim(U)=i} \max_{x\in U} \frac{x^{\top} \mm' x}{x^{\top} x}
= &~ \min_{\dim(U)=i} \max_{x\in U} \frac{x^{\top} \md \mm \md x}{x^{\top} \md^2 x} \cdot \frac{x^{\top} \md^2 x}{x^{\top} x} \\
\leq &~ \min_{\dim(U)=i} \max_{x\in U} \frac{x^{\top} \md \mm \md x}{x^{\top} \md^2 x} \cdot c_2^2 
=  \lambda_i(\mm) \cdot c_2^2.
\end{align*}
An analogous argument implies that $\lambda_i(\mm') \geq \lambda_i(\mm) / c_1^2$.
\end{proof}

\begin{proof}[Proof of \Cref{lem:N_approx_Ivvt}]
    \Cref{lem:eta_to_not} implies that $\lambda_2(\mn_{G}) \geq \lambda_2(\mn_{\gbar}) - 10 \balloss \geq \lambda_2(\mn_{\gbar})/2$. %
    Since $\lambda_n(\mn_{\gbar}) \leq 2$, we also have that $\lambda_n(\mn_G) \leq 2 + 10 \balloss \leq 3$. %
    Now, since $\frac{1}{c}\mI \preceq \md_G\md^{-1} \preceq \mI$, \Cref{fact:eigenvalue_scaling} implies that $\lambda_2(\md^{-1/2}\mL_G\md^{-1/2})\geq \frac{1}{c}\lambda_2(\mn_G) \geq \lambda_2(\mn_{\gbar})/(2c)$, and $\lambda_n(\md^{-1/2}\mL_G\md^{-1/2})\leq \lambda_n(\mn_G) \leq 3$. Finally using that $0 \leq \epsad \leq 1$ we get the inequalities as desired. 
\end{proof}

In order to prove \Cref{thm:spectral_sparsification}, we also need to guarantee that there is a gap between the smallest and second smallest eigenvalues of $\md^{-1/2}\mL_G\md^{-1/2} + \epsad \mI$. 
The following lemma shows that when $\balloss$ is small this gap is at least $\Omega(\lambda_2(\mn_{\gbar}))$. 
Since the rescaling matrix $\md$ is not exactly $\md_G$, the next lemma requires an even smaller $\balloss$ than \Cref{lem:N_approx_Ivvt} to ensure such a spectral gap. 
\begin{lemma}\label{lem:spectral_gap}
    Let $G = (V,E,\eta)$ be a $\balloss$-balanced lossy flow graph with $\lambda_2(\mn_{\gbar}) \leq 1$. Let $d\in \R_{\geq 0}^V$ satisfy $d_G \leq d \leq c \cdot d_G$ for some $c \geq 1$, and let $\md = \mdiag(d)$. If $\balloss \leq \frac{\lambda_2(\mn_{\gbar})}{40c}$ then,
    \begin{align*}
        \lambda_2\left(\md^{-1/2}\mL_G\md^{-1/2} + \epsad \mI\right) - \lambda_1\left(\md^{-1/2}\mL_G\md^{-1/2} + \epsad \mI\right) \geq \frac{\lambda_2(\mn_{\gbar})}{4c}.
    \end{align*}
\end{lemma}
\begin{proof}
    $\lambda_2(\md^{-1/2}\mL_G\md^{-1/2}) \geq \frac{\lambda_2(\mn_{\gbar})}{2c}$ by \Cref{lem:N_approx_Ivvt} with $\epsad = 0$. \Cref{lem:eta_to_not} also implies that $\lambda_1(\mn_G) \leq \lambda_1(\mn_{\gbar}) + 10 \balloss = 10\balloss$. Since \Cref{fact:eigenvalue_scaling}, $\lambda_1(\md^{-1/2}\mL_G\md^{-1/2}) \leq 10\balloss$. 
    When $\balloss \leq \frac{\lambda_2(\mn_{\gbar})}{40c}$, we have that $\lambda_1(\md^{-1/2}\mL_G\md^{-1/2}) \leq \frac{1}{2} \lambda_2(\md^{-1/2}\mL_G\md^{-1/2})$, yielding the desired result. 
\end{proof}

Next we show that if two unit vectors $v_{1}$ and $v_2$ are close to each other, i.e., $v_{1}^{\top} v_2$ is close to $1$, then the two matrices $\mI-(1-\lambda) v_{1}v_{1}^{\top}$ and $\mI-(1-\lambda) v_{2}v_{2}^{\top}$ are good spectral approximations of each other. This gives a sufficient condition for a vector $v$ to approximate the true eigenvector $v_{\min}$ so that $\mI-(1-\lambda) v v^{\top} \approx \mI-(1-\lambda) v_{\min}v_{\min}^{\top}$. The proof of the following lemma is inspired by the proof of Lemma~27 of \cite{clmps16}.

\begin{lemma}\label{lem:rayleigh_implies_spectral}
    Let $v_{1} \neq v_2 \in \R^{n}$ be unit vectors where $1 - (v_{1}^\top v_2)^2 \leq c\lambda$, for $\lambda \in (0,1]$ and $c \geq 1$. Additionally, let $\mm_{i}\defeq\mI-(1-\lambda) v_{i}v_{i}^{\top}$ for $i\in \{1,2\}$. Then,
   $\frac{1}{3c} \mm_1 \preceq \mm_2 \preceq 3c\mm_1$.
\end{lemma}

\begin{proof}
    By symmetry, it suffices to show that $\mm_2 \preceq (1 + 2c)\mm_1$, which is equivalent to
    \begin{align*}
        \mm_1^{-1/2}(\mm_2 - \mm_1)\mm_1^{-1/2} \preceq 2c\mI.
    \end{align*}
   We will explicitly compute the eigenvalues of $\mm_1^{-1/2}(\mm_2 - \mm_1)\mm_1^{-1/2}$, and show that they are at most $2c$. Note that $\mm_{1}^{-1/2}(\mm_{2}-\mm_{1})\mm_{1}^{-1/2}v=\gamma v$
if and only if $x=\mm_{1}^{-1/2}v$ satisfies 
\[
(\mm_{2}-\mm_{1})x=\gamma\mm_{1}x\,.
\]
Additionally, since $\mm_{2}-\mm_{1}=(1-\lambda)(v_{1}v_{1}^{\top}-v_{2}v_{2}^{\top})$, we see that $\gamma$ is a non-zero eigenvalue of $\mm_{1}^{-1/2}(\mm_{2}-\mm_{1})\mm_{1}^{-1/2}$ only when $x$ is in the span of $v_{\min}$ and $v_2$. Let $x=\alpha v_{1}+\beta v_{2}$ such that $(\mm_{2}-\mm_{1})x=\gamma\mm_{1}x$.
Since
\begin{align*}
(\mm_{2}-\mm_{1})x & =(1-\lambda)\left[\big(\alpha+\beta(v_{1}^{\top}v_{2})\big)v_{1}-\big(\alpha(v_{1}^{\top}v_{2})+\beta\big)v_{2}\right]\text{ and}\\
\gamma\mm_{1}x & =\gamma\left[\big(\lambda\alpha-\beta(1-\lambda)(v_{1}^{\top}v_{2})\big)v_{1}+\beta v_{2}\right]\,,
\end{align*}
we have $(\mm_{2}-\mm_{1})x=\gamma\mm_{1}x$ is equivalent to
\begin{align*}
(1-\lambda)\left(\begin{array}{cc}
1 & v_{1}^{\top}v_{2}\\
-v_{1}^{\top}v_{2} & -1
\end{array}\right)\left(\begin{array}{c}
\alpha\\
\beta
\end{array}\right)=\gamma\left(\begin{array}{cc}
\lambda & -(1-\lambda)v_{1}^{\top}v_{2}\\
0 & 1
\end{array}\right)\left(\begin{array}{c}
\alpha\\
\beta
\end{array}\right),
\end{align*}
or equivalently:
\begin{align*}
\left(\begin{array}{cc}
1-\lambda-\gamma\lambda & (1 - \lambda)(1 + \gamma) v_{1}^\top v_2\\
-(1-\lambda)v_{1}^{\top}v_{2} & \lambda- 1 - \gamma
\end{array}\right)\left(\begin{array}{c}
\alpha\\
\beta
\end{array}\right)=0.
\end{align*}
Denote the matrix in the last equation above as $\mn$. A non-trivial solution $(\alpha,\beta) \neq (0,0)$ exists if and only if $\det(\mn) = 0$, and 
\begin{align*}
0 = \det(\mn) = &~ (1-\lambda - \gamma\lambda)\cdot(\lambda-1-\gamma) - \big((1 - \lambda)(1 + \gamma) v_{1}^\top v_2\big)\cdot \big(-(1-\lambda) v_{1}^\top v_2\big) \\
= &~ \lambda\gamma^2 -(1 - \lambda)^2\gamma -(1-\lambda)^2+(1 - \lambda)^2(v_{1}^\top v_2)^2\gamma + (1 - \lambda)^2(v_{1}^\top v_2)^2 \\
= &~ \lambda\gamma^2 - (1 -\lambda)^2 \big(1-(v_{1}^\top v_2)^2 \big)\gamma - (1 - \lambda)^2 \big(1-(v_{1}^\top v_2)^2\big).
\end{align*}
Solving this quadratic equation, the two possible values of $\gamma$ are:
\begin{align*}
    \gamma = \frac{(1 - \lambda)^2\big(1 - (v_{1}^\top v_2)^2\big)\pm \sqrt{(1 - \lambda)^4\big(1 - (v_{1}^\top v_2)^2\big)^2+4\lambda(1 - \lambda)^2\big(1 - (v_{1}^\top v_2)^2\big)}}{2\lambda}.
\end{align*}
Since $\sqrt{a^2 + b^2} \leq a + b$ for any $a \geq 0$ and $b \geq 0$, we have that
\begin{align*}
    |\gamma| \leq \frac{2(1 - \lambda)^2\big(1-(v_{1}^\top v_2)^2\big) + 2\sqrt{\lambda}(1 - \lambda)\sqrt{\big(1-(v_{1}^\top v_2)^2\big)}}{2\lambda}.
\end{align*}
Since $v_{1}$ and $v_2$ satisfy $1 - (v_{1}^\top v_2)^2 \leq c\lambda$, we have that
\begin{align*}
|\gamma| \leq c(1 - \lambda)^2 + \sqrt{c}(1 - \lambda) \leq 2c.
\end{align*}
Consequently, each eigenvalue of $\mm_1^{-1/2}(\mm_2 - \mm_1)\mm_1^{-1/2}$ has absolute value at most $2c$, and therefore $\mm_2 \preceq (1 + 2c) \mm_1 \preceq 3c \mm_1$. By symmetry, $\mm_1 \preceq 3c \mm_2$ as well, completing the proof.
\end{proof}

Next, in \Cref{lem:bound_on_v_topv} we prove that $1 - v^{\top} v_{\min}(\mm)$ is small if the Rayleigh quotient $v^{\top} \mm v$ is small. Combined with \Cref{lem:rayleigh_implies_spectral}, \Cref{lem:bound_on_v_topv} shows that a vector $v$ with a small Rayleigh quotient is sufficient to obtain the spectral approximation $\mm \approx \mI - (1-\lambda) v v^{\top}$. 

\begin{lemma}%
\label{lem:bound_on_v_topv} 
$1 - \left(v^{\top} v_{1}(\mm)\right)^2 \leq \frac{v^{\top} \mm v}{\lambda_2(\mm)}$ for any PSD  $\mm \in \R^{n \times n}$ and unit $v \in \R^n$.
\end{lemma}
\begin{proof}
Note that $v = \sum_{i \in [n]} (v_i(\mm)^\top v) v_i(\mm)$, and $1 = \|v\|_2^2 = \sum_{i  \in [n]} (v_i(\mm)^\top v)^2$. Consequently,
\begin{align*}
v^\top \mm v = 
 \sum_{i = 1}^{n} (v_i(\mm)^\top v)^2 \lambda_i(\mm)
\geq
\sum_{i = 2}^{n} (v_i(\mm)^\top v)^2 \lambda_2(\mm)
= \left(1 - (v_{1}(\mm)^\top v)^2\right) \lambda_2(\mm)\,.
\end{align*}
Rearranging terms proves this lemma.
\end{proof}

Putting these lemmas together, we now have the tools needed to prove \Cref{thm:spectral_sparsification}.
\begin{proof}[Proof of \Cref{thm:spectral_sparsification}]
In this proof we define $\mm \defeq \md^{-1/2}\mL_G\md^{-1/2} + \epsad\mI$, and we denote $v_{\min} \defeq v_{\min}(\mm)$. 
By \Cref{lem:N_approx_Ivvt}, 
    \begin{align}\label{eq:thm_3_1_eq_1}
        \frac{\lambda_2(\mn_{\gbar})}{2c_1}\left(\mI - (1 - \lambda_1(\mm))v_{\min}v_{\min}^\top\right) 
        \preceq  \mm
        \preceq 4\left(\mI - (1 - \lambda_1(\mm))v_{\min}v_{\min}^\top\right).
    \end{align}
    In particular, the above equation implies that $\lambda_2(\mm) \geq \frac{\lambda_2(\mn_{\gbar})}{2c_{1}}$. Now, applying \Cref{lem:bound_on_v_topv} to $\mm= \md^{-1/2}\mL_G\md^{-1/2} + \epsad\mI$: %
    \begin{align*}
        1 - (v^\top v_{\min})^2 \leq \frac{v^{\top} \mm v}{\lambda_2(\mm)} \leq \frac{2c_1c_2}{\lambda_2(\mn_{\gbar})}\lambda_1(\mm). 
    \end{align*}
    Using this upper bound on $1 - (v^\top v_{\min})^2$, we can apply \Cref{lem:rayleigh_implies_spectral} to $v$,$v_{\min}$ and $\lambda_1(\mm)$:
    \begin{align}\label{eq:thm_3_1_eq_2}
        \frac{\lambda_2(\mn_{\gbar})}{6c_1c_2}\left(\mI - (1 - \lambda_1(\mm))vv^\top\right) \preceq \mI - (1 - \lambda_1(\mm))v_{\min}v_{\min}^\top \preceq \frac{6c_1c_2}{\lambda_2(\mn_{\gbar})}\left(\mI -(1 - \lambda_1(\mm))vv^\top\right).
    \end{align}
    Since $\lambda_1(\mm) \leq \lambda \leq c_2\lambda_1(\mm)$, 
    \begin{align}\label{eq:thm_3_1_eq_3}
        \frac{1}{c_2}\left(\mI - (1 - \lambda)vv^\top\right) \preceq \mI - (1 - \lambda_1(\mm))vv^\top \preceq \mI - (1 - \lambda)vv^\top.
    \end{align}
    Combining Eq.~\eqref{eq:thm_3_1_eq_1}, \eqref{eq:thm_3_1_eq_2}, and \eqref{eq:thm_3_1_eq_3} implies the desired bounds of
    \begin{align*}
        \mm &\preceq 4\left(\mI - (1 - \lambda_1(\mm))v_{\min}v_{\min}^\top\right)
        \preceq \frac{24c_1c_2}{\lambda_2(\mn_{\gbar})}\left(\mI - (1 - \lambda_1(\mm))vv^\top\right)
        \preceq \frac{24c_1c_2}{\lambda_2(\mn_{\gbar})}\left(\mI - (1 - \lambda)vv^\top\right)
    \end{align*}
   and
   \begin{align*}
       \mm &\succeq \frac{\lambda_2(\mn_{\gbar})}{2c_1}\left(\mI - (1 - \lambda_1(\mm))v_{\min}v_{\min}^\top\right) \\ 
       &\succeq \frac{(\lambda_2(\mn_{\gbar}))^2}{12c_1^2c_2}\left(\mI - (1 - \lambda_1(\mm))vv^\top\right) 
       \succeq \frac{(\lambda_2(\mn_{\gbar}))^2}{12c_1^2c_2^2}\left(\mI - (1 - \lambda)vv^\top\right)\,. \qedhere
   \end{align*}
   \end{proof}

\paragraph{Power iteration for lossy Laplacian.} 
Using the spectral results proven in this section, we can also show that the standard power iteration can compute a constant approximation to the least eigenvector of the lossy Laplacian. We will use the following standard power iteration with a random start (see e.g.~\cite{kw92}). %
\begin{theorem}[Power Iteration]\label{thm:power_method}
There is an algorithm $\textsc{PowerIteration}(\mm, \epsilon)$ that for any input PSD $\mm \in \mathbb{R}^{n\times n}$, and a parameter $\epsilon \in (0,1)$,
returns w.h.p.~a unit vector $v \in \R^{n}$ such that $v^\top \mm v \geq (1 
-\epsilon)\lambda_{n}(\mm)$ in $O\Big(\nnz(\mm)\log(\frac{n}{\epsilon})\max\{1, (\frac{\lambda_{n}(\mm)}{\lambda_{n-1}(\mm)}-1)^{-1}\} \Big)$ time.
\end{theorem}

Next we apply the spectral results proven in this section to prove the guarantees of \textsc{PowerIteration} when used to compute the least eigenvalue of $\md^{-1/2}\mL_G \md^{-1/2}+\epsad \mI$.
\begin{theorem}\label{thm:power_iteration_least_eigenvalue}
Let $G = (V,E,\eta)$ be a $\balloss$-balanced lossy flow graph with $\lambda_2(\mn_{\gbar}) \leq 1$, let $d\in \R_{\geq 0}^V$ be a vector that satisfies $d_G \leq d \leq c \cdot d_G$ for some $c \geq 1$, and let $\md = \mdiag(d)$. Furthermore, assume that $\balloss \leq \frac{\lambda_2(\mn_{\gbar})}{40c}$.
Let the unit vector 
\[
v = \textsc{PowerIteration}\left(4 \mI - (\md^{-1/2}\mL_G \md^{-1/2}+\epsad \mI),\epsad/4\right),
\]
where $\textsc{PowerIteration}$ is defined in \Cref{thm:power_method}. Then w.h.p.~$v$ satisfies 
\[
v^{\top} \left(\md^{-1/2}\mL_G\md^{-1/2}+\epsad \mI\right) v \leq 2 \lambda_1 \left(\md^{-1/2}\mL_G\md^{-1/2}+\epsad \mI\right),
\]
and can be computed in $O\left(m(\lambda_2(\mn_{\gbar}))^{-1}c \log{(\frac{n}{\epsad})}\right)$ time.
\end{theorem}
\begin{proof}
In this proof we denote $\mm \defeq \md^{-1/2}\mL_G\md^{-1/2}+\epsad \mI$.

\textbf{Correctness:} We first show that 
$v^{\top} \mm v \leq 2 \lambda_1 (\mm)$. 
\Cref{lem:N_approx_Ivvt} implies that $\lambda_n(\mm) \leq 4$, and as such $4\mI - \mm$ is PSD. Let $v_{\min} \defeq v_{\min}(\mm) = v_n(4\mI - \mm)$. The guarantees of $\textsc{PowerIteration}$ in \Cref{thm:power_method} imply that
\begin{align*}
v^\top (4\mI - \mm)v &\geq (1 - \epsad/4) v_{\min}^\top (4\mI - \mm) v_{\min},
\end{align*}
and re-arranging this inequality yields
\begin{align*}
v^\top \mm v &\leq (1 - \epsad/4) v_{\min}^\top\mm v_{\min} + \epsad 
\leq 2\lambda_1(\mm).
\end{align*}

\textbf{Time complexity:} By \Cref{thm:power_method}, the time of each call to \textsc{PowerIteration} is 
\[
O\Big(m\log(\frac{n}{\epsad}) (\frac{\lambda_n(4 \mI - \mm)}{\lambda_{n-1}(4 \mI - \mm)}-1)^{-1}\Big).
\]
We now bound $\frac{\lambda_n(4 \mI - \mm)}{\lambda_{n-1}(4 \mI - \mm)} - 1$. By \Cref{lem:spectral_gap} %
we have that $\lambda_{2}(\mm) - \lambda_1(\mm) \geq \frac{\lambda_2(\mn_{\gbar})}{4c}$, which implies that
\begin{align*}
    \lambda_n(4 \mI - \mm) - \lambda_{n-1}(4 \mI - \mm) \geq \frac{\lambda_2(\mn_{\gbar})}{4c}.
\end{align*}
Since $\mm$ is PSD, we also have that $\lambda_{n-1}(4 \mI - \mm) \leq 4$, which gives us
\begin{align*}
    \frac{\lambda_n(4 \mI - \mm)}{\lambda_{n-1}(4 \mI - \mm)} - 1 &= \frac{\lambda_{n}(4 \mI - \mm)-\lambda_{n-1}(4 \mI - \mm)}{\lambda_{n-1}(4 \mI - \mm)}
    \geq \frac{\lambda_2(\mn_{\gbar})}{16c}. %
\end{align*}
This lower bound implies that $\left(\frac{\lambda_n(4 \mI - \mm)}{\lambda_{n-1}(4 \mI - \mm)} -1\right)^{-1} \leq \frac{16c}{\lambda_2(\mn_{\gbar})}$, so the total time complexity is $O\Big(m\log(\frac{n}{\epsad}) (\frac{\lambda_n(4 \mI - \mm)}{\lambda_{n-1}(4 \mI - \mm)}-1)^{-1}\Big) \leq O\left(m(\lambda_2(\mn_{\gbar}))^{-1} c \log{(\frac{n}{\epsad})}\right)$.
\end{proof}

\section{Uniformity of $v_{\min}(\mL_G)$}\label{sec:v_uniformity}
In this section we show that $v_{\min}$ of a lossy Laplacian is approximately the all-ones vector when the underlying lossy graph is an expander and its flow multipliers are close to $1$. We use this structural property in later sections when performing vertex deletions in our heavy hitter data structure for lossy incidence matrices.

Our proof of the uniformity of $v_{\min}$ can be viewed as a strengthening of the standard argument that the all-ones vector is in the kernel of a non-lossy Laplacian (and is therefore an eigenvector corresponding to the smallest eigenvalue). For a standard graph Laplacian $\mL_{\gbar}$, any eigenvector $v$ with eigenvalue $\lambda$ satisfies 
\begin{align*}
    \lambda v_i = (\mL_{\gbar}v)_i = d_iv_i - \sum_{j \in \mathcal{N}(i)}v_j.
\end{align*}
Consequently, when $\lambda = 0$, $
    (v_{\min})_i = \frac{1}{d_i}\sum_{j \in \mathcal{N}(i)}(v_{\min})_j$, 
or in other words, for each vertex $i$, $(v_{\min})_i$ has the average value of its neighboring vertices. This is consistent with the fact that $v_{\min}$ is a scaling of $\allone$.

Now consider the lossy Laplacian $\mL_G$. %
Let $V_{\text{large}}$ denote the set of vertices for which $(v_{\min})_i $ is at least some threshold $\zeta$, i.e., $V_{\text{large}} = S_{\geq \zeta}(v)$. We show that many vertices adjacent to $V_{\text{large}}$ must have $v_{\min}$ values that are larger than $\zeta-\epsilon$ for some small $\epsilon$. If the underlying graph is an expander, after repeating this argument for roughly $\log{n}$ steps, we can bound each entry of $v_{\min}$. (Similar proof ideas were used for proving spectral clustering results, see e.g.~\cite{aalg18,gklms21}.)

Applying this approach, we prove the following theorem.

\begin{theorem}[Uniformity of $\frac{v_{\min}}{\sqrt{d}}$]\label{thm:unifomity_v}
Let $G=(V,E,\eta)$ be a $\balloss$-balanced lossy flow graph that is connected and has expansion $\phi$, %
with $|E|=m$. Let $d \in \R^{V}_{\geq 0}$ satisfy $d_G \leq d \leq c\cdot d_G$ for $c \geq 1$, %
and let $\md = \mdiag({d})$. Let $\lambda \defeq \lambda_{\min}(\md^{-1/2}\mL_G\md^{-1/2})$, and let $v \defeq v_{\min}(\md^{-1/2}\mL_G\md^{-1/2})$. 
If $\balloss < 0.1$, $\phi < 0.1$, $c\lambda < 1$, and $\balloss + c\lambda \leq \frac{\phi^2}{100 \log{m}}$, then 
\[
\frac{\max_{i \in V} z_i}{\min_{i \in V} z_i} 
\leq  \exp\left(O\left(\frac{(\balloss + c\lambda) \log m}{\phi^2}\right)\right), \text{ where } z_i \defeq \frac{v_i}{\sqrt{d_i}} \text{ for all } i \in V.
\]
\end{theorem}

Before proving this theorem, consider a lossy graph where $\balloss \leq \frac{\phi^2}{2000 c \log^2(m)}$, then \Cref{lem:eta_to_not} and \Cref{fact:eigenvalue_scaling} imply that $\lambda = \lambda_{\min}(\md^{-1/2}\mL_G\md^{-1/2}) \leq 10\balloss$, so the parameters satisfy $\balloss + c \lambda \leq \frac{\phi^2}{100 \log^2(m)}$, which means $\exp\left(O\left(\frac{(\balloss + c\lambda) \log m}{\phi^2}\right)\right) \leq 1 + O(\frac{(\balloss + c\lambda) \log m}{\phi^2}) \leq 1+ O(\frac{1}{\log m})$.
\begin{corollary}\label{cor:uniformity_v_with_parameters}
In the setting of \Cref{thm:unifomity_v}, if the parameters satisfy $\balloss \leq \frac{\phi^2}{2000 c \log^2(m)}$, then
\[
\frac{\max_{i \in V} \frac{v_i}{\sqrt{d_i}}}{\min_{i \in V} \frac{v_i}{\sqrt{d_i}}} 
\leq 1 + O\left(\frac{1}{\log m}\right).
\]
\end{corollary}

To prove \Cref{thm:unifomity_v}, we first prove the following basic fact that $\lambda_{\min}(\md^{-1/2}\mL_G\md^{-1/2})$ is a simple eigenvalue and all entries of $v_{\min}(\md^{-1/2}\mL_G\md^{-1/2})$ are positive. 
\begin{fact}\label{fact:v_min_positive}
    Let $G = (V,E,\eta)$ be a lossy flow graph %
    that is connected. Let $d \in \R^{V}_{\geq 0}$ satisfy $d_G \leq d \leq c\cdot d_G$ for some $c \geq 1$, and let $\md = \mdiag(d)$. Then $\lambda_{\min}(\md^{-1/2}\mL_G\md^{-1/2})$ is a simple eigenvalue, and $v = v_{\min}(\md^{-1/2}\mL_G\md^{-1/2})$ is strictly positive, i.e., $v_i > 0$ for all $i \in V$. %
\end{fact}
\begin{proof}
    First, note that $\md^{-1/2}\mL_G \md^{-1/2}$ is a matrix with non-positive off diagonal entries, and a positive diagonal. As such, we can write it as $\md^{-1/2}\mL_G \md^{-1/2} = \mu\mI - \ma$ for some positive integer $\mu$, and some non-negative matrix $\ma$. Now, $\md^{-1/2}\mL_G \md^{-1/2}$ and $\ma$ have the same eigenspaces. $\ma$ is irreducible because $G$ is connected, so the Perron-Frobenius theorem implies that $\lambda_{\max}(\ma)$ is a simple eigenvalue, and it has a corresponding eigenvector $\lambda_{\max}(\ma)$ can be scaled to entrywise positive. %
    This vector is also an eigenvector of $\md^{-1/2}\mL_G \md^{-1/2}$. 
\end{proof}

Next we prove the following two key lemmas. The first lemma provides a lower bound on how much the %
volume of the set $S_{\geq \zeta}(z)$ increases when $\zeta$ is decreased.

\begin{lemma}[Sweep cut bound of $S_{\geq  \zeta}(z)$ with decreasing threshold]\label{lem:v_unifomity_one_step}

In the setting of \Cref{thm:unifomity_v}, for any $\zeta > 0$ such that $\max_{i \in V}{z} \leq 2\zeta$ and $\sum_{i \in S_{\geq \zeta}(z)}(d_{\gbar})_i \leq \frac{1}{2}\sum_{i \in V}(d_{\gbar})_i$, 
\[
\sum_{i \in S_{\geq \zeta'}(z)} (d_{\gbar})_i \geq \left(1 + \frac{\phi}{2}\right) \cdot \sum_{i \in S_{\geq \zeta}(z)} (d_{\gbar})_i~~
\text{ for } 
\zeta' \defeq \left(1 - \frac{10(\balloss + c\lambda)}{\phi}\right) \zeta\,.
\]
\end{lemma}
\begin{proof}
Let $S \defeq S_{\geq \zeta}(z)$ and $S' \defeq S_{\geq \zeta'}(z)$. Also recall that we use $d$ to denote our approximation to $d_G$, the diagonal of the Laplacian $\mL_G$, and $d_{\gbar}$ to denote the degrees of the smoothed graph $\gbar$. %

Since $v$ is the eigenvector corresponding to the smallest eigenvalue $\lambda$ of $\md^{-1/2}\mL_G\md^{-1/2}$, it satisfies that $\lambda v = \md^{-1/2}\mL_G\md^{-1/2}v$. Therefore $z = \md^{-1/2}v$ satisfies $\lambda \md z = \mL_Gz$. %
By considering $S$ and $z$ we see that,
\begin{align*}
    \lambda \cdot \sum_{i \in S} d_i \cdot z_i = &~ \sum_{i \in S} (\mL_G\cdot z)_i\\
    =&~\sum_{i \in S}\left(\sum_{e =(a,i)\in E}(z_i-\eta_ez_a) + \sum_{e =(i,a)\in E}(-\eta_ez_a+\eta_e^2z_i)\right)\\
    \geq&~\sum_{i \in S}\left((d_G)_i\cdot z_i- (1 + \balloss)\sum_{e = (i,a) \text{ or }(a,i) \in E}z_a\right),
\end{align*}
where the second step follows from $\mL_G = \mb_G^\top \mb_G$, and the row $e=(a,b)$ of $\mb_G$ is $\indicVec{b}-\eta_e\indicVec{a}$, and the third step follows from $\eta_e \leq (1 + \balloss)$. Rearranging yields that:
\begin{align*}
\sum_{i \in S}\Big((d_G)_i - \lambda d_i\Big) \cdot z_i \leq (1 + \balloss)\sum_{i \in S}\sum_{e = (a,i) \text{ or }(i,a) \in E }z_a.
\end{align*}
Since $d \leq c\cdot d_G$, the above equation implies that: %
\begin{align*}
    (1 - c\lambda)\sum_{i \in S}(d_G)_i \cdot z_i &\leq (1 + \balloss)\sum_{i \in S}\sum_{e = (a,i) \text{ or }(i,a) \in E }z_a \notag \\
    &= (1 + \balloss)\left(\sum_{i \in S}|E(i,S)|z_i + \sum_{i\in S} \sum_{e=(i,a) \text{ or } (a,i) \in E(S, \ov{S})} z_a\right),
\end{align*}
where in the second step we define $\ov{S} \defeq E \backslash S$ and split the last summation over edges from $S$ to $\overline{S}$ and edges from $S$ to $S$. 
Rearranging terms yields that
\begin{align}\label{eq:v_unifomity_one_step_01}
    &~ (1 - c\lambda) \left(\sum_{i \in S}(d_G)_i \cdot z_i - \sum_{i \in S}|E(i,S)|z_i \right) \notag \\
    \leq &~ (\balloss + c \lambda) \sum_{i \in S}|E(i,S)|z_i + (1 + \balloss)\sum_{i\in S} \sum_{e=(i,a) \text{ or } (a,i) \in E(S, \ov{S})} z_a .
\end{align}
Since for every $i \in S$, $(d_G)_i \geq |E(i,S)|$ and $z_i \geq \zeta$, 
\begin{align}\label{eq:v_unifomity_one_step_02}
(1 - c\lambda) \left(\sum_{i \in S}(d_G)_i \cdot z_i - \sum_{i \in S}|E(i,S)|z_i \right) \geq (1 - c\lambda) \left(\sum_{i \in S}(d_G)_i - \sum_{i \in S}|E(i,S)| \right) \zeta.
\end{align}
Since $z_i \leq 2 \zeta$, we also have
\begin{align}\label{eq:v_unifomity_one_step_03}
&~(\balloss + c \lambda) \sum_{i \in S}|E(i,S)|z_i + (1 + \balloss)\sum_{i\in S} \sum_{e=(i,a) \text{ or } (a,i) \in E(S, \ov{S})} z_a \notag \\
\leq &~ 2(\balloss + c \lambda) \sum_{i \in S}|E(i,S)| \zeta + (1 + \balloss)\sum_{i\in S} \sum_{e=(i,a) \text{ or } (a,i) \in E(S, \ov{S})} z_a.
\end{align}
Combining the above three equations \Cref{eq:v_unifomity_one_step_01}, \eqref{eq:v_unifomity_one_step_02}, and \eqref{eq:v_unifomity_one_step_03} yields that 
\begin{align}\label{eq:v_unifomity_one_step_1}
\left((1 - c\lambda) \sum_{i \in S}(d_G)_i - \Big(1 - c\lambda + 2(\balloss + c \lambda) \Big) \sum_{i \in S}|E(i,S)| \right) \zeta \leq (1 + \balloss)\sum_{i\in S} \sum_{e=(i,a) \text{ or } (a,i) \in E(S, \ov{S})} z_a.
\end{align}
Since $\sum_{i \in S} (d_G)_i \geq \sum_{i \in S} (d_{\gbar})_i = (\sum_{i \in S} |E(i,S)|) + |E(S, \ov{S})|$,
\begin{align*}
&~(1 - c\lambda) \cdot \sum_{i \in S} (d_G)_i - \Big(1 - c\lambda + 2(\balloss + c \lambda) \Big)\cdot \sum_{i \in S} |E(i,S)| \\
=&~ (1 - c\lambda) \cdot \sum_{i \in S} (d_G)_i - \Big(1 + c\lambda + 2\balloss\Big)\cdot \left(\sum_{i \in S} (d_{\gbar})_i-|E(S, \ov{S})|\right)\\
\geq &~ \left(1 + c\lambda + 2\balloss\right)\cdot|E(S, \ov{S})|-(2\balloss + 2 c\lambda)\sum_{i\in S}(d_G)_i.
\end{align*}
Since the graph $\gbar$ has expansion $\phi$, and since $\sum_{i \in S}(d_{\gbar})_i \leq \frac{1}{2} \sum_{i \in V}(d_{\gbar})_i$, we have that $|E(S,\ov{S})| \geq \phi \sum_{i \in S}(d_{\gbar})_i \geq \phi (1 + \balloss)^{-2}\sum_{i \in S}(d_G)_i$, hence the above equation gives
\begin{align*}
&~(1 - c\lambda) \cdot \sum_{i \in S} (d_G)_i - \Big(1 - c\lambda + 2(\balloss + c \lambda) \Big)\cdot \sum_{i \in S} |E(i,S)| \\
\geq &~ \left(1 + c\lambda + 2\balloss - \frac{(2\balloss+ 2c\lambda)(1 + \balloss)^2}{\phi}\right) \cdot |E(S,\ov{S})|\\
\geq &~\left(1 - \frac{3(\balloss + c\lambda)}{\phi}\right)\cdot |E(S,\ov{S})|.
\end{align*}
Combining this equation with \Cref{eq:v_unifomity_one_step_1} yields that
\begin{align}\label{eq:v_unifomity_one_step_2}
\sum_{i\in S} \sum_{e=(i,a) \text{ or } (a,i) \in E(S, \ov{S})} z_a &\geq 
\left(1 - \frac{3(\balloss + c\lambda)}{\phi}-\balloss\right) \cdot |E(S, \ov{S})| \cdot \zeta \notag \\
&\geq \left(1 - \frac{4\balloss + 3c\lambda}{\phi}\right) \cdot |E(S, \ov{S})| \cdot \zeta.
\end{align}

Let $T \subseteq \ov{S}$ denote the set of vertices such that $z_a \geq \zeta' = (1 - \frac{10 (\balloss + c\lambda)}{\phi}) \cdot \zeta$. \Cref{eq:v_unifomity_one_step_2} implies that we must have $|E(S,T)| \geq \frac{\phi}{2} \cdot \sum_{i \in S} (d_{\gbar})_i$, and next we prove this claim by contradiction. Assume that we instead have $|E(S,T)| < \frac{\phi}{2} \cdot \sum_{i \in S} (d_{\gbar})_i$. First note that since any $a \notin S$ satisfies $z_a \leq \zeta$ and any $a \notin S \cup T$ satisfies $z_a \leq (1 - \frac{10 (\balloss + c\lambda)}{\phi}) \cdot \zeta$, 
\begin{align*}
\sum_{i \in S} \sum_{e = (i,a) \text{ or } (a,i) \in E(S, \ov{S})} z_a = &~ \sum_{i \in S} \sum_{e = (i,a) \text{ or } (a,i) \in E(S, T)} z_a + \sum_{i \in S} \sum_{e = (i,a) \text{ or } (a,i) \in E(S, \ov{S}\backslash T)} z_a \\
\leq &~ \zeta \cdot |E(S,T)| + \left(1 - \frac{10 (\balloss + c\lambda)}{\phi}\right) \cdot \zeta \cdot |E(S, \ov{S}\backslash T)| \\
= &~ \frac{10 (\balloss + c\lambda)}{\phi} \cdot \zeta \cdot |E(S,T)| + \left(1 - \frac{10 (\balloss + c\lambda)}{\phi}\right) \cdot \zeta \cdot |E(S, \ov{S})|,
\end{align*}
where in the third step we used that $|E(S, \ov{S}\backslash T)| = |E(S, \ov{S})| - |E(S, T)|$. Using the assumption that $|E(S,T)| < \frac{\phi}{2} \cdot \sum_{i \in S} (d_{\gbar})_i$, and since $|E(S,\ov{S})| \geq \phi \sum_{i \in S}(d_{\gbar})_i$, the above equation gives
\begin{align*}
\sum_{i \in S} \sum_{e = (i,a) \text{ or } (a,i) \in E(S, \ov{S})} z_a < &~ \frac{10(\balloss + c\lambda)}{2} \cdot \zeta \cdot \sum_{i \in S} (d_{\gbar})_i + \left(1 - \frac{10(\balloss + c\lambda)}{\phi}\right) \cdot \zeta \cdot |E(S, \ov{S})| \\
\leq &~ \left( 1 - \frac{10 (\balloss + c\lambda)}{\phi} + \frac{10(\balloss + c\lambda)}{2 \phi} \right) \cdot \zeta \cdot |E(S, \ov{S})| \\
< &~ \left(1 - \frac{4\balloss + 3c\lambda}{\phi}\right) \cdot \zeta \cdot |E(S, \ov{S})|.
\end{align*}
This inequality contradicts with \Cref{eq:v_unifomity_one_step_2}, so we have proved $|E(S,T)| \geq \frac{\phi}{2} \cdot \sum_{i \in S} (d_{\gbar})_i$. 

Finally note that by definition $S' = S \cup T$, and $|E(S,T)| \leq \sum_{i \in T} (d_{\gbar})_i$, so we have %
\begin{equation*}
\sum_{i \in S'} (d_{\gbar})_i = \sum_{i \in S} (d_{\gbar})_i + \sum_{i \in T} (d_{\gbar})_i \geq \sum_{i \in S} (d_{\gbar})_i + |E(S,T)| \geq \left(1 + \frac{\phi}{2}\right) \cdot \sum_{i \in S} (d_{\gbar})_i. \qedhere
\end{equation*}
\end{proof}

Note that \Cref{lem:v_unifomity_one_step} only holds for sets $S \subseteq V$ where $\sum_{i \in S}(d_{\gbar})_i \leq \frac{1}{2}\sum_{i \in V}(d_{\gbar})_i$. For the rest of the vertices, the following lemma provides a lower bound on how much the volume of the set $S_{\leq \zeta}(z)$ increases when $\zeta$ is increased.

\begin{lemma}[Sweep cut bound of $S_{\leq \zeta}(z)$ with increasing threshold]
\label{lem:v_unifomity_one_step_increase}
In the setting of \Cref{thm:unifomity_v}, for any $\zeta > 0$ such that $\sum_{i \in S_{\leq \zeta}(z)}(d_{\gbar})_i \leq \frac{1}{2}\sum_{i \in V}(d_{\gbar})_i$,
\[
\sum_{i \in S_{\leq \zeta'}(z)} (d_{\gbar})_i \geq \left(1 + \frac{\phi}{2}\right) \cdot \sum_{i \in S_{\leq \zeta}(z)} (d_{\gbar})_i~~ \text{for }\zeta' \defeq \left(1 + \frac{10\balloss}{\phi}\right) \zeta.
\]

\end{lemma}
\begin{proof}
In this proof we define $S \defeq S_{\leq \zeta}(z)$ and $S' \defeq S_{\leq \zeta'}(z)$. %
Since $v$ is the eigenvector corresponding to the smallest eigenvalue $\lambda$ of $\md^{-1/2}\mL_G\md^{-1/2}$, we have that $\lambda v = \md^{-1/2}\mL_G\md^{-1/2}v$. So we have that $\lambda \md z = \mL_Gz$. By considering the vertices in $S$ for the previous equation, we get:
\begin{align*}
    \lambda \cdot \sum_{i \in S} d_i \cdot z_i = &~ \sum_{i \in S} (\mL_G\cdot z)_i\\
    =&~\sum_{i \in S}\left(\sum_{e =(a,i)\in E}(z_i-\eta_ez_a) + \sum_{e =(i,a)\in E}(-\eta_ez_a+\eta_e^2z_i)\right)\\
    \leq&~\sum_{i \in S}\left((d_G)_i\cdot z_i - \sum_{e = (i,a) \text{ or }(a,i) \in E }z_a\right),
\end{align*}
where the second step follows from $\mL_G = \mb_G^\top \mb_G$, and the row $e=(a,b)$ of $\mb_G$ is $\indicVec{b}-\eta_e\indicVec{a}$, and the third step follows since $\eta_e \geq 1$ and all $z_a \geq 0$ by \Cref{fact:v_min_positive}. Rearranging, we get
\begin{align*}
\sum_{i \in S}\Big((d_G)_i - \lambda d_i\Big) \cdot z_i \geq \sum_{i \in S}\sum_{e = (a,i) \text{ or }(i,a) \in E}z_a.
\end{align*}
Since $d_G \leq d \leq c\cdot d_G$, we have that $(d_G)_i - \lambda d_i \leq (1-\lambda) (d_G)_i$, and combining with the above equation we have
\begin{align*}
    (1 - \lambda)\sum_{i \in S}(d_G)_i \cdot z_i &\geq \sum_{i \in S}\sum_{e = (a,i) \text{ or }(i,a) \in E }z_a\\
    &= \sum_{i \in S}|E(i,S)|z_i + \sum_{i \in S} \sum_{e = (i,a) \text{ or } (a,i) \in E(S, \ov{S})} z_a,
\end{align*}
where in the second step we define $\ov{S} \defeq E \backslash S$ and split the last summation over edges from $S$ to $\overline{S}$ and edges from $S$ to $S$. 
Rearranging terms yields that
\begin{align*}
    (1 - \lambda) \left(\sum_{i \in S}(d_G)_i \cdot z_i - \sum_{i \in S}|E(i,S)|z_i \right) &\geq \lambda \sum_{i \in S}|E(i,S)|z_i + \sum_{i \in S} \sum_{e = (i,a) \text{ or } (a,i) \in E(S, \ov{S})} z_a.
\end{align*}
Since for every $i \in S$, $(d_G)_i \geq |E(i,S)|$ and $0 \leq z_i \leq \zeta$, the above inequality implies that
\begin{align*}
(1 - \lambda) \left(\sum_{i \in S}(d_G)_i - \sum_{i \in S}|E(i,S)| \right) \zeta \geq \sum_{i \in S} \sum_{e = (i,a) \text{ or } (a,i) \in E(S, \ov{S})} z_a.
\end{align*}
Since $\sum_{i \in S}(d_G)_i \leq (1+\balloss)^{2} \sum_{i \in S}(d_{\gbar})_i = (1+\balloss)^{2} \cdot \big(|E(S,\ov{S})| + \sum_{i \in S}|E(i,S)|\big)$, the above inequality implies that
\begin{align}\label{eq:v_unifomity_one_step_increase_1}
\sum_{i \in S} \sum_{e = (i,a) \text{ or } (a,i) \in E(S, \ov{S})} z_a \leq &~ (1 - \lambda) \left( (1+\balloss)^{2} |E(S,\ov{S})|  + (2\balloss+\balloss^2) \sum_{i \in S}|E(i,S)| \right) \zeta \notag \\
\leq &~ \left( (1+\balloss)^{2} |E(S,\ov{S})|  + \frac{(2\balloss+\balloss^2)}{\phi} |E(S,\ov{S})| \right) \zeta \notag \\
\leq &~ \left(1 + \frac{4\balloss}{\phi}\right) |E(S,\ov{S})| \zeta,
\end{align}
where the second step follows from $\gbar$ has expansion $\phi$, which means $|E(S,\ov{S})| \geq \phi \sum_{i \in S}(d_{\gbar})_i$, and so $\sum_{i \in S} |E(i,S)| \leq \sum_{i\in S}(d_{\gbar})_i \leq \frac{1}{\phi} |E(S,\ov{S})|$.

Let $T \subseteq \ov{S}$ denote the set of vertices such that $z_a \leq \zeta' = (1 + \frac{10\balloss}{\phi}) \cdot \zeta$. \Cref{eq:v_unifomity_one_step_increase_1} implies that we must have $|E(S,T)| \geq \frac{\phi}{2} \cdot \sum_{i \in S} (d_{\gbar})_i$, and next we prove this claim by contradiction. Assume that we instead have $|E(S,T)| < \frac{\phi}{2} \cdot \sum_{i \in S} (d_{\gbar})_i$. First note that since any $a \notin S$ satisfies $z_a \geq \zeta$ and any $a \notin S \cup T$ satisfies $z_a \geq (1 + \frac{10\balloss}{\phi}) \cdot \zeta$, 
\begin{align*}
\sum_{i \in S} \sum_{e = (i,a) \text{ or } (a,i) \in E(S, \ov{S})} z_a = &~ \sum_{i \in S} \sum_{e = (i,a) \text{ or } (a,i) \in E(S, T)} z_a + \sum_{i \in S} \sum_{e = (i,a) \text{ or } (a,i) \in E(S, \ov{S}\backslash T)} z_a \\
\geq &~ \zeta \cdot |E(S,T)| + \left(1 + \frac{10 \balloss}{\phi}\right) \cdot \zeta \cdot |E(S, \ov{S}\backslash T)| \\
= &~ -\frac{10 \balloss}{\phi} \cdot \zeta \cdot |E(S, T)| +  \left(1 + \frac{10 \balloss}{\phi}\right) \cdot \zeta \cdot |E(S, \ov{S})|,
\end{align*}
where in the third step we used that $|E(S, \ov{S}\backslash T)| = |E(S, \ov{S})| - |E(S, T)|$. Using the assumption that $|E(S,T)| < \frac{\phi}{2} \cdot \sum_{i \in S} (d_{\gbar})_i$, and since $|E(S,\ov{S})| \geq \phi \sum_{i \in S}(d_{\gbar})_i$, the above equation gives
\begin{align*}
\sum_{i \in S} \sum_{e = (i,a) \text{ or } (a,i) \in E(S, \ov{S})} z_a 
> &~ -\frac{10\balloss}{2} \cdot \zeta \cdot \sum_{i \in S} (d_{\gbar})_i + \left(1 + \frac{10\balloss}{\phi}\right) \cdot \zeta \cdot |E(S, \ov{S})| \\
\geq &~ \left( 1 + \frac{10 \balloss}{\phi} - \frac{10\balloss}{2 \phi} \right) \cdot \zeta \cdot |E(S, \ov{S})| \\
\geq &~ \left(1 + \frac{5\balloss}{\phi}\right) \cdot \zeta \cdot |E(S, \ov{S})|.
\end{align*}
This inequality contradicts with \Cref{eq:v_unifomity_one_step_increase_1}, so we have proved $|E(S,T)| \geq \frac{\phi}{2} \cdot \sum_{i \in S} (d_{\gbar})_i$. 

Finally note that by definition $S' = S \cup T$, and $|E(S,T)| \leq \sum_{i \in T} (d_{\gbar})_i$, so we have
\begin{equation*}
\sum_{i \in S'} (d_{\gbar})_i = \sum_{i \in S} (d_{\gbar})_i + \sum_{i \in T} (d_{\gbar})_i \geq \sum_{i \in S} (d_{\gbar})_i + |E(S,T)| \geq \left(1 + \frac{\phi}{2}\right) \cdot \sum_{i \in S} (d_{\gbar})_i. \qedhere
\end{equation*}
\end{proof}

When $S$ is the set of vertices with $z$ value greater than $\zeta$, we have shown in \Cref{lem:v_unifomity_one_step} that a large fraction of the %
neighboring vertices of $S$ has $z$ value that is only slightly smaller than $\zeta$. Similarly we have also shown in \Cref{lem:v_unifomity_one_step_increase} that when $S$ is the set of vertices with $z$ value less than $\zeta$, a large fraction of the neighboring vertices of $S$ has $z$ value that is only slightly larger than $\zeta$. Using these two lemmas, together with the fact that the underlying graph is an expander, we can prove the bound on $z$ over the whole graph.

\begin{proof}[Proof of \Cref{thm:unifomity_v}]
Initially let $\zeta_0 \defeq \max_{i} z_i$ and let $S_0 \defeq S_{\geq \zeta_0}(z)$, and note that $|S_0| \geq 1$. For any integer $k\geq 1$, let $\zeta_k \defeq (1 - \frac{10(\balloss + c\lambda)}{\phi}) \zeta_{k-1}$, and let $S_k \defeq S_{\geq \zeta_k}(z)$. For any integer $k \geq 0$ such that $(1 - \frac{10(\balloss + c\lambda)}{\phi})^{k} \geq \frac{1}{2}$ and $\sum_{i\in S_k} (d_{\gbar})_i \leq \frac{1}{2} \sum_{i \in V} (d_{\gbar})_i$, any $k' < k$ also satisfies these two conditions, so applying \Cref{lem:v_unifomity_one_step} for $k+1$ times yields that
\begin{align*}
\sum_{i \in S_{k+1}} (d_{\gbar})_i \geq \left(1 + \frac{\phi}{2}\right) \cdot \sum_{i \in S_{k}} (d_{\gbar})_i \geq \cdots \geq \left(1 + \frac{\phi}{2}\right)^{k+1} \sum_{i \in S_{0}} (d_{\gbar})_i \geq \left(1 + \frac{\phi}{2}\right)^{k+1}.
\end{align*}
Let $k^*$ denote the largest integer such that $(1 - \frac{10(\balloss + c\lambda)}{\phi})^{k^*} \geq \frac{1}{2}$ and $\sum_{i\in S_{k^*}} (d_{\gbar})_i \leq \frac{1}{2} \sum_{i \in V} (d_{\gbar})_i$. Then the above inequality implies that $\left(1 + \frac{\phi}{2}\right)^{k^*} \leq \sum_{i \in S_{k^*}} (d_{\gbar})_i \leq m$, which gives that $k^* \leq \frac{2 \log m}{\phi}$. Also note that $(1 - \frac{10(\balloss + c\lambda)}{\phi})^{k^*} \geq \frac{1}{2}$ must be true because of the assumption $(\balloss + c\lambda)\log{m} \leq \frac{1}{100}\phi^2$. So by definition the set $S_{k^*+1}$ satisfies $\sum_{i\in S_{k^*+1}} (d_{\gbar})_i > \frac{1}{2} \sum_{i \in V} (d_{\gbar})_i$, and all vertices $i \in S_{k^*+1}$ satisfy
\begin{align}\label{eq:unifomity_v_step_1}
\frac{v_i}{\sqrt{d_i}} \geq \zeta_{k^*+1} \geq \left(1 - \frac{10(\balloss + c\lambda)}{\phi}\right)^{2 \log m / \phi + 1} \cdot \zeta_0 \geq \exp\left(-O\left(\frac{(\balloss + c\lambda) \log m}{\phi^2}\right)\right) \cdot \max_{i'} \frac{v_{i'}}{\sqrt{d_{i'}}}.
\end{align}

Similarly, we let $\zeta'_0 \defeq \min_{i} z_i$, and $S'_0 \defeq S_{\leq \zeta'_0}(z)$. Note that $\zeta'_0 > 0$ by \Cref{fact:v_min_positive}, and $|S'_0| \geq 1$. For any integer $k \geq 1$, let $\zeta'_k \defeq (1+\frac{10\balloss}{\phi}) \zeta'_{k-1}$, and $S'_k \defeq S_{\leq\zeta'_k}(z)$. For any integer $k \geq 0$ such that $\sum_{i \in S'_k} (d_{\gbar})_i \leq \frac{1}{2} \sum_{i \in V} (d_{\gbar})_i$, applying \Cref{lem:v_unifomity_one_step_increase} for $k+1$ times yields that 
\[
\sum_{i \in S'_{k+1}} (d_{\gbar})_i \geq (1+\frac{\phi}{2})^{k+1}.
\]
Let $k'^*$ denote the largest interger such that $\sum_{i \in S'_{k'^*}} (d_{\gbar})_i \leq \frac{1}{2} \sum_{i \in V} (d_{\gbar})_i$. Then the above inequality implies that $\left(1 + \frac{\phi}{2}\right)^{k'^*} \leq \sum_{i \in S'_{k'^*}} (d_{\gbar})_i \leq m$, which gives that $k'^* \leq \frac{2 \log m}{\phi}$. So by definition the set $S'_{k'^*+1}$ satisfies $\sum_{i\in S'_{k'^*+1}} (d_{\gbar})_i > \frac{1}{2} \sum_{i \in V} (d_{\gbar})_i$, and all vertices $i \in S'_{k'^*+1}$ satisfy
\begin{align}\label{eq:unifomity_v_step_2}
\frac{v_i}{\sqrt{d_i}} \leq \zeta'_{k'^*+1} \leq \left(1 + \frac{10\balloss}{\phi}\right)^{2 \log m / \phi + 1} \cdot \zeta'_0 \leq \exp\left(O\left(\frac{\balloss \log m}{\phi^2}\right)\right) \cdot \min_{i'} \frac{v_{i'}}{\sqrt{d_{i'}}},
\end{align}

Finally since $\sum_{i\in S_{k^*+1}} (d_{\gbar})_i > \frac{1}{2} \sum_{i \in V} (d_{\gbar})_i$ and $\sum_{i\in S'_{k'^*+1}} (d_{\gbar})_i > \frac{1}{2} \sum_{i \in V} (d_{\gbar})_i$, we have that $S_{k^*+1} \cap S'_{k'^*+1} \neq \emptyset$. Combining \Cref{eq:unifomity_v_step_1} and \Cref{eq:unifomity_v_step_2} for $i \in S_{k^*+1} \cap S'_{k'^*+1}$ gives us the claimed bound of this theorem.
\end{proof}

\section{Heavy Hitters for Balanced Lossy  Expanders}\label{sec:heavy_hitters_balanced_expanders}

In this section we present our heavy hitter data structure for lossy graphs that are $\balloss$-balanced $\phi$-expanders undergoing edge and vertex deletions, under the assumption that the degree of each vertex remains within a $1/9$ fraction of its original degree. 
The pseodocode for the data structure is given in Algorithm~\ref{alg:heavy_hitter_expander} and \ref{alg:heavy_hitter_expander_continued}. This section proves \Cref{thm:heavy_hitter_expander} below  which analyzes the pseudocode and thereby essentially provides the formal version of \Cref{sec:tech_heavy_hitter_balanced_expander} from the technical overview and serves as a key building block towards our heavy hitter data structure for general two-sparse matrices, which we will present in the next section.

Before analyzing  \Cref{thm:heavy_hitter_expander}, we first provide an overview of  $\textsc{QueryHeavy}$, a key operation of this data structure (Line~\ref{algline:query_heavy_start} to \ref{algline:query_heavy_end}), which find entries $e$ of $\mb h$ such that $(\mb h)_e \geq \epsilon$ for a query $h$. All other operations supported by the data structure are either straightforward or easily implementable using the variables maintained to efficiently answer these heavy hitter queries.

Finding the heavy hitters of $\mb h$ is equivalent to finding the heavy hitters of $\mb \md^{-1/2} \md^{1/2} h$, so $\textsc{QueryHeavy}$ finds the heavy hitters of $\mb\md^{-1/2} g$ for $g = \md^{1/2}h$. 
As discussed in the technical overview, the key technical part of the algorithm is to maintain an approximate vector $v$ to the smallest eigenvector $v_{\min}$ of the normalized lossy Laplacian $\mntilde$. Using $v$, the algorithm decomposes $g = \gv + \gpv$ where $\gv = v v^{\top} g$ and $\gpv = (\mI - v v^{\top}) g$, and then finds the heavy hitters of $\mb\md^{-1/2} \gv$ and $\mb\md^{-1/2} \gpv$ separately (see Line~\ref{algline:large_entries_gv} and Line~\ref{algline:large_entries_gpv}). 
By \Cref{thm:spectral_sparsification}, a unit vector $v$ is a sufficiently good approximation to $v_{\min}$ if $v^\top \mntilde v \leq 10\lambda_1(\mntilde)$. In each heavy hitter query, we check whether this condition still holds by testing the necessary condition $\|\gpv\|_2^2 \leq \wt{O}(g^{\top} \mntilde g)$ (see Line~\ref{algline:if_JL}), where the norms can be computed efficiently using a Johnson-Lindenstrauss sketch. %
If the condition no longer holds, then we recompute $v$. 
One final detail is that the Rayleigh quotient $v^\top \mntilde v$ may increase when $v$ is renormalized after vertex deletions (see Line~\ref{algline:vertex_deletion_renormalization}). Crucially, by \Cref{thm:unifomity_v}, $\frac{v}{\sqrt{d}}$ is approximately uniform, so as long as we do not delete too many edges, this renormalization cannot increase the Rayleigh quotient by too much. %

\begin{theorem}[Heavy hitters on balanced lossy expanders]\label{thm:heavy_hitter_expander}
There is a data structure (Algorithm~\ref{alg:heavy_hitter_expander}) that supports the following operations w.h.p.~against adaptive inputs:
\begin{itemize}
\item $\textsc{Initialize}(G=(V,E,\eta), \epsad \in (0,1), \ov{\tau} \in \R_{\geq 0}^E)$: Initializes with a lossy $\balloss$-balanced $\phi$-expander $G=(V,E,\eta)$ where $|V|=n$, $|E|=m$, and $\balloss \leq \frac{\phi^2}{10^5 \log^2(m)}$, a parameter $0 \leq \epsad\leq \frac{\phi^2}{10^5 \log^2(m)}$, and a vector $\ov{\tau} \geq \allzero$ in amortized $O(m\phi^{-2}\log^2(n / \epsad)\log(m))$ time. 
\item $\textsc{Delete}(F \subseteq E)$: %
Delete the edges in $F$ from $G$ in amortized $O(|F| \log m)$ time, also remove any vertex whose degree drops to $0$,  provided that the graph after deletion is still a $\phi$-expander and the degree of any vertex $v$ remains greater than $1/9$ of the original initialized degree. %
\item \textsc{ScaleTau}$(e \in E, b\in \R_{\geq 0})$: Sets $\ov{\tau}_e \leftarrow b$ in worst-case $O(1)$ time.
\item $\textsc{QueryHeavy}(h \in \R^V, \epsilon \in (0,1))$: Returns a set $I \subseteq E$ containing exactly all $e \in E$ that satisfies $|\mb_G h|_e \geq \epsilon$ in amortized $O\left(\phi^{-4} \epsilon^{-2} \|\mb_G h\|_2^2 + \phi^{-4} \epsilon^{-2} \epsad \|\md_G^{1/2}h\|_2^2 + n\right)$ time. 
\item \textsc{Norm}$(h \in \R^V)$: Returns $L \in \R_{\geq 0}$ in amortized $O(n)$ time such that 
\begin{align*}
\|\mb_G h\|_2^2 \leq L \leq O\left(\phi^{-4} \|\mb_G h\|_2^2 + \phi^{-4} \epsad \|\md_G^{1/2}h\|_2^2\right). %
\end{align*}
\item \textsc{Sample}$(h \in \R^V, C_0, \ov{C}_1, \ov{C}_2, C_3)$: Let a vector $p \in \R^E$ satisfy
\[
p_e \geq \min \left\{1, ~ \ov{C_1} \cdot (\mb_G h)_e^2 + \ov{C}_2 + C_3 \ov{\tau}_e \right\}.
\]
Let $S = \sum_{e \in E} p_e$.  Let $X$ be a random variable which equals to $p_e^{-1} \indicVec{e}$ with probability $p_e / S$ for all $e \in E$. This operation returns a random diagonal matrix $\mr = C_0^{-1} \sum_{j=1}^{C_0 S} \mdiag(X_j)$, where $X_j$ are i.i.d.~copies of $X$. The amortized time of this operation and also the output size of $\mr$ are bounded by 
\[
O\left(C_0 \ov{C}_1 \phi^{-4} \log m(\|\mb_G h\|_2^2 + \epsad \|\md_G^{1/2}h\|_2^2) + C_0 \ov{C}_2 m \log m + C_0 C_3 \|\ov{\tau}\|_1 \log m + n \log n \right).
\]
\end{itemize}
\end{theorem}

\begin{algorithm}[ht!]
\caption{Heavy hitter on $\balloss$-balanced lossy $\phi$-expander}
\label{alg:heavy_hitter_expander}
\DontPrintSemicolon
\SetKw{Global}{Global:}
\SetKwInOut{Input}{Input}
\SetKwInOut{Output}{Output}
\SetKwFunction{Initialize}{Initialize}
\SetKwFunction{QueryHeavy}{QueryHeavy}
\SetKwFunction{Delete}{Delete}
\SetKwFunction{ComputeEigenvector}{ComputeEigenvector}
\SetKwFunction{Reset}{Reset}
\SetKwFunction{ScaleTau}{ScaleTau}
\SetKwFunction{Norm}{Norm}
\SetKwFunction{Sample}{Sample}
\SetKwProg{Fn}{procedure}{:}{}

\Global $v \in \R^V$, $u \in \R^E$, $\ov{\tau} \in \R^E_{\geq 0}$, an array $\pi_u$ of size $|E|$, a diagonal matrix $\md \in \R^{V \times V}$, $\mj \in \R^{k \times E}$, $\mm \in \R^{k \times V}$ \tcp*{$k = O(\log|E|)$ chosen at initialization}

\Fn{\Initialize{$G=(V,E,\eta), \epsad \in (0,1), \ov{\tau} \in \R_{\geq 0}^E$}}{%
    $\md \leftarrow \md_{\gbar}$, $\ov{\tau} \leftarrow \ov{\tau}$ \; %
    $\textsc{Reset}()$
}

\Fn{\QueryHeavy{$h \in \mathbb{R}^V, \epsilon \in (0,1)$}}{\label{algline:query_heavy_start}
    $g \leftarrow \md^{1/2} h$,  $\gv \leftarrow v v^{\top} g$, and $\gpv \leftarrow (\mI - v v^{\top}) g$ \; \label{algline:first_line_queryheavy}
    Compute $t = \phi^{-4} \cdot \left(\|\mm h\|_2^2 + \epsad \|\md^{1/2}h\|_2^2\right)$ \; \label{algline:t_JL}
    \lIf{$\|\gpv\|_2^2 \geq 5\cdot10^6 t$} {\label{algline:if_JL} %
        $\textsc{Reset}()$ \label{algline:recompute_v}
    }
    Compute $I \leftarrow S_{\geq \frac{\epsilon}{2}}(\|\gv\|_2 \cdot |u|)$ by binary search on $\pi_u$  \label{algline:large_entries_gv} \\
    \tcp*{Compute heavy entries of $\mb_G \md^{-1/2} \gv = \|\gv\|_2 \cdot u$} 
    \For{$j \in V$}{\label{algline:for_loop_vertices}
        \If{$|{d_j^{-1/2}\gpv_j}| \geq \frac{\epsilon}{6}$}{\label{algline:large_entries_gpv}
            Add all edges $e \in E$ that are adjacent to vertex $j$ to $I$ \\
            \tcp*{Compute heavy entries of $\mb_G \md^{-1/2} \gpv$}
        }
    }
    Remove from $I$ all edges $e$ that doesn't satisfy $|\mb_G h|_e \geq \epsilon$. \;
    \Return $I$\label{algline:query_heavy_end}
}

\Fn{\Delete{$F \subseteq E$}}{
    $\mm \leftarrow \mm - \mj_{:,F} (\mb_G)_{F,:}$\;
    Delete the edges in $F$ from the graph $G$, the vector $u$, and the list $\pi_u$. \;
    Let $S \subset V$ be the set of vertices adjacent to edges in $F$ whose degree has dropped to $0$. Delete the vertices $S$ from the graph $G$ and the matrix $\md$. \;
    $\mm \leftarrow \mm_{:V\backslash S}$\;
    $R \leftarrow \|v_{V\backslash S}\|_2$, $v \leftarrow \frac{v_{V\backslash S}}{R}$, $u \leftarrow \frac{u}{R}$ \;\label{algline:vertex_deletion_renormalization}
    \lIf{number of edges in $G$ decreased by a factor of 2 since the last reset}{\label{algline:edge_decrease_by_2}
        $\textsc{Reset}()$
    }
}

\Fn{\Reset{}}{
\tcp*{Reset $v$ to a close approximation of $v_{\min}$, generate a new JL matrix $\mj$, and update $u$, $\pi_u$, $\mm$ accordingly}
    $v \leftarrow \textsc{PowerIteration}\big(4 \mI - (\md^{-1/2}\mb_G^\top \mb_G\md^{-1/2}+\epsad \mI),\epsad/4\big)$ \tcp*{\Cref{thm:power_method}} \label{algline:power_iteration}
    $u \leftarrow \mb_G\md^{-1/2} \cdot v$ \;
    Compute an array $\pi_u$ of indices such that $|u_{\pi_u(1)}| \geq |u_{\pi_u(2)}| \geq \cdots \geq |u_{\pi_u(|E|)}|$ \;
    $\mj \leftarrow \textsc{JL}(|V|,\poly(|E|),0.01,|V|^{-10}) \in \R^{k\times E}$ where $k = O(\log |E|)$ \tcp*{\Cref{lem:jl}} 
    $\mm \leftarrow \mj \mb_G$\;
}

\end{algorithm}

\begin{algorithm}[ht!]
\caption{Heavy hitter on $\balloss$-balanced lossy $\phi$-expander (Algorithm ~\ref{alg:heavy_hitter_expander} continued)}
\label{alg:heavy_hitter_expander_continued}
\DontPrintSemicolon
\SetKwInOut{Input}{Input}
\SetKwInOut{Output}{Output}
\SetKwFunction{Initialize}{Initialize}
\SetKwFunction{QueryHeavy}{QueryHeavy}
\SetKwFunction{Delete}{Delete}
\SetKwFunction{ComputeEigenvector}{ComputeEigenvector}
\SetKwFunction{Reset}{Reset}
\SetKwFunction{Norm}{Norm}
\SetKwFunction{Sample}{Sample}
\SetKwFunction{ScaleTau}{ScaleTau}
\SetKwProg{Fn}{procedure}{:}{}
\textbf{procedure} \ScaleTau{$e \in E, b \in\R_{\geq 0}$}\textbf{:} {$\ov{\tau}_e \leftarrow b$ }

\Fn{\Norm{$h \in \R^V$}}{
    Execute Line~\ref{algline:first_line_queryheavy}  to Line~\ref{algline:recompute_v} of $\textsc{QueryHeavy}$ to update $v$ if necessary, and to compute vectors $g, \gv, \gpv$. \;
    \Return $10 (\|\gv\|_2^2 \cdot \|u\|_2^2 + \|\gpv\|_2^2)$ %
}

\Fn{\Sample{$h \in \R^V, C_0, \ov{C}_1, \ov{C}_2, C_3$}}{
    Execute Line~\ref{algline:first_line_queryheavy}  to Line~\ref{algline:recompute_v} of $\textsc{QueryHeavy}$ to update $v$ if necessary, and to compute vectors $g, \gv, \gpv$. \;
    Implicitly define $p_e = \min\left\{1,~5\ov{C}_1 \big(\|\gv\|_2^2 u_e^2 + \frac{(\gpv_{a_e})^2}{d_{a_e}} + \frac{(\gpv_{b_e})^2}{d_{b_e}}\big) + \ov{C}_2 + C_3 \ov{\tau}_e \right\}$, $\forall e \in E$ \;
    $S_1 \leftarrow 5\ov{C}_1 \|\gv\|_2^2 \|u\|_2^2$, $S_2 \leftarrow \sum_{i \in V}\frac{5 \ov{C}_1 (\gpv_i)^2 \cdot (d_{\gbar})_i}{d_i}$, $S_3 \leftarrow \ov{C}_2 m + C_3 \|\ov{\tau}\|_1$, $S \leftarrow S_1 + S_2 + S_3$\;\label{algline:def_S_value}
    \For{$j \in [C_0 S]$}{
    Sample a uniformly random number $r \in [0,1]$.\;
    \If{$r \leq \frac{S_1}{S}$}{
    Let $x_j \leftarrow p_e^{-1} \indicVec{e}$ with probability $\frac{5\ov{C}_1 \|\gv\|_2^2 u_e^2}{S_1}$ for each $e \in E$.\;\label{algline:sample_by_u}
    }
    \If{$\frac{S_1}{S} < r \leq \frac{S_1+S_2}{S}$}{
    Sample a vertex $i \in V$ with probability $\frac{5 \ov{C}_1 (\gpv_i)^2 \cdot (d_{\gbar})_i}{d_i \cdot S_2}$, then uniformly sample an edge $e$ adjacent to $i$, and let $x_j \leftarrow p_e^{-1} \indicVec{e}$. \;\label{algline:sample_by_vertex}
    }
    \If{$\frac{S_1+S_2}{S} < r \leq 1$}{
    Let $x_j \leftarrow p_e^{-1} \indicVec{e}$ with probability $\frac{\ov{C}_2 + C_3  \ov{\tau}_e}{S_3}$ for each $e \in E$.\; \label{algline:sample_by_tau}
    }
    }
    \Return $\mr = C_0^{-1} \sum_{j=1}^{C_0 S} \mdiag(x_j)$
}
\end{algorithm}

To prove \Cref{thm:heavy_hitter_expander}, we use the following tools:
\paragraph{Cheeger's inequality.} We use Cheeger's inequality which relates the conductance of a graph and the second smallest eigenvalue of its normalized Laplacian. 
\begin{theorem}[Cheeger's Inequality]\label{thm:cheegers}
Let $G = (V,E,w)$ be a weighted graph, let $\mn_G$ be its normalized Laplacian, and let $\phi(G)$ be the conductance of $G$. Then, 
\begin{align*}
    \frac{\phi(G)^2}{2} \leq \lambda_2(\mn_G) \leq 2\cdot\phi(G).
\end{align*}
\end{theorem}
Using Cheeger's inequality, the lossy $\balloss$-balanced $\phi$-expander $G$ of \Cref{thm:heavy_hitter_expander} satisfies $\lambda_2(\mn_{\gbar}) \geq \frac{\phi^2}{2} > 20\balloss$, allowing us to use the spectral result \Cref{thm:spectral_sparsification}.

\paragraph{JL estimate.}
We use the Johnson-Lindenstrauss Lemma that allows us to approximate the $\ell_2$ norm of a vector via a lower-dimensional embedding:
\begin{lemma}[Johnson-Lindenstrauss Lemma \cite{jl84}]\label{lem:jl} 
There exists a function $\textsc{JL}(n,m,\epsilon,\delta)$ that returns a random matrix $\mj \in \R^{k\times n}$ where $k = O(\epsilon^{-2}\log(m/\delta))$, and $\mj$ satisfies that for any $m$-element subset $V \subset \R^n$,
\begin{align*}
    \Pr\Big[\forall v \in V, (1 - \epsilon) \|v\|_2 \leq \|\mj v\|_2 \leq (1 + \epsilon)\|v\|_2\Big] \geq 1 - \delta. 
\end{align*}
Furthermore, the function $\textsc{JL}$ runs in $O(kn)$ time. 
\end{lemma}

\paragraph{Power Method.} We also use the standard power iteration to compute the maximum eigenvector of a matrix, as defined in \Cref{thm:power_method}. The guarantees of applying \textsc{PowerIteration} to our desired matrix $\md^{-1/2}\mL_G \md^{-1/2}+\epsad \mI$ are given in \Cref{thm:power_iteration_least_eigenvalue}.

\begin{proof}[Proof of Theorem~\ref{thm:heavy_hitter_expander}]
\textbf{Correctness of JL estimates against adaptive inputs:} 
We first prove that w.h.p.~in any execution of $\textsc{QueryHeavy}(\cdot)$, $\|\mm h\|_2 \approx_{1.01} \|\mb_G h\|_2$. Since the IPM makes at most $\poly(m)$ calls to $\textsc{QueryHeavy}$, the Johnson-Lindenstrauss lemma (\Cref{lem:jl}) implies that $\|\mm h\|_2 \approx_{1.01} \|\mb_G h\|_2$ for all queries with high probability, provided the query vectors $h$ are oblivious to the JL matrix $\mj$. 

Next we prove that the query vectors $h$ are indeed oblivious to $\mj$. We consider an auxiliary algorithm \textsc{AlgSlow} that computes $t' = 2 \phi^{-4} \cdot \left(\|\mb_G h\|_2^2 + \epsad \|\md^{1/2}h\|_2^2\right)$ on Line~\ref{algline:t_JL}, and uses $t'$ instead of $t$ on Line~\ref{algline:if_JL} to determine if $\|\gpv\|_2^2 \geq 5\cdot10^6 t'$. The outputs of Algorithm~\ref{alg:heavy_hitter_expander} and \textsc{AlgSlow} are exactly the same until one of them enters the if-clause on Line~\ref{algline:if_JL}. Conditioned on the high probability event that $\|\mm h\|_2 \approx_{1.01} \|\mb_G h\|_2$ for all previous queries, whenever \textsc{AlgSlow} enters the if-clause on Line~\ref{algline:if_JL}, Algorithm~\ref{alg:heavy_hitter_expander} also does so, and Algorithm~\ref{alg:heavy_hitter_expander} generates a new JL matrix $\mj$ whenever this happens. Consequently, with high probability the output of Algorithm~\ref{alg:heavy_hitter_expander} is independent of the current JL matrix $\mj$.

For the remainder of the proof, we condition on the high probability event that $\|\mm h\|_2 \approx_{1.01} \|\mb_G h\|_2$ in any execution of $\textsc{QueryHeavy}(\cdot)$.

\textbf{Time complexity of \textsc{Initialize} and \textsc{Reset}:} We first bound the running time of one call to $\textsc{Reset}(\cdot)$. The most time-consuming step in $\textsc{Reset}(\cdot)$ is the call to $\textsc{PowerIteration}(\cdot)$, which by \Cref{thm:power_iteration_least_eigenvalue} takes $O(m\phi^{-2}\log{(\frac{n}{\epsad})})$ time since $c = 9$, and because $\textsc{Delete}(\cdot)$ guarantees the degree of any vertex $v$ remains greater than $1/9$ of the original initialized degree or drops to $0$. All other steps, including computing $\mm = \mj \mb_G$, and sorting the list $\pi_u$, can all be computed in $O(m \log m)$ time.

Note that apart from $\textsc{Initialize}(\cdot)$, $\textsc{Reset}(\cdot)$ is only called inside $\textsc{QueryHeavy}(\cdot)$ and $\textsc{Delete}(\cdot)$. Next we prove that during all $\textsc{QueryHeavy}(\cdot)$ and $\textsc{Delete}(\cdot)$, $\textsc{Reset}(\cdot)$ is invoked for at most $O(\log(\frac{n}{\epsad}) + \log(m))$ times, which then proves that the total runtime of all $\textsc{Reset}(\cdot)$ is bounded by $O(m\phi^{-2}\log^2(\frac{n}{\epsad}) \log(m))$ as claimed. Note that since $\textsc{Reset}(\cdot)$ is invoked inside a $\textsc{Delete}(\cdot)$ only when the number of edges in $G$ decreased by a factor of 2 since the last reset (see Line~\ref{algline:edge_decrease_by_2} of Algorithm~\ref{alg:heavy_hitter_expander}), we have that during all $\textsc{Delete}(\cdot)$, $\textsc{Reset}(\cdot)$ is invoked for at most $O(\log(m))$ times. It remains to prove that during all $\textsc{QueryHeavy}(\cdot)$, $\textsc{Reset}(\cdot)$ (see Line~\ref{algline:recompute_v} of Algorithm~\ref{alg:heavy_hitter_expander}) is invoked for at most $\log(\epsad^{-1})$ number of times.

Line~\ref{algline:recompute_v} is executed only if $\|\gpv\|_2^2 \geq 5\cdot10^6 t$, and we will show that in this case we must have 
\begin{align}\label{eq:rayleigh_vs_lambda_1}
v^\top\left(\md^{-1/2}\mL_G \md^{-1/2} + \epsad \mI\right)v \geq 10\lambda_1\left(\md^{-1/2}\mL_G \md^{-1/2} + \epsad \mI\right).
\end{align}
Suppose on the contrary that $v^\top(\md^{-1/2}\mL_G \md^{-1/2} + \epsad \mI)v < 10\lambda_1(\md^{-1/2}\mL_G \md^{-1/2} + \epsad \mI)$. This assumption satisfies the requirement on $v$ for \Cref{thm:spectral_sparsification} with $c_2 = 10$. We next show that $\md \approx_{30} \md_{G}$, which would satisfy the requirement on $d$ with $c_1 = 30$ for \Cref{thm:spectral_sparsification}. %
The guarantees of $\textsc{Delete}(\cdot)$ ensure that $\md \approx_{9} \md_{\gbar}$. Since $\balloss \leq 0.01$, \Cref{lem:eta_to_not} implies that $\md_{\gbar} \approx_{1.03} \md_G$, and so $\md \approx_{30} \md_{G}$. Finally, by Cheeger's inequality (\Cref{thm:cheegers}) and the fact that $G$ is a lossy $\balloss$-balanced $\phi$-expander with $\balloss \leq \frac{\phi^2}{10^5 \log^2(n)}$, the second smallest eigenvalue of $\mn_{\gbar}$ satisfies $\lambda_2(\mn_{\gbar}) \geq \frac{\phi^2}{2} > 20\balloss$. Hence, all requirements of \Cref{thm:spectral_sparsification} are satisfied. 
Now, applying \Cref{thm:spectral_sparsification} with $c_1 = 30, c_2 = 10$,
\begin{align*}
    \mI - (1 - \lambda)vv^\top \preceq \frac{48\cdot30^2\cdot10^2}{\phi^4} \left(\md^{-1/2}\mL_G \md^{-1/2} + \epsad \mI\right) 
    \preceq \frac{4.4\cdot10^6}{\phi^4} \left(\md^{-1/2}\mL_G \md^{-1/2} + \epsad \mI\right).
\end{align*}
Note that $\|\gpv\|_2^2 = g^{\top} \left(\mI - v v^{\top}\right) g \leq g^{\top} \left(\mI - (1 - \lambda) v v^{\top}\right) g$, so the above equation implies that
\begin{align*}
\|\gpv\|_2^2 %
\leq &~ \frac{4.4\cdot10^6}{\phi^4} g^{\top} \left(\md^{-1/2}\mL_G \md^{-1/2} + \epsad \mI\right) g \\
= &~ \frac{4.4\cdot10^6}{\phi^4} \left(\|\mb_G h\|_2^2 + \epsad \|\md^{1/2} h\|_2^2\right) \\
\leq &~ 1.01 \frac{4.4\cdot10^6}{\phi^4} \left(\|\mm h\|_2^2 + \epsad \|\md^{1/2} h\|_2^2\right) < 5\cdot10^6t,
\end{align*}
which contradicts with the condition that $\|\gpv\|_2^2 \geq 5\cdot 10^6 t$. Consequently, we have proven by contradiction that \Cref{eq:rayleigh_vs_lambda_1} holds.

Next we prove that the Rayleigh quotient $v^\top(\md^{-1/2}\mL_G \md^{-1/2} + \epsad \mI)v$ does not increase too much due to the vertex deletions that occur between two $\textsc{Reset}$s. %
Suppose the algorithm executes Line~\ref{algline:recompute_v} when running $\textsc{QueryHeavy}$ with vector $v$ and matrix $\mL_G$. Let $G^{\old}, V^{\old}, \mL_G^{\old},\md^{\old},d^{\old}$ and $v^{\old}$ be the $G, V, \mL_G,\md,d$ and $v$ at the last time when \textsc{Reset} was invoked, at which point
\[
(v^{\old})^{\top} ((\md^{\old})^{-1/2} \mL_G^{\old} (\md^{\old})^{-1/2} + \epsad \mI) v^{\old} \leq 2 \lambda_1 ((\md^{\old})^{-1/2} \mL_G^{\old} (\md^{\old})^{-1/2} + \epsad \mI)
\]
since $v^{\old}$ was recomputed to be a $2$-approximate least eigenvector by \Cref{thm:power_iteration_least_eigenvalue}. Our goal is to show that: 
\begin{align*}
    v^\top\left(\md^{-1/2}\mL_G \md^{-1/2} + \epsad \mI\right)v 
     \leq 3(v^{\old})^{\top} \left((\md^{\old})^{-1/2} \mL_G^{\old} (\md^{\old})^{-1/2} + \epsad \mI\right) v^{\old}.
\end{align*}
Let $v^{\ext}$ %
be a vector extended to the same length as $v^{\old}$, sharing the same values as $v^{\old}$ on the support of $V$, and $0$ otherwise. See Figure~\ref{fig:v_ext_v_old} for an illustration. Let $\lambda_i^{\old} \defeq \lambda_i((\md^{\old})^{-1/2} \mL_G^{\old} (\md^{\old})^{-1/2})$, and $v_i^{\old} \defeq v_i((\md^{\old})^{-1/2} \mL_G^{\old} (\md^{\old})^{-1/2})$ denote the true eigenvalues and eigenvectors. Since $v^{\old}$ is a constant factor approximation to $v^{\old}_1$ in the Rayleigh quotient, \Cref{lem:bound_on_v_topv} implies that $1 - ((v^{\old}_1)^\top v^{\old})^2 \leq \frac{10\lambda_1^{\old}}{\lambda_2^{\old}}$. 
\begin{figure}[!ht]
\centering
\begin{align*}
\underset{\substack{ \vphantom{a}\\\\\textstyle{v}}}{
\begin{array}{@{}c@{}}{
  \begin{bmatrix}
    (v^{\old})_1/{\|v^{\ext}\|_2} \\ 
    (v^{\old})_2/{\|v^{\ext}\|_2} \\
    \vdots \vphantom{\vdots}\vphantom{\vdots} \\
    (v^{\old})_{|V|}/{\|v^{\ext}\|_2}
  \end{bmatrix}}\\
  \vspace{3mm}\vphantom{(v^{\old})_{|V|+1}}    \\
  \vphantom{\vdots} \\
  \vphantom{(v^{\old})_{|V^{\old}|}}
\end{array}
}\;\;\;
\underset{\substack{ \\\\\textstyle{v^{\ext}}}}{
  \begin{bmatrix}
    (v^{\old})_1 \\
    (v^{\old})_2 \\
    \vdots \\
    (v^{\old})_{|V|} \vspace{3mm} \\ 
    0 \vphantom{(v^{\old})_{|V|+1}}    \\
    \vdots \\
    0 \vphantom{(v^{\old})_{|V^{\old}|}}
  \end{bmatrix}
  }\;\;\;\;\;
\underset{\substack{ \\\\\textstyle{v^{\old}}}}{
  \begin{bmatrix}
    (v^{\old})_1 \\ (v^{\old})_2 \\ \vdots \\ (v^{\old})_{|V|} \vspace{3mm} \\ (v^{\old})_{|V|+1} \\ \vdots \\ (v^{\old})_{|V^{\old}|}
  \end{bmatrix}
}
\begin{tabular}{l}
$\left.\lefteqn{\phantom{\begin{matrix} (v^{\old})_1 \\ (v^{\old})_2 \\ \vdots \\ (v^{\old})_{|V|} \ \end{matrix}}} \right\}V$\\
$\left.\lefteqn{\phantom{\begin{matrix} \vspace{3mm} (v^{\old})_{|V|+1} \\ \vdots \\ (v^{\old})_{|V^{\old}|} \ \end{matrix}}} \right\}V^{\old}\backslash V$
\end{tabular}
\end{align*}
\caption{Illustration of $v$, $v^{\ext}$, and $v^{\old}$.}
\label{fig:v_ext_v_old}
\end{figure}
\Cref{lem:eta_to_not} and \Cref{lem:N_approx_Ivvt} imply that $\lambda_1^{\old} \leq 10 \balloss + \epsad $, and that $\lambda_2^{\old} \geq \frac{\phi^2}{4c}$ where $c=9$. By our choices of parameters that $\balloss,\epsad \leq \frac{\phi^2}{10^5 \log^2(m)}\leq\frac{\phi^2}{400c\log{m}}$, we have that $\frac{10\lambda_1^{\old}}{\lambda_2^{\old}}\leq \frac{1}{\log{m}}$. Furthermore, since $(\balloss + 9\lambda_1^{\old})\log{m} \leq  \frac{\phi^2}{100 \log m}$, by \Cref{cor:uniformity_v_with_parameters},
\begin{align}\label{eq:unifomity_v}
    \frac{(v_{\min}^{\old})_i}{\sqrt{(d^{\old})_i}} \geq \left(1 - \frac{1}{\log{m}}\right) \cdot \max_{i'} \frac{(v_{\min}^{\old})_{i'}}{\sqrt{(d^{\old})_{i'}}}.
\end{align}
Next we bound $\|v^{\ext}\|_2 = (v^{\old})_i/(v)_i$, for any $i \in V$, the normalizing factor after deleting some subset of vertices. Let $\chi_V$ be the indicator vector that has support $V^{\old}$ and is $1$ on $V$ and $0$ otherwise. By definition, 
\begin{align*}
    \|v^{\ext}\|_2^2 &= \|{v^{\old}}\cdot\chi_{V}\|_2^2\\
    &\geq \frac{1}{2} \left\|(v_{\min}^{\old} (v_{\min}^{\old})^\top v^{\old})\cdot\chi_{V}\right\|_2^2 - \left\|(v^{\old} -v_{\min}^{\old} (v_{\min}^{\old})^\top v^{\old} )\cdot\chi_{V}\right\|_2^2\\
    &\geq \frac{1}{2} ((v_{\min}^{\old})^\top v^{\old})^2 \cdot \|v_{\min}^{\old}\cdot \chi_{V}\|_2^2 - \left(\|v^{\old}\|_2^2 + \big((v_{\min}^{\old})^\top v^{\old}\big)^2 \|v_{\min}^{\old}\|_2^2 - 2 \big((v_{\min}^{\old})^\top v^{\old}\big)^2 \right). 
    \end{align*}
We bound the $(v_{\min}^{\old})^\top v^{\old}$ term using $1-((v_{\min}^{\old})^\top v^{\old})^2 \leq \frac{10\lambda_1^{\old}}{\lambda_2^{\old}} \leq \frac{1}{\log{m}}$ that we've proved, and we bound the $\|v^{\old}\|_2$ term by $\|v^{\old}\|_2 = \|v_{\min}^{\old}\|_2 = 1$, so the above equation gives
\begin{align}\label{eq:v_ext_bound}
    \|v^{\ext}\|_2^2 %
    &\geq \frac{1}{2} \left(1 - \frac{1}{\log m}\right) \cdot \|v_{\min}^{\old}\cdot \chi_{V}\|_2^2 - \frac{1}{\log m} \notag \\
    &\geq \frac{1}{2} \left(1 - \frac{1}{\log m}\right) \cdot \frac{(1 - \frac{1}{\log m})\sum_{i \in V}(d^{\old})_i}{\sum_{i \in V^{\old}}(d^{\old})_i} - \frac{1}{\log{m}} \notag \\
    &\geq \frac{1}{3},
\end{align}
where the second line holds by \Cref{eq:unifomity_v}, and the third line holds since if we ever delete more than half the edges we reset (see Line~\ref{algline:edge_decrease_by_2}), giving us that $\sum_{i \in V} (d^{\old})_i \geq \frac{1}{2}\sum_{i \in V^{\old}} (d^{\old})_i$. 

The term that we want to bound is
\begin{align*}
v^\top(\md^{-1/2}\mL_G \md^{-1/2} + \epsad \mI)v 
     &=\epsad + \sum_{e=(j,i) \in G}\Big((d^{-1/2})_i(v)_i - \eta_e(d^{-1/2})_j(v)_j\Big)^2\\
     &=\epsad + \sum_{e = (j,i) \in G}\frac{1}{\|v^{\ext}\|_2^2}\Big((d^{-1/2})_i(v^{\old})_i - \eta_e(d^{-1/2})_j(v^{\old})_j\Big)^2.
\end{align*}
Using \Cref{eq:v_ext_bound}, the above equation gives
\begin{align*}
    v^\top(\md^{-1/2}\mL_G \md^{-1/2} + \epsad \mI)v 
     &\leq3\epsad + 3\sum_{e = (j,i) \in G}\Big((d^{-1/2})_i(v^{\old})_i - \eta_e(d^{-1/2})_j(v^{\old})_j\Big)^2\\
     &\leq 3\epsad + 3\sum_{e = (j,i) \in G^{\old}}\Big((d^{-1/2})_i(v^{\old})_i - \eta_e(d^{-1/2})_j(v^{\old})_j\Big)^2\\
     &= 3(v^{\old})^{\top} \Big((\md^{\old})^{-1/2} \mL_G^{\old} (\md^{\old})^{-1/2} + \epsad \mI\Big) v^{\old},
\end{align*}
where the second line follows from $G$ undergoing only edge deletions, and the last line follows from $d$ and $d^{\old}$ are the same on all vertices $i \in V$.

Combining with \eqref{eq:rayleigh_vs_lambda_1} that $v^\top\left(\md^{-1/2}\mL_G \md^{-1/2} + \epsad \mI\right)v \geq 10\lambda_1\left(\md^{-1/2}\mL_G \md^{-1/2} + \epsad \mI\right)$, and $(v^{\old})^{\top} \left((\md^{\old})^{-1/2} \mL_G^{\old} (\md^{\old})^{-1/2} + \epsad \mI\right) v^{\old} \leq 2 \lambda_1 \left((\md^{\old})^{-1/2} \mL_G^{\old} (\md^{\old})^{-1/2} + \epsad \mI\right)$, yields
\begin{align*}
     10\lambda_1\left(\md^{-1/2}\mL_G \md^{-1/2} + \epsad \mI\right) & \leq v^\top\left(\md^{-1/2}\mL_G \md^{-1/2} + \epsad \mI\right)v \\
     &\leq 3(v^{\old})^{\top} \left((\md^{\old})^{-1/2} \mL_G^{\old} (\md^{\old})^{-1/2} + \epsad \mI\right) v^{\old} \\
     &\leq 6 \lambda_1 \left((\md^{\old})^{-1/2} \mL_G^{\old} (\md^{\old})^{-1/2} + \epsad \mI\right).
\end{align*}

As such, we have that each time we call \textsc{Reset} in \textsc{QueryHeavy}, the smallest eigenvalue of $\md^{-1/2}\mL_G \md^{-1/2} + \epsad \mI$ must drop by at least a factor of 1.5. %
Since the smallest eigenvalue of $\md^{-1/2}\mL_G \md^{-1/2} + \epsad \mI$ is lower bounded by $\epsad$, and upper bounded by $O(1)$, we have that the smallest eigenvalue can drop for at most $O(\log(\epsad^{-1}))$ times.

\textbf{Correctness and time complexity of \textsc{Delete} and \textsc{ScaleTau}:}
The correctness of these two operations directly follow from their algorithm description.

In $\textsc{Delete}(\cdot)$, updating $\mm$ requires $O(|F| k) = O(|F| \log m)$ time, since each row of $\mb_G$ is 2-sparse. Deleting the edges in $F$ from $G$, $u$, and $\pi_u$ takes $O(|F|)$ time. 
Finally, it's straightforward to see $\textsc{ScaleTau}(\cdot)$ takes $O(1)$ time.

\textbf{Correctness of \textsc{QueryHeavy}:} If there is an entry $e$ of $\mb_G h$ such that $|\mb_G \md^{-1/2} g|_e \geq \epsilon$, then at least one of $|\mb_G \md^{-1/2} \gpv|_e$ and $|\mb_G \md^{-1/2} \gv|_e$ will be at least $\epsilon/2$. As such, it suffices to find the heavy hitters of $\mb_G \md^{-1/2}\gv$ and $\mb_G \md^{-1/2}\gpv$.
\begin{itemize}
    \item Heavy hitters of $\mb_G \md^{-1/2}\gv$: These are covered in line~\ref{algline:large_entries_gv}, which finds all the entries of $|\mb_G \md^{-1/2}\gv| = \|\gv\|_2 \cdot |u|$ that are greater than $\epsilon/2$. 
    \item Heavy hitters of $\mb_G \md^{-1/2}\gpv$: These are covered in line~\ref{algline:large_entries_gpv}. If some $e = (i,j)$ satisfies that
    \begin{align*}
        \left|d_j^{-1/2}\gpv_j - \eta_e d_i^{-1/2}\gpv_i\right| = |\mb_G \md^{-1/2}\gpv|_e \geq \epsilon/2,
    \end{align*} then since $|\eta_e | \leq 1 + \balloss \leq 1.5$, we have that either $|d_i^{-1/2}\gpv_i| \geq \epsilon/6$ or $|d_j^{-1/2}\gpv_j| \geq \epsilon/6$.
\end{itemize}

\textbf{Time complexity of \textsc{QueryHeavy}:} 
The runtime of $\textsc{QueryHeavy}(\cdot)$ has the following parts:
\begin{itemize}
\item Time to compute vectors $\gv$ and $\gpv$, and time to go over all vertices in the for loop on Line~\ref{algline:for_loop_vertices}, and these steps take $O(n)$ time in total.
\item Time to perform binary search over $\pi_u$: $O(\log m)$ time. %
\item Time to add edges into set $I$: We bound the size of the set $I$. There are two places where we add edges to $I$. In line~\ref{algline:large_entries_gv}, we add all the heavy hitters of $\mb_G\md^{-1/2}\gv$ to $I$. There can be at most $O(\|\mb_G\md^{-1/2}\gv\|^2_2\epsilon^{-2})$ such heavy hitters. In line~\ref{algline:large_entries_gpv}, we check to see if a vertex $j$ satisfies $|{d_j^{-1/2}\gpv_j}| \geq \epsilon/6$, which is exactly when $d_j^{-1} (\gpv_j)^2 \epsilon^{-2} \geq 1/36$. For each of these vertices, we have to add $d_j$ entries to $I$, so the total number of edges added to $I$ is bounded by $O(\| \gpv\|_2^2\epsilon^{-2})$. 

We first bound the latter norm. By Johnson-Lindenstrauss lemma, $\|\mm h\|_2 \approx_{1.01}\|\mb_G h\|_2$, and we have also proved that $\md \approx_{30} \md_G$, so
\begin{align}\label{eq:bound_gpv}
    \|\gpv\|_2^2 &\leq O(t) 
    \leq O\Big(\phi^{-4}\|\mb_G h\|_2^2 + \phi^{-4} \epsad \|\md_G^{1/2}h\|_2^2\Big).
\end{align}

To bound the first norm, see that since $g = \gpv + \gv$,
\begin{align}\label{eq:bound_B_gv}
    \|\mb_G\md^{-1/2}\gv\|_2^2 &\leq 2 \|\mb_G\md^{-1/2}g\|_2^2 + 2\|\mb_G\md^{-1/2}\gpv\|_2^2 \notag \\
    &\leq 2\|\mb_G h\|_2^2 + 2\|\mb_G \md^{-1/2}\|_2^2\|\gpv\|_2^2 \notag \\
    &\leq 2\|\mb_G h\|_2^2 + 8 \|\gpv\|_2^2,
\end{align}
where the last step makes use of \Cref{lem:N_approx_Ivvt}, bounding $\|\mb_G\md^{-1/2}\|_2^2$ by 4. Thus, this step takes $O(\phi^{-4}\epsilon^{-2}\|\mb_G h\|_2^2 + \phi^{-4} \epsad \epsilon^{-2}\|\md_G^{1/2}h\|_2^2)$ time. 
\end{itemize}
Adding these terms together, the total amortized time for $\textsc{QueryHeavy}(\cdot)$ is $O(\phi^{-4} \epsilon^{-2} \|\mb_G h\|_2^2 + \phi^{-4} \epsilon^{-2} \epsad \|\md_G^{1/2}h\|_2^2 + n + \log m)$.

{\bf Correctness and time complexity of \textsc{Norm}.} %
The lower bound follows from 
\begin{align*}
\|\gv\|_2^2 \cdot \|u\|_2^2 + \|\gpv\|_2^2 = &~ \|\mb_G \md^{-1/2} \gv\|_2^2 + \|\gpv\|_2^2 \\
\geq &~ \|\mb_G \md^{-1/2} \gv\|_2^2 + \frac{\|\mb_G \md^{-1} \gpv\|_2^2}{4}
\geq \frac{\|\mb_G h\|_2^2}{8},
\end{align*}
where the second step follows from $\lambda_n(\md^{-1/2} \mb_G^{\top} \mb_G \md^{-1/2}) \leq 4$ by \Cref{lem:N_approx_Ivvt}, the third step follows from $\|\mb_G \md^{-1/2} \gv\|_2^2 + \|\mb_G \md^{-1} \gpv\|_2^2 \geq (\|\mb_G \md^{-1/2} \gv\|_2 + \|\mb_G \md^{-1} \gpv\|_2)^2 / 2 \geq \|\mb_G \md^{-1} g\|_2^2 / 2 = \|\mb_G h\|_2^2 / 2$.

The upper bound follows from 
\begin{align*}
\|\gv\|_2^2 \cdot \|u\|_2^2 + \|\gpv\|_2^2 = &~ \|\mb_G \md^{-1/2} \gv\|_2^2 + \|\gpv\|_2^2 \\
\leq &~ 2 \|\mb_G h\|_2^2 + 9 \|\gpv\|_2^2 \\
\leq &~ O\Big( \phi^{-4} \|\mb_G h\|_2^2 + \phi^{-4} \epsad \|\md^{1/2}h\|_2^2 \Big),
\end{align*}
where the second step follows from Eq.~\eqref{eq:bound_B_gv} that $\|\mb_G \md^{-1/2} \gv\|_2^2 \leq 2 \|\mb_G h\|_2^2 + 8 \|\gpv\|_2^2$, and the third step follows from Eq.~\eqref{eq:bound_gpv} that $\|\gpv\|_2^2 \leq O(\phi^{-4}\|\mb_G h\|_2^2 + \phi^{-4} \epsad \|\md^{1/2}h\|_2^2)$.

Finally we bound the time complexity of $\textsc{Norm}(\cdot)$. Note that the algorithm can maintain $\|u\|_2$ whenever it re-computes $u$, and computing $\|\gv\|_2$ and $\|\gpv\|_2$ takes $O(n)$ time.

{\bf Correctness of \textsc{Sample}.} In the algorithm our goal is to sample each edge $e = (a_e,b_e)$ with probability $p_e/S$ where $p_e = \min\{1,~5 \ov{C}_1 \big(\|\gv\|_2^2 u_e^2 + \frac{(\gpv_{a_e})^2}{d_{a_e}} + \frac{(\gpv_{b_e})^2}{d_{b_e}}\big) + \ov{C}_2 + C_3 \ov{\tau}_e \}$ and $S = \sum_{e\in E} p_e$. First note that by the definitions of $S_1,S_2,S_3,S$ on Line~\ref{algline:def_S_value}, the $S$ computed by the algorithm satisfies $S = \sum_{e \in E} p_e$. We first prove that $5 (\|\gv\|_2^2 u_e^2 + \frac{(\gpv_{a_e})^2}{d_{a_e}} + \frac{(\gpv_{b_e})^2}{d_{b_e}}) \geq (\mb_G h)_e^2$, which will imply $p_e \geq \min \left\{1, ~ \ov{C_1} \cdot (\mb_G h)_e^2 + \ov{C}_2 + C_3 \ov{\tau}_e \right\}$ as required. By definitions $(\mb_G h)_e^2 = (\mb_G \md^{-1/2} \gv + \mb_G \md^{-1/2} \gpv)_e^2 \leq 2 (\mb_G \md^{-1/2} \gv)_e^2 + 2 (\mb_G \md^{-1/2} \gpv)_e^2$. Since $(\mb_G \md^{-1/2} \gv)_e = \|\gv\|_2^2 u_e^2$ and $(\mb_G \md^{-1/2} \gpv)_e = \frac{\gpv_{b_e}}{d_{b_e}^{1/2}} - \frac{\eta_{e} \gpv_{a_e}}{d_{a_e}^{1/2}}$, we have
\begin{align*}
(\mb_G h)_e^2 %
\leq &~ 2 \|\gv\|_2^2 u_e^2 + 2 (\frac{\gpv_{b_e}}{d_{b_e}^{1/2}} - \frac{\eta_{e} \gpv_{a_e}}{d_{a_e}^{1/2}})^2 \\
\leq &~ 2 \|\gv\|_2^2 u_e^2 + 4 \frac{(\gpv_{b_e})^2}{d_{b_e}} + 5 \frac{(\gpv_{a_e})^2}{d_{a_e}},
\end{align*}
where we used that $\eta_e \leq 1+\balloss < 1.1$.

Next we prove that each edge is indeed sampled with probability $p_e/S$. For each $j \in [C_0 S]$, and for each edge $e$, we compute the probability that $x_j = p_e^{-1} \indicVec{e}$ by summing the probability of the three cases of Line~\ref{algline:sample_by_u}, \ref{algline:sample_by_vertex}, and \ref{algline:sample_by_tau}:
\begin{align*}
\Pr[x_j = p_e^{-1} \indicVec{e}] 
= &~ \frac{S_1}{S} \cdot \frac{5\ov{C}_1 \|\gv\|_2^2 u_e^2}{S_1} + \frac{S_3}{S} \cdot \frac{\ov{C}_2 + C_3  \ov{\tau}_e}{S_3} \\
&~ + \frac{S_2}{S} \cdot (\frac{5 \ov{C}_1 (\gpv_{a_e})^2 (d_{\gbar})_{a_e}}{d_{a_e} \cdot S_2} \cdot \frac{1}{(d_{\gbar})_{a_e}} + \frac{5 \ov{C}_1 (\gpv_{b_e})^2 (d_{\gbar})_{b_e}}{d_{b_e} \cdot S_2} \cdot \frac{1}{(d_{\gbar})_{b_e}}) 
= \frac{p_e}{S}.
\end{align*}

{\bf Time complexity of \textsc{Sample}.} 
$\textsc{Sample}(\cdot)$ has $C_0 S$ iterations, where each iteration has the following parts:
\begin{itemize}
\item Sample each edge $e$ with probability $\frac{5\ov{C}_1 \|\gv\|_2^2 u_e^2}{S_1}$ on Line~\ref{algline:sample_by_u}: Since the algorithm maintains the vector $u$, using a binary tree that stores partial sums of the probabilities, this sampling step can be implemented in time $O(\log m)$. 
\item Sample each edge $e$ with probability $\frac{(\gpv_{a_e})^2}{d_{a_e}} + \frac{(\gpv_{b_e})^2}{d_{b_e}}$ on Line~\ref{algline:sample_by_vertex}: The algorithm first samples a vertex $i$ with probability $\frac{5 \ov{C}_1 (\gpv_i)^2 \cdot (d_{\gbar})_i}{d_i \cdot S_2}$, which can be done in $O(\log n)$ time using a binary tree that stores partial sums of the probabilities. This binary tree can be precomputed in $O(n \log n)$ time and reused for all $C_0 S$ samples. The algorithm then uniformly samples an incident edge of $i$ in $O(\log (d_{\gbar})_i)$ time. So the total amortized time of this step is $O(\log n + \frac{n \log n}{C_0 S})$.
\item Sample each edge $e$ with probability $\frac{\ov{C}_2 + C_3  \ov{\tau}_e}{S_3}$ on Line~\ref{algline:sample_by_tau}: Since the algorithm maintains the vector $\tau$, using a binary tree that stores partial sums of the probabilities, this sampling step can be implemented in time $O(\log m)$. 
\end{itemize}
Combining these three parts, we have that sampling one edge takes $O(\log m +\frac{n \log n}{C_0 S})$ time. Next we bound $S$ using Eq.~\eqref{eq:bound_gpv} and \eqref{eq:bound_B_gv}:
\begin{align*}
S = &~ \sum_e \left(5 \ov{C}_1 \big(\|\gv\|_2^2 u_e^2 + \frac{(\gpv_{a_e})^2}{d_{a_e}} + \frac{(\gpv_{b_e})^2}{d_{b_e}}\big) + \ov{C}_2 + C_3 \ov{\tau}_e \right) \\
\leq &~ O\left(\ov{C}_1 \big(\|\gv\|_2^2 \|u\|_2^2 + \|\gpv\|_2^2\big) + \ov{C}_2 m + C_3 \|\ov{\tau}\|_1 \right) \\
\leq &~ O\Big(\ov{C}_1 \phi^{-4} \|\mb_G h\|_2^2 + \ov{C}_1 \phi^{-4} \epsad \|\md^{1/2}h\|_2^2 + \ov{C}_2 m + C_3 \|\ov{\tau}\|_1 \Big). 
\end{align*}
So the total time $C_0 S \cdot O(\log m +\frac{n \log n}{C_0 S})$ is bounded as claimed in the theorem statement.
\end{proof}

\section{General Heavy Hitters
and Sampler for Two-Sparse Matrices
}\label{sec:general_heavy_hitters}

In this section we present our general heavy hitter and sampler data structure for two-sparse matrices. In \Cref{sec:other_data_structures} we show that this data structure suffices to implement the $\textsc{HeavyHitter}$ and $\textsc{HeavySampler}$ data structures required by the reduction in \cite{bll+21}. While the reduction in \cite{bll+21} only requires the ability to rescale and delete rows, we prove a stronger data structure that supports row insertions and deletions.

We begin by stating our main theorem, \Cref{thm:heavy_hitter_two_sparse_general} for two-sparse matrices. %
In the remaining sections we present the proof of \Cref{thm:heavy_hitter_general} in a bottom-up fashion. In \Cref{sec:heavy_hitter_balanced}, we show that the heavy hitter data structure on balanced lossy expanders (\Cref{thm:heavy_hitter_expander}) implies a heavy hitter data structure on all balanced lossy graphs (\Cref{thm:heavy_hitter_balanced}). In \Cref{sec:heavy_hitter_general}, we further show that the heavy hitter data structure on balanced lossy graphs (\Cref{thm:heavy_hitter_balanced}) implies a heavy hitter data structure on general lossy graphs (\Cref{thm:heavy_hitter_general}). %
Finally in \Cref{sec:heavy_hitter_reduction_two_sparse_to_lossy}, we prove our main theorem \Cref{thm:heavy_hitter_two_sparse_general} using the heavy hitter data structure general lossy graphs (\Cref{thm:heavy_hitter_general}).

\paragraph{Additional notation.}
Since all data structures in this section support update operations --- including row insertions and deletions, edge insertions and deletions, and vector entry updates --- the matrix $\ma$ and the lossy graph $G$ can change over time, and as such all bounding parameters $W_g, W_\eta, \lambda_1$ are defined as their maximum values over all times when a data structure operation is called. Throughout this section, we use the superscript $^{(t)}$ to denote the object after the $t$-th update, e.g., $\ma^{(t)}$ denotes the matrix $\ma$ after the $t$-th update. We omit this superscript when the object is clear from context, using it only when distinguishing between updates is necessary.

\begin{theorem}[Heavy hitters and sampler on two-sparse matrices]\label{thm:heavy_hitter_two_sparse_general}
Let $\ma \in \R^{m\times n}$ be a dynamic two-sparse matrix undergoing row insertions and deletions. Let $\lambda_1$ be the largest value such that $\|\ma h\|_2 \geq \sqrt{\lambda_1}\|h\|_2$ at any time the algorithm calls a query operation with input $h$, and let $W_{A} \defeq \max_t W_{\ma^{(t)}}$ (\Cref{def:W_A}). There is a data structure that w.h.p.~supports the following operations:
\begin{itemize}
\item \textsc{Initialize}$(\ma \in \R^{m \times n}, \ov{\tau} \in \R_{\geq 0}^m)$: Initializes with a two-sparse matrix $\ma$ and a vector $\ov{\tau}$ in amortized $\wt{O}(m \log^2(\lambda_1^{-1}) \log^{8}(W_A))$ time.
\item \textsc{Insert}$(a \in \R^n)$:
	Appends two-sparse row $a$ to $\ma$ in amortized $\wt{O}\left(\log^2(\lambda_1^{-1}) \log^6(W_A)\right)$ time. 
\item \textsc{Delete}$(e \in [m])$: 
    Deletes the row $e$ in $\ma$ in the same time as $\textsc{Insert}$. %
\item \textsc{ScaleTau}$(e\in [m], b\in \R)$: Sets $\ov{\tau}_e \leftarrow b$ in worst-case $O(1)$ time.
\item \textsc{QueryHeavy}$(h \in \R^{n}, \epsilon \in (0,1) )$:
	Returns $I \subseteq [m]$ containing exactly those $i$ with $|(\ma h)_i| \ge \epsilon$
	in amortized $\wt{O}( \epsilon^{-2} \|\ma h \|_2^2 + n \log^2(W_A))$ time. 
\item \textsc{Sample}$(h \in \R^V, C_0, C_1, C_2, C_3)$: Let a vector $p \in \R^E$ satisfy
\[
p_e \geq \min \left\{1, ~ C_1 \frac{m}{\sqrt{n}} \cdot \frac{(\ma h)_e^2}{\|\ma h\|_2^2} + C_2 \frac{1}{\sqrt{n}} + C_3 \ov{\tau}_e \right\}.
\]
Let $S = \sum_{e \in E} p_e$.  Let $X$ be a random variable which equals to $p_e^{-1} \indicVec{e}$ with probability $p_e / S$ for all $e \in E$. This operation returns a random diagonal matrix $\mr = C_0^{-1} \sum_{j=1}^{C_0 S} \mdiag(X_j)$, where $X_j$ are i.i.d.~copies of $X$. The amortized time of this operation and also the output size of $\mr$ are bounded by 
\[
\wt{O}\left(C_0 C_1 \frac{m}{\sqrt{n}} + ( C_0 C_2 \frac{m}{\sqrt{n}} + n) \cdot \log^2(W_A) + C_0 C_3 \|\ov{\tau}\|_1 \right).
\]
\end{itemize}
\end{theorem}

We note that $\lambda_1 = \lambda_1(\ma^\top \ma)$ suffices for the above theorem to hold, and that we prove the slightly stronger statement only requiring the bound on the Rayleigh quotient of $\ma^\top \ma$ for the query vectors $h$, as opposed to all vectors.

\subsection{Heavy Hitters on Balanced Lossy Graphs (Reduction to Balanced Lossy Expanders)}\label{sec:heavy_hitter_balanced}
In this section we present a heavy hitter data structure for any $\balloss$-balanced lossy graph, built using the data structure of \Cref{thm:heavy_hitter_expander}. The reduction in this section decomposes the input graph into subgraphs and maintains the following three properties for each subgraphs, as required by \Cref{thm:heavy_hitter_expander}: (1) Updates are only deletions. (2) The subgraphs are expanders. (3) The degrees of the subgraphs are approximately preserved under updates. 

\begin{theorem}[Heavy hitters on balanced lossy graphs]\label{thm:heavy_hitter_balanced}
Let $\balloss \leq 0.0005$. There is a data structure (Algorithm~\ref{alg:heavy_hitter_balanced} and \ref{alg:heavy_hitter_balanced_continued}) that w.h.p.~supports the following operations:
\begin{itemize}
\item $\textsc{Initialize}(G=(V,E,\eta), \epsad \in (0,1), \phi \in ((2000\balloss)^{1/2}, 1), \ov{\tau} \in \R_{\geq 0}^m)$: Initializes with a $\balloss$-balanced lossy graph $G=(V,E,\eta)$ with $|V|=n$ and $|E|=m$, parameters $\epsad$ and $\phi$, and a vector $\ov{\tau} \geq 0$ in amortized $O(m \phi^{-2}\log^2{(n/\epsad)})$ time. %
\item $\textsc{Delete}(e \in E)$: Deletes the edge $e$ from $G$ in amortized $O(\phi^{-3}\log(m) \log^2{(n/\epsad)})$ time. 
\item $\textsc{Insert}(a_e \in V, b_e \in V, \eta_e \in [1, 1+\balloss])$: Inserts a new edge $e = (a_e, b_e)$ with multiplier $\eta_e$ to $G$ in amortized $O(\phi^{-2}\log(m) \log^2{(n/\epsad)})$ time.
\item \textsc{ScaleTau}$(e \in E, b\in \R_{\geq 0})$: Sets $\ov{\tau}_e \leftarrow b$ in worst-case $O(1)$ time.
\item $\textsc{QueryHeavy}(h \in \R^V, \epsilon \in (0,1))$: Returns a set $I \subseteq E$ containing exactly all $e \in E$ that satisfies $|\mb_G h|_e \geq \epsilon$ in amortized $O(\phi^{-4} \epsilon^{-2} \|\mb_G h\|_2^2 + \phi^{-4} \epsad \epsilon^{-2} \|\md_G^{1/2}h\|_2^2 + n \log^2 m)$ time.
\item \textsc{Norm}$(h \in \R^V)$: Returns a value $L$ in amortized $O(n \log^2(m))$ time such that
\begin{align*}
\|\mb_G h\|_2^2 \leq L \leq O(\phi^{-4} \|\mb_G h\|_2^2 + \phi^{-4} \epsad \|\md_G^{1/2}h\|_2^2).
\end{align*}
\item \textsc{Sample}$(h \in \R^V, C_0, \ov{C}_1, \ov{C}_2, C_3)$: Let a vector $p \in \R^E$ satisfy
\[
p_e \geq \min \left\{1, ~ \ov{C_1} \cdot (\mb_G h)_e^2 + \ov{C}_2 + C_3 \ov{\tau}_e \right\}.
\]
Let $S = \sum_{e \in E} p_e$.  Let $X$ be a random variable which equals to $p_e^{-1} \indicVec{e}$ with probability $p_e / S$ for all $e \in E$. This operation returns a random diagonal matrix $\mr = C_0^{-1} \sum_{j=1}^{C_0 S} \mdiag(X_j)$, where $X_j$ are i.i.d.~copies of $X$. The amortized time of this operation and also the output size of $\mr$ are bounded by 
\[
O\left(C_0 \ov{C}_1 \phi^{-4} \log m(\|\mb_G h\|_2^2 + \epsad \|\md_G^{1/2}h\|_2^2) + C_0 \ov{C}_2 m \log m + C_0 C_3 \|\ov{\tau}\|_1 \log m + n \log^3(m) \right).
\]
\end{itemize}
\end{theorem}

\begin{algorithm}
\caption{Heavy hitter on $\balloss$-balanced lossy graphs}
\label{alg:heavy_hitter_balanced}
\DontPrintSemicolon
\SetKwInOut{Input}{Input}
\SetKwInOut{Output}{Output}
\SetKwInOut{Invariant}{Invariant}
\SetKwInOut{Member}{Member}

\SetKwFunction{Initialize}{Initialize}
\SetKwFunction{QueryHeavy}{QueryHeavy}
\SetKwFunction{Delete}{Delete}
\SetKwFunction{Insert}{Insert}
\SetKwFunction{Rebuild}{Rebuild}
\SetKwFunction{ScaleTau}{ScaleTau}
\SetKwFunction{Norm}{Norm}
\SetKwFunction{Sample}{Sample}
\SetKwFunction{ExpanderPrune}{ExpanderPrune}
\SetKwFunction{DegreePrune}{DegreePrune}

\Member{Parameters $\epsad$, $\phi$, and $k_1, \cdots, k_{\lceil\log m\rceil}$. 

Graph decomposition $G = \bigcup_{\ell=1}^{\lceil\log m\rceil} \bigcup_{j=1}^{k_{\ell}} G^{(\ell, j)}$, along with data structures $\DS^{(\ell,j)}$ of \Cref{thm:heavy_hitter_expander}, initial degrees $d^{(\ell,j)} \in \R^V$, initial number of edges $m_{\init}^{(\ell,j)}$ and a counter $m_{\cnt}^{(\ell,j)}$.}

\SetKwProg{Fn}{procedure}{:}{}
\Fn{\Initialize{$G=(V,E,\eta), \epsad \in (0,1), \phi \in ((2000\balloss)^{1/2}, 1), \ov{\tau} \in \R_{\geq 0}^E$}}{
    $\epsad \leftarrow \epsad$, $\phi \leftarrow \phi$. $k_1 \leftarrow 0, \cdots, k_{\lceil\log m\rceil-1} \leftarrow 0$, and $k_{\lceil\log m\rceil} \leftarrow 1$. \;
    Set $\ov{\tau} \leftarrow \ov{\tau}$, $G^{(\lceil\log m\rceil, 1)} \leftarrow G$, and call $\textsc{Rebuild}(\lceil\log m\rceil)$.
}

\Fn{\Delete{$e_{\start} \in E$}}{
    Find $G^{(\ell,j)}$ that contains edge $e_{\start}$, and set $F_{\tot} \leftarrow \emptyset$, $F \leftarrow \{e_{\start}\}$.\;
    \While{$F$ is not empty}{\label{algline:while_start}
        $F_{\tot} \leftarrow F_{\tot} \cup F$, increment $m_{\cnt}^{(\ell,j)} \leftarrow m_{\cnt}^{(\ell,j)} + |F|$ \;\label{algline:increment_m_cnt}
        $F_{\exp}, F'_{\exp} \leftarrow \textsc{ExpanderPrune}(F, \ell, j)$, delete edges in $F \cup F_{\exp} \cup F'_{\exp}$ from $G^{(\ell,j)}$ \;\label{algline:expander_prune}
        $F_{\deg} \leftarrow \textsc{DegreePrune}(F \cup F_{\exp}, \ell, j)$, delete edges in $F_{\deg}$ from $G^{(\ell,j)}$ \;\label{algline:degree_prune}
        $F \leftarrow F_{\deg}$, $F_{\tot} \leftarrow F_{\tot} \cup F_{\exp} \cup F'_{\exp}$  \;\label{algline:while_end}
    }
    Call $\DS^{(\ell,j)}.\textsc{Delete}\big(F_{\tot})$. Call $\textsc{Insert}(a_{e}, b_{e}, \eta_{e})$ for each edge $e \in F_{\tot}\backslash \{e_{\start}\}$. \; \label{algline:insert_and_delete_edges_in_delete}
    \If{$m^{(\ell,j)}_{\cnt} \geq (\phi/10)m_{\init}^{(\ell,j)}$}{\label{algline:delete_restart}
        Call $\textsc{Insert}(a_{e}, b_{e}, \eta_{e})$ for all edges $e \in E(G^{(\ell,j)})$, and destruct $G^{(\ell,j)}$ and $\DS^{(\ell,j)}$.
    }
}

\Fn{\Insert{$a_e \in V, b_e \in V, \eta_e \in [1, 1+\balloss]$}}{
    $k_1 \leftarrow k_1+1$, and let $G^{(1, k_1)}$ contain the single edge $e=(a_e, b_e)$ with multiplier $\eta_e$. \;
    \For{$\ell \in [\lceil\log m\rceil]$}{
        \If{total number of edges in level $\ell$ is at most $2^{\ell}$}{\label{algline:if_level_l_size}
            Call $\textsc{Rebuild}(\ell)$, and \textbf{break} \;
        }
        \Else{
            $G^{(\ell+1,j+k_{\ell+1})} \leftarrow G^{(\ell,j)}$ for all $j \in [k_{\ell}]$ \; 
            $k_{\ell+1} \leftarrow k_{\ell+1} + k_{\ell}$, $k_{\ell} \leftarrow 0$ \;\label{algline:if_level_l_size_end}
        }
    }
}

\Fn{\Rebuild{$\ell \in [\lceil\log m\rceil]$}}{
    Let $G^{(\ell)}$ be the union of all graphs of level $\ell$. \;
    Use \Cref{lem:expander_decomposition} to decompose the smoothed graph $\ov{G}^{\ell}$ of $G^{\ell}$ into edge-disjoint $10 \phi$-expanders, and let $G^{(\ell,1)}, \cdots, G^{(\ell,k_{\ell})}$ be the corresponding lossy graphs. \;
    \For{$j \in [k_{\ell}]$}{
    Initialize a data structure of \Cref{thm:heavy_hitter_expander}: $\DS^{(\ell,j)}.\textsc{Initialize}(G^{(\ell,j)}, \epsad, \ov{\tau}_{E(G^{(\ell,j)})})$.\;
    Let $d^{(\ell,j)} \in \R^V$ be the degree of $G^{(\ell,j)}$, $m_{\init}^{(\ell,j)} \leftarrow |E(G^{(\ell,j)})|$, $m_{\cnt}^{(\ell,j)} \leftarrow 0$.
    }
}

\Fn{\ExpanderPrune{$F \subset E, \ell, j$}}{\label{algline:def_expander_prune}
    Use \Cref{lem:expander_pruning} to find the set of vertices $S$ to be pruned from $G^{(\ell,j)}$ so that it remains a $\phi$-expander after deleting vertices in $F$. \; \label{algline:pruning}
    \Return $F_{\exp} \leftarrow E_{G^{(\ell,j)}}(S, V \backslash S)$ and $F'_{\exp} \leftarrow E_{G^{(\ell,j)}}(S, S)$
}

\Fn{\DegreePrune{$F \subset E, \ell, j$}}{\label{algline:def_degree_prune}
    $F_{\deg} \leftarrow \emptyset$. \;
    For $e \in F$, compute the degree $d'_{a_e}$, $d'_{b_e}$ of the two endpoints $a_e$, $b_e$ in $G^{(\ell,j)}$. \;
    For $e \in F$, for $v \in \{a_e, b_e\}$, if $d'_{v} < d^{(\ell,j)}_{v} / 9$ then add to $F_{\deg}$ all edges adjacent to $v$.\;\label{algline:degree_drop}
    \Return $F_{\deg}$
}
\end{algorithm}

\begin{algorithm}
\caption{Heavy hitter on $\balloss$-balanced lossy graphs (continued from Algorithm~\ref{alg:heavy_hitter_balanced})}
\label{alg:heavy_hitter_balanced_continued}
\DontPrintSemicolon
\SetKwInOut{Input}{Input}
\SetKwInOut{Output}{Output}
\SetKwInOut{Invariant}{Invariant}
\SetKwInOut{Member}{Member}

\SetKwFunction{Initialize}{Initialize}
\SetKwFunction{QueryHeavy}{QueryHeavy}
\SetKwFunction{Delete}{Delete}
\SetKwFunction{Insert}{Insert}
\SetKwFunction{Rebuild}{Rebuild}
\SetKwFunction{ScaleTau}{ScaleTau}
\SetKwFunction{Norm}{Norm}
\SetKwFunction{Sample}{Sample}
\SetKwProg{Fn}{procedure}{:}{}

\Fn{\ScaleTau{$e \in E, b \in \R_{\geq 0}$}}{
    $\ov{\tau}_e \leftarrow b$\;
    Find the subgraph $G^{(\ell,j)}$ that contains edge $e$, and call $\DS^{(\ell,j)}.\textsc{ScaleTau}(e, b)$. \;
}

\Fn{\QueryHeavy{$h \in \mathbb{R}^V, \epsilon \in (0,1)$}}{
    \lFor{$\ell \in [\lceil\log m\rceil]$ and $j \in [k_{\ell}]$}{
        $I^{(\ell,j)} \leftarrow \DS^{(\ell,j)}.\textsc{QueryHeavy}(h, \epsilon)$ 
    }
    \Return $I \leftarrow \bigcup_{\ell=1}^{\lceil\log m\rceil} \bigcup_{j=1}^{k_{\ell}} I^{(\ell,j)}$ \;
}

\Fn{\Norm{$h \in \mathbb{R}^V$}}{
    \lFor{$\ell \in [\lceil\log m\rceil]$ and $j \in [k_{\ell}]$}{
        $L^{(\ell,j)} \leftarrow \DS^{(\ell,j)}.\textsc{Norm}(h, \epsilon)$ 
    }
    \Return $L \leftarrow \sum_{\ell=1}^{\lceil\log m\rceil} \sum_{j=1}^{k_{\ell}} L^{(\ell,j)}$ \;
}

\Fn{\Sample{$h \in \mathbb{R}^V, \ov{C}_1, \ov{C}_2, C_3$}}{
    \For{$\ell \in [\lceil\log m\rceil]$ and $j \in [k_{\ell}]$}{
        $\mr^{(\ell,j)} \leftarrow \DS^{(\ell,j)}.\textsc{Sample}(h, \ov{C}_1, \ov{C}_2, C_3)$ \;
    }
    \Return $\mr$ which is a concatenation of all $\mr^{(\ell,j)}$ \;
}
\end{algorithm}

\paragraph{Expander decomposition and pruning.}
The data structure of \Cref{thm:heavy_hitter_balanced} uses the following tools of expander decomposition and expander pruning. We note that the expander decomposition and expander pruning results extend naturally to graphs with multi-edges, in the same way that they extend to weighted graphs in Section~4.1 of \cite{sw21}.
\begin{lemma}[Expander decomposition, Theorem~5.1 of \cite{cklpgs22} and Theorem~1.2 of \cite{sw21}]\label{lem:expander_decomposition}
There is an algorithm $\textsc{Decompose}(G)$ that takes as input any unweighted, undirected graph $G$, and w.h.p.~in $O(m \log^7(m))$ time computes an edge-disjoint partition of $G$ into graphs $G_0, G_1, \cdots, G_k$ for $k = O(\log n)$ such that each nontrivial connected component $X$ of $G_i$ is a $\phi$-expander for $\phi = \Theta(1/\log^3(m))$. %
\end{lemma}

\begin{lemma}[Expander pruning, Theorem~1.3 of \cite{sw21}]\label{lem:expander_pruning}
Let $G = (V,E)$ be a $\phi$-expander with $m$ edges. There is a deterministic algorithm with access to adjacency lists of $G$ such that, given an online sequence of $k \leq \phi m / 10$ edge deletions in $G$, can maintain a pruned set $P \subseteq V$ such that the following property holds. Let $G_i$ and $P_i$ be the graph $G$ and the set $P$ after the $i$-th deletion. We have, for all $i$,
\begin{enumerate}
    \item $P_0 = \emptyset$ and $P_i \subseteq P_{i+1}$,
    \item $\vol(P_i) \leq 8 i / \phi$ and $|E(P_i, V-P_i)| \leq 4i$, and
    \item $G_i\{V - P_i\}$ is a $\phi/6$-expander.
\end{enumerate}
The total time for updating $P_0, \cdots, P_k$ is $O(k \phi^{-2} \log(m))$. %
\end{lemma}

Before proving \Cref{thm:heavy_hitter_balanced}, we first prove some invariants that the algorithm maintains.
\begin{lemma}[Invariants of Algorithm~\ref{alg:heavy_hitter_balanced} and \ref{alg:heavy_hitter_balanced_continued}]\label{lem:invariant_heavy_hitter_balanced}
Let $\balloss \leq 0.0005$. Assuming the input graph $G$ to the data structure of Algorithm~\ref{alg:heavy_hitter_balanced} and \ref{alg:heavy_hitter_balanced_continued} is a $\balloss$-balanced lossy graph, and the input parameter $\phi \geq (2000\balloss)^{1/2}$, then the data structure maintains the following invariants after any operation:
\begin{enumerate}
\item {\bf (Edge-disjoint decomposition)} Each data structure $\DS^{(\ell,j)}$ maintains a subgraph $G^{(\ell,j)}$ that is edge disjoint with each other, and %
$\bigcup_{\ell = 1}^{\lceil\log m\rceil} \bigcup_{j = 1}^{k_{\ell}} E(G^{(\ell,j)}) = E(G)$, where $G$ is maintained by the data structure such that each $\textsc{Delete}(e)$ operation deletes the edge $e$ from $G$, and each $\textsc{Insert}(a_e, b_e, \eta_e)$ operation inserts a new edge $(a_e, b_e)$ with multiplier $\eta_e$ to $G$. 
\item {\bf (Size of each level)} For any $\ell \in [\lceil\log m\rceil]$, the subgraphs of level $\ell$ satisfy $\sum_{j \in [k_{\ell}]} |E(G^{(\ell,j)})| \leq 2^{\ell}$, and $\sum_{j \in [k_{\ell}]} |V(G^{(\ell,j)})| \leq O(n \log n)$.
\item {\bf (Properties of subgraphs)} $\forall \ell \in [\lceil\log m\rceil]$ and $\forall j \in [k_{\ell}]$, the subgraph $G^{(\ell,j)}$ is a $\balloss$-balanced lossy $\phi$-expander. Furthermore, after any $\DS^{(\ell,j)}.\textsc{Delete}$ operation, the degree of any vertex $v$ in $G^{(\ell,j)}$ either remains at least $1/9$ of the original degree $d^{(\ell,j)}_v$, or it drops to $0$.
\end{enumerate}
\end{lemma}
\begin{proof}
{\bf Part 1 (Edge-disjoint decomposition).} 
First note that by expander decomposition (\Cref{lem:expander_decomposition}), after each $\textsc{Rebuild}(\ell)$ operation we are guaranteed that every edge of level $\ell$ is included in one of the edge-disjoint expanders $G^{(\ell, 1)}, \cdots, G^{(\ell, k_{\ell})}$, and each data structure $\DS^{(\ell,j)}$ maintains an expander $G^{(\ell,j)}$.

$\textsc{Initialize}(\cdot)$ first sets the initial graph $G$ to have level $\lceil\log m\rceil$, and it then calls $\textsc{Rebuild}(\lceil\log m\rceil)$. This $\textsc{Rebuild}(\lceil\log m\rceil)$ ensures that initially $\bigcup_{\ell = 1}^{\lceil\log m\rceil} \bigcup_{j = 1}^{k_{\ell}} E(G^{(\ell,j)}) = E(G)$, and all the subgraphs are edge disjoint. Next we prove that we maintain an edge-disjoint decomposition after insertions and deletions.

Each $\textsc{Insert}(\cdot)$ inserts a new edge $e=(a_e, b_e)$ into a new subgraph $G^{(1,k_1)}$ in level $1$, so we still have $\bigcup_{\ell = 1}^{\lceil\log m\rceil} \bigcup_{j = 1}^{k_{\ell}} E(G^{(\ell,j)}) = E(G)$, this operation then calls $\textsc{Rebuild}(\cdot)$ recursively that still maintain an edge-disjoint decomposition.

Each $\textsc{Delete}(e)$ inserts all the deleted edges except $e$ back to the data structure by using $\textsc{Insert}(\cdot)$. So we still have $\bigcup_{\ell = 1}^{\lceil\log m\rceil} \bigcup_{j = 1}^{k_{\ell}} E(G^{(\ell,j)}) = E(G)$.

{\bf Part 2 (Size of each level).} We first prove the total number of edges in level $\ell$ is at most $2^{\ell}$. This bound holds because the algorithm could only add edges to level $\ell$ during $\textsc{Insert}(\cdot)$, and it would move all edges from level $\ell$ to $\ell+1$ if the total number of edges in level $\ell$ exceeds $2^{\ell}$ (see the if-else-clause from Line~\ref{algline:if_level_l_size} to \ref{algline:if_level_l_size_end} of Algorithm~\ref{alg:heavy_hitter_balanced}).

Next we bound the number of vertices in level $\ell$. Note that after each $\textsc{Rebuild}(\ell)$, \Cref{lem:expander_decomposition} implies that $\sum_{j=1}^{k_{\ell}} |V(G^{(\ell,j)})| \leq O(n \log n)$. The algorithm could only add more subgraphs to level $\ell$ in $\textsc{Insert}(\cdot)$, and there are two cases: (1) either the total number of edges in level $\ell$ is bounded by $2^{\ell}$, and the algorithm rebuilds level $\ell$, we again restore $\sum_{j=1}^{k_{\ell}} |V(G^{(\ell,j)})| \leq O(n \log n)$, (2) or the total number of edges in level $\ell$ exceeds $2^{\ell}$, and the algorithm moves all graphs in level $\ell$ to $\ell+1$, and this level becomes empty.

{\bf Part 3 (Properties of subgraphs).} Every subgraph $G^{(\ell,j)}$ is a $\balloss$-balanced lossy graph because the initial graph is $\balloss$-balanced, and any added edge satisfies $\eta_e \in [1, 1+\balloss]$. 

Every subgraph $G^{(\ell,j)}$ remains a $\phi$-expander throughout the algorithm because the subgraphs are constructed to be $10\phi$-expanders in $\textsc{Rebuild}$ by \Cref{lem:expander_decomposition}, and whenever the algorithm deletes an edge from a subgraph in $\textsc{Delete}$, we immediately use the pruning procedure of \Cref{lem:expander_pruning} to ensure that it remains a $\phi$-expander (see Line~\ref{algline:expander_prune} and \ref{algline:def_expander_prune}). Moreover, the algorithm uses the if-clause on Line~\ref{algline:delete_restart} to ensure that at most $\phi/10$ fraction of edges are pruned using \Cref{lem:expander_pruning}, so the pruning procedure is correct. Finally note that we never add any edge to a subgraph.

Finally note that the \textsc{DegreePrune} procedure of the algorithm ensures that at the end of each while-loop in \textsc{Delete} (Line~\ref{algline:while_end}), only the vertices that are endpoints of edges in $F$ may have dropped below $1/9$, so when the \textsc{Delete} operation ends, we must have that the degree of any vertex $v$ in $G^{(\ell,j)}$ either remains at least $1/9$ of the original degree or directly drops to $0$.
\end{proof}

Using these invariants, we are ready to prove the main result of this section.
\begin{proof}[Proof of \Cref{thm:heavy_hitter_balanced}]
First note that by the invariants proved in Part 3 of \Cref{lem:invariant_heavy_hitter_balanced}, all the requirements of \Cref{thm:heavy_hitter_expander} are satisfied by the subgraphs $G^{(\ell,j)}$, so all the function calls to data structures $\DS^{(\ell,j)}$ are correct. The invariants also ensure the correctness of $\textsc{Delete}(\cdot)$, $\textsc{Insert}(\cdot)$, and $\textsc{Rebuild}(\cdot)$. Also note that the correctness of $\textsc{ScaleTau}(\cdot)$ is straightforward, and its runtime is bounded by $O(n)$ by \Cref{thm:heavy_hitter_expander}. It remains to bound the time complexity of all other operations, and prove the correctness of $\textsc{QueryHeavy}(\cdot)$, $\textsc{Norm}(\cdot)$, and $\textsc{Sample}(\cdot)$.

{\bf Time complexity of $\textsc{Rebuild}$ and $\textsc{Initialize}$.}
By Part 2 of \Cref{lem:invariant_heavy_hitter_balanced}, $\sum_{j=1}^{k_{\ell}} |E(G^{(\ell,j)})| \leq 2^{\ell}$, so the expander decomposition algorithm of \Cref{lem:expander_decomposition} runs in $O(2^{\ell} \log^7 m)$ time. Initialization of the data structures $\DS^{(\ell,1)}, \cdots, \DS^{(\ell,k_{\ell})}$ of \Cref{thm:heavy_hitter_expander} takes $O(2^{\ell} \phi^{-2}\log^2{(n/\epsad)})$ time. 
Since $\textsc{Initialize}$ calls $\textsc{Rebuild}(\lceil\log m\rceil)$, it takes $O(m \phi^{-2}\log^2{(n/\epsad)})$ time. %

{\bf Amortized time complexity of $\textsc{Insert}$.} 
Whenever we perform $\textsc{Insert}(\cdot)$, we first assign $O\big(\log(m) \phi^{-2}\log^2{(n/\epsad)}\big)$ number of tokens to the edge $e$ that is being inserted. Next we show that when $\textsc{Insert}(\cdot)$ makes recursive calls to $\textsc{Rebuild}(\cdot)$, it has enough tokens to cover their cost. Whenever we call $\textsc{Rebuild}(\ell)$ during $\textsc{Insert}$, there must have been at least $2^{\ell-1}$ number of edges that just got moved up to level $\ell$ from level $\ell-1$, and we charge $O(\phi^{-2}\log^2{(n/\epsad)})$ number of tokens from every such edge. In this way we collect enough tokens to cover the time complexity of $\textsc{Rebuild}(\ell)$.

Note that every single edge that is moved up from some level $\ell-1$ to level $\ell$ must have been inserted by some $\textsc{Insert}(\cdot)$ since the initial edges are all in level $\lceil\log m\rceil$, and each edge can move up at most $\lceil\log m\rceil$ levels, so every edge has enough tokens to pay the charges of $\textsc{Rebuild}(\cdot)$.

{\bf Amortized time complexity of $\textsc{Delete}$.} 
We say that $\textsc{Delete}(e_{\start})$ is performed on a subgraph $G^{(\ell,j)}$ if $e_{\start} \in E(G^{(\ell,j)})$. The algorithm maintains a counter $m_{\cnt}^{(\ell,j)}$ for each subgraph $G^{(\ell,j)}$, which records the total size of all sets $F$ passed as inputs to $\textsc{ExpanderPrune}(F,\ell,j)$ (see Line~\ref{algline:increment_m_cnt} and \ref{algline:expander_prune}). We first prove the following key amortization claim: for any $G^{(\ell,j)}$ and any integer $t$, after $t$ number of $\textsc{Delete}(\cdot)$ performed on $G^{(\ell,j)}$, 
\begin{equation}\label{eq:counter_amortization}
m_{\cnt}^{(\ell,j)} \leq 7 t.
\end{equation}
We prove this claim by a token-based amortization argument. We maintain a pool of tokens that pays for each increment of $m_{\cnt}^{(\ell,j)}$ on Line~\ref{algline:increment_m_cnt}. Every time $\textsc{Delete}(e_{\start})$ is called on an edge $e_{\start} \in E(G^{(\ell,j)})$, we assign $7$ tokens to $e_{\start}$. We maintain the invariant that at the beginning of each while-loop with set $F$ (Line~\ref{algline:while_start}), every edge in $F$ has at least $7$ tokens. We also maintain the invariant that for any vertex that remains in the graph, it has 1 token for each edge deleted from it. These invariants holds initially, since the first while-loop starts with $F = \{e_{\start}\}$. In each iteration of the while-loop (Line~\ref{algline:while_start}), tokens are transferred among edges and vertices as follows: 
\begin{itemize}
\item \emph{Counter increment step.} When $m_{\cnt}^{(\ell,j)}$ is incremented by $|F|$ on Line~\ref{algline:increment_m_cnt}, each edge in $F$ pays $1$ token for this increment.
\item \emph{Expander pruning step.} In each call to $\textsc{ExpanderPrune}(F, \ell, j)$ on Line~\ref{algline:expander_prune} (defined on Line~\ref{algline:def_expander_prune}), each edge in $F$ pays $6$ tokens. Among these, $2$ tokens are given to each of the two endpoints of each edge in $F$, and for every cut edge $e = (u,v) \in E(S, V \backslash S)$ with $u \in S$ and $v \in V \backslash S$, $1$ token is given to the vertex $v$, restoring the second invariant. 

By \Cref{lem:expander_pruning}, $|E(S, V\backslash S)|$ is at most $4$ times the total number of edges passed as inputs to the expander pruning procedure of \Cref{lem:expander_pruning}, so $6 = 2 + 4$ tokens per edge is enough to pay for this step.
\item \emph{Degree pruning step.} In each call to $\textsc{DegreePrune}(F, \ell, j)$ on Line~\ref{algline:degree_prune} (defined on Line~\ref{algline:def_degree_prune}), each vertex $v$ detected on Line~\ref{algline:degree_drop} with $d'_{v} < d^{(\ell,j)}_{v} / 9$ pays $8 d'_{v}$ tokens. Among these, for every edge $e = (u,v)$ adjacent to $v$ that is added to $F_{\deg}$, $7$ tokens are given to $e$, and $1$ token is given to the other endpoint $u$ of $e$.

Since $d'_{v} < d^{(\ell,j)}_{v} /9$, and each deleted edge adjacent to $v$ gives $1$ token to $v$, $v$ has accumulated $d^{(\ell,j)}_{v} - d'_{v} > (8/9) d^{(\ell,j)}_{v} > 8 d'_{v}$ number of tokens, so $v$ has enough tokens to pay for this step. 
\end{itemize}
Finally, since each edge in $F_{\deg}$ receives $7$ tokens, and the next iteration of the while-loop begins with $F \leftarrow F_{\deg}$, the invariant that each edge in $F$ has at least $7$ tokens is preserved. This finishes the proof of \eqref{eq:counter_amortization}.

Next we use \eqref{eq:counter_amortization} to bound the amortized time complexity of $\textsc{Delete}(\cdot)$. Consider a fixed $\textsc{Delete}(e_{\start})$ performed on $G^{(\ell,j)}$, and let $\Delta m_{\cnt}^{(\ell,j)}$ denote the increase in $m_{\cnt}^{(\ell,j)}$ during this operation. In each while-loop with set $F$ (Line~\ref{algline:while_start}), $m_{\cnt}^{(\ell,j)}$ is incremented by $|F|$, and the edges in $F$, $F_{\exp}$, and $F'_{\exp}$ are added to $F_{\tot}$. By \Cref{lem:expander_pruning}, we have $|F_{\exp} \cup F'_{\exp}| \leq 8 |F| / \phi$, so when all while-loops terminate and Line~\ref{algline:insert_and_delete_edges_in_delete} is reached,
\begin{equation*}
|F_{\tot}| \leq (1 + 8/\phi) \cdot \Delta m_{\cnt}^{(\ell,j)}.
\end{equation*}
Apart from the if-clause on Line~\ref{algline:delete_restart} (which is called at most once for each $G^{(\ell,j)}$), the most time-consuming step of this $\textsc{Delete}(e_{\start})$ are the following two steps: (1) expander pruning steps of Line~\ref{algline:expander_prune}, (2) the calls to $\DS^{(\ell,j)}.\textsc{Delete}(F_{\tot})$ and $\textsc{Insert}(\cdot)$ on Line~\ref{algline:insert_and_delete_edges_in_delete}. So by \Cref{lem:expander_pruning}, \Cref{thm:heavy_hitter_expander} and the time complexity of $\textsc{Insert}(\cdot)$ proved earlier, the worst-case runtime of this $\textsc{Delete}(e_{\start})$ is 
\begin{equation*}
\Delta m_{\cnt}^{(\ell,j)} \cdot \phi^{-2} \log(m) + |F_{\tot}| \cdot O\left(\log(m) \phi^{-2}\log^2{(n/\epsad)}\right) \leq \Delta m_{\cnt}^{(\ell,j)} \cdot O\left(\log(m) \phi^{-3}\log^2{(n/\epsad)}\right).
\end{equation*}
Combining the above equation and \Cref{eq:counter_amortization}, we have that apart from the if-clause on Line~\ref{algline:delete_restart}, the amortized time of each $\textsc{Delete}(\cdot)$ is $O\left(\log(m) \phi^{-3}\log^2{(n/\epsad)}\right)$. Finally, since the if-clause on Line~\ref{algline:delete_restart} is executed only when $m^{(\ell,j)}_{\cnt} \geq (\phi/10)m_{\init}^{(\ell,j)}$, by \Cref{eq:counter_amortization} this if-clause is only executed after $O(\phi m_{\init}^{(\ell,j)})$ number of $\textsc{Delete}(\cdot)$ performed on $G^{(\ell,j)}$. Since the if-clause on Line~\ref{algline:delete_restart} inserts at most $m_{\init}^{(\ell,j)}$ number of edges, the amortized cost of this step is also $O(1/\phi) \cdot O\left(\log(m) \phi^{-2}\log^2{(n/\epsad)}\right) = O\left(\log(m) \phi^{-3}\log^2{(n/\epsad)}\right)$.

{\bf Correctness and time complexity of \textsc{QueryHeavy}, \textsc{Norm}, and \textsc{Sample}.} 
By Part 1 of \Cref{lem:invariant_heavy_hitter_balanced}, the algorithm maintains $\bigcup_{\ell = 1}^{\lceil\log m\rceil} \bigcup_{j = 1}^{k_{\ell}} E(G^{(\ell,j)}) = E(G)$, so the correctness of these three operations directly follows from the correctness of the three operations of $\DS^{(\ell,j)}$ by \Cref{thm:heavy_hitter_expander}.

Using \Cref{thm:heavy_hitter_expander}, the runtime of $\textsc{QueryHeavy}(\cdot)$ is
\begin{align*}
&~ \sum_{\ell=1}^{\lceil\log m\rceil} \sum_{j=1}^{k_{\ell}} O\left(\phi^{-4} \epsilon^{-2} \|\mb^{(\ell,j)} h\|_2^2 + \phi^{-4} \epsad \epsilon^{-2} \|(\md^{(\ell,j)})^{1/2}h\|_2^2 + |V(G^{(\ell,j)})|\right) \\
\leq &~ O\left(\phi^{-4} \epsilon^{-2} \|\mb_G h\|_2^2 + \phi^{-4} \epsad \epsilon^{-2} \|\md^{1/2}h\|_2^2 + n \log^2 m\right),
\end{align*}
where we bound the first term by Part 1 of \Cref{lem:invariant_heavy_hitter_balanced} that $\{G^{(\ell,j)}\}$ forms an edge-disjoint decomposition of $G$, and we bound the second term by $\sum_{\ell=1}^{\lceil\log m\rceil} \sum_{j=1}^{k_{\ell}} \md^{(\ell,j)} = \md$ because this decomposition is edge disjoint, and we bound the third term by Part 2 of \Cref{lem:invariant_heavy_hitter_balanced} that $\sum_{j=1}^{k_{\ell}} |V(G^{(\ell,j)})| \leq n \log n$. 

Using \Cref{thm:heavy_hitter_expander}, the runtime of $\textsc{Norm}(\cdot)$ is
\begin{align*}
\sum_{\ell=1}^{\lceil\log m\rceil} \sum_{j=1}^{k_{\ell}} O\left(|V(G^{(\ell,j)})|\right) \leq O(n \log^2(m)),
\end{align*}
which follows from Part 2 of \Cref{lem:invariant_heavy_hitter_balanced} that $\sum_{j=1}^{k_{\ell}} |V(G^{(\ell,j)})| \leq n \log n$. 

Using \Cref{thm:heavy_hitter_expander}, and denote $m^{(\ell,j)} = |E(G^{(\ell,j)})|$, $n^{(\ell,j)} = |V(G^{(\ell,j)})|$, and $\ov{\tau}^{(\ell,j)} = \ov{\tau}_{E(G^{(\ell,j)})}$, the runtime of $\textsc{Sample}(\cdot)$ is
{\small
\begin{align*}
& \sum_{\ell=1}^{\lceil\log m\rceil} \sum_{j=1}^{k_{\ell}} O\left( \big(\ov{C}_1 (\phi^{-4} \|\mb^{(\ell,j)} h\|_2^2 + \phi^{-4} \epsad \|(\md^{(\ell,j)})^{1/2}h\|_2^2) + \ov{C}_2 m^{(\ell,j)} + C_3 \|\ov{\tau}^{(\ell,j)}\|_1 \big) C_0 \log m + n^{(\ell,j)} \log n \right) \\
& \leq O\left(C_0 \ov{C}_1 \phi^{-4} \log m(\|\mb_G h\|_2^2 + \epsad \|\md_G^{1/2}h\|_2^2) + C_0 \ov{C}_2 m \log m + C_0 C_3 \|\ov{\tau}\|_1 \log m + n \log^3 m \right),
\end{align*}
}
where we again used that $\{G^{(\ell,j)}\}$ forms an edge-disjoint decomposition of $G$, so $\sum_{\ell=1}^{\lceil\log m\rceil} \sum_{j=1}^{k_{\ell}} \md^{(\ell,j)} = \md$, and by Part 2 of \Cref{lem:invariant_heavy_hitter_balanced} that $\sum_{j=1}^{k_{\ell}} |V(G^{(\ell,j)})| \leq n \log n$. 
\end{proof}

\subsection{Heavy Hitters on General Lossy Graphs (Reduction to Balanced Lossy Graphs)}\label{sec:heavy_hitter_general}

In this section we present a heavy hitter data structure for general weighted lossy graphs  (\Cref{thm:heavy_hitter_general}), built using the data structure of \Cref{thm:heavy_hitter_balanced}. The reduction in this section ensures two properties that are required by \Cref{thm:heavy_hitter_balanced}: (1) We decompose a general lossy graph to $\balloss$-balanced lossy graphs. (2) We decompose the weighted graph into unweighted graphs, and implement $\textsc{Scale}(\cdot)$ using $\textsc{Insert}(\cdot)$ and $\textsc{Delete}(\cdot)$.

\begin{theorem}[Heavy hitters and sampler on general lossy graphs]\label{thm:heavy_hitter_general}
Let $G = (V,E,\eta)$ be a dynamic lossy graph undergoing edge insertions and deletions, and let $\ma$ denote its incidence matrix. 
Let $g \in \R^{E}_{>0}$ be a positive weight vector. %
Let $m \defeq \max_t |E^{(t)}|$, and $n \defeq \max_t |V^{(t)}|$. 
Let $W_{\eta} \defeq \max_{e,t} \eta^{(t)}_e$, and let $W_g \defeq \max_{t} \frac{\max_{e} g^{(t)}_e}{\min_{e} g^{(t)}_e}$. 
Let $\lambda_1$ be the largest value such that $\|\ma h\|_2 \geq \sqrt{\lambda_1}\|h\|_2$ at any time the algorithm calls a query operation with input $h$. There is a data structure that w.h.p.~supports the following operations:
\begin{itemize}
\item \textsc{Initialize}$(\ma \in \R^{E \times V}, g \in \R_{>0}^E, \ov{\tau} \in \R^E)$: 
Initializes with the incidence matrix $\ma$ of a lossy graph $G = (V,E,\eta)$, and two vectors $g$ and $\ov{\tau}$ in amortized time 
\[
\wt{O}\Big(m \log^2({\lambda}_1^{-1}) \log^3(W_{\eta}) \log^3(W_g) \Big).
\]
\item \textsc{Insert}$(i\in V, j\in V, \eta\in \R, b\in \R)$:
	Appends $\ma$ with the row $\indicVec{i} - (1 + \eta)\indicVec{j}$, and appends a corresponding entry of $b$ to $g$ in amortized $\wt{O}\left(\log^2({\lambda}_1^{-1}) \log^2(W_{\eta}) \log^2(W_g)\right)$ time. 
\item \textsc{Delete}$(e \in E)$: 
    Deletes the row $e$ in both $g$ and $\ma$ in the same time as $\textsc{Insert}$. %
\item \textsc{Scale}$(e \in E, b \in \R)$:
	Sets $g_e \leftarrow b$ in the same time as $\textsc{Insert}$.
\item \textsc{ScaleTau}$(e\in E, b \in \R)$: Sets $\ov{\tau}_e \leftarrow b$ in worst-case $O(1)$ time.
\item \textsc{QueryHeavy}$(h \in \R^V, \epsilon \in (0,1) )$:
	Returns $I \subset E$ containing exactly those $e$ with $|(\mg \ma h)_e| \ge \epsilon$
	in amortized $O( \epsilon^{-2} \| \mg \ma h \|_2^2) + \wt{O}(n \log(W_{\eta}) \log(W_g)) $ time.
\item \textsc{Sample}$(h \in \R^V, C_0, C_1, C_2, C_3)$: Let a vector $p \in \R^E$ satisfy
\[
p_e \geq \min \left\{1, ~ C_1 \frac{|E|}{\sqrt{|V|}} \cdot \frac{(\mg\ma h)_e^2}{\|\mg\ma h\|_2^2} + C_2 \frac{1}{\sqrt{|V|}} + C_3 \ov{\tau}_e \right\}.
\]
Let $S = \sum_{e \in E} p_e$.  Let $X$ be a random variable which equals to $p_e^{-1} \indicVec{e}$ with probability $p_e / S$ for all $e \in E$. This operation returns a random diagonal matrix $\mr = C_0^{-1} \sum_{j=1}^{C_0 S} \mdiag(X_j)$, where $X_j$ are i.i.d.~copies of $X$. The amortized time of this operation and also the output size of $\mr$ are bounded by 
\[
\wt{O}\Big(C_0 C_1 \frac{m}{\sqrt{n}} + ( C_0 C_2 \frac{m}{\sqrt{n}} + n) \cdot \log(W_{\eta}) \log(W_g) + C_0 C_3 \|\ov{\tau}\|_1 \Big).
\]
\end{itemize}
\end{theorem}

Before proving \Cref{thm:heavy_hitter_general}, we first prove the following key lemma which provides a procedure for decomposing any general lossy graph into $\balloss$-balanced lossy graphs. This lemma first decomposes a general lossy graph into bipartite graphs whose edges are all oriented in the same direction, then partitions the edges in each bipartite graph by their flow multipliers $\eta_e$, and finally finds an appropriate vertex scaling factor to ensure all flow multipliers $\eta_e$ are bounded by $1 + \balloss$.

\begin{lemma}[Decomposition into balanced lossy graphs]\label{lem:decomposition_balanced}
Consider any vertex set $V$ and any $W_\eta > 1, \balloss >  0$. %
There exists a map $f:\{(u,v,\eta_e) \mid u,v \in V, \eta_e \leq 1 + W_\eta\}\rightarrow [\lceil\log n\rceil \cdot 2 \lceil\frac{\log(W_{\eta})}{\log(1+\balloss)}\rceil]$ mapping edges into buckets, as well as a diagonal vertex scaling matrix $\ms^{(i)} \in \R^{V \times V}_{> 0}$ for each bucket $i \in [\lceil\log n\rceil\cdot 2\lceil\frac{\log(W_{\eta})}{\log(1+\balloss)}\rceil]$ such that for any lossy graph $G = (V,E,\eta)$ with $\eta_e \leq 1 + W_\eta$, we have: 
\begin{enumerate}
\item Let $E^{(i)} \subseteq E$ denote the edges mapped into $i$ by $f$. Let $\mb^{(i)}$ denote the submatrix of $\mb_G$ with rows in $E^{(i)}$. Then the matrix $\mb^{(i)} \cdot \ms^{(i)}$ is the incidence matrix of a $\balloss$-balanced lossy graph. 
\item $\forall i,k, W_{\eta}^{-1} \leq \ms^{(i)}_{k,k} \leq 1$.
\item $f$ is computable in $O(\log{n})$ time. 
\end{enumerate}
\end{lemma}

\begin{proof}
For clarity, we construct this mapping in two steps, mapping each edge of form $(u,v,\eta_e)$ where $u,v \in V$ and $\eta_e \leq 1 + W_\eta$ into a bucket $[\lceil\log{n}\rceil]\times [2 \lceil\frac{\log(W_\eta)}{\log(1+\balloss)}\rceil]$. First, label each vertex by an integer from $1$ to $n$. For an edge $e = (u,v)$, let $i$ be the first index in which $u$ and $v$ differ in their binary representations, and assign this edge to the $i$-th bucket $E^{(i)}$, consisting of the set of edges that would map to this bucket, i.e., the preimage of $i$ in $E$. 

Note that edges mapped to $E^{(i)}$ are bipartite, since edges in $E^{(i)}$ are only between vertices whose $i$-th bit differs. Let us consider this bipartitioning. Let $L^{(i)}$ be the set of vertices whose $i$-th bit are $0$, and let $R^{(i)}$ be the set of vertices whose $i$-th bit are $1$. Denote the edges in $E^{(i)}$ that are directed from $L^{(i)}$ to $R^{(i)}$ as $E_{L^{(i)}\to R^{(i)}}$, and $R^{(i)}$ to $L^{(i)}$ as $E_{R^{(i)}\to L^{(i)}}$. For $j = 0, 1, \cdots, \lceil\frac{\log(W_{\eta})}{\log(1+\balloss)}\rceil$, define the edge sets:
\[
E^{(i,j)} \defeq \Big\{e \in E_{L^{(i)}\to R^{(i)}} : \eta_e \in [(1+\balloss)^{j}, (1+\balloss)^{j+1})\Big\}.
\]
This partitions $E^{(i)}$, and forms the map $f$ for edges going from $L^{(i)}$ to $R^{(i)}$. More precisely, for any edge $e = (u,v,\eta_e)$, we map this edge to $(i,j)$, where $i$ is the first bit that differs in the bit representations of $u$ and $v$, and $j$ is $\lfloor \frac{\log{(1+\eta_e)}}{\log (1 +\balloss)}\rfloor$ if the edge $u$ has $0$ as its $i$-th bit.

Next, define the scaling matrix $\ms^{(i,j)}$ as follows: for all $v \in R^{(i)}$ set $\ms^{(i,j)}_{v,v} = 1$, and for all $v \in L^{(i)}$ set $\ms^{(i,j)}_{v,v} = (1+\balloss)^{-j}$. It is easy to see that both conditions hold. Every edge set $E^{(i,j)}$ is constructed as $E^{(i,j)} = \{e \in E_{L^{(i)}\to R^{(i)}} : \eta_e \in [(1+\balloss)^{j}, (1+\balloss)^{j+1})\}$, and since we define the scaling matrix $\ms^{(i,j)}$ to have $(1+\balloss)^{-j}$ on entries corresponding to $v \in L$, after scaling all the flow multipliers are between $1$ and $1 + \balloss$. 

Likewise, define $E^{(i,j)}$ for $j = \lceil\frac{\log(W_{\eta})}{\log(1+\balloss)}\rceil + 1, \lceil\frac{\log(W_{\eta})}{\log(1+\balloss)}\rceil + 2, \cdots, 2\lceil\frac{\log(W_{\eta})}{\log(1+\balloss)}\rceil$ for the edges going from $R^{(i)}$ to $L^{(i)}$. This gives us the full map. For any edge $e = (u,v,\eta_e)$, we map this edge to:
\begin{align}
f(e)=\begin{cases}
(i,\lfloor \frac{\log{\eta_e}}{\log (1 +\balloss)}\rfloor),& \text{if }e \in E_{L^{(i)}\rightarrow R^{(i)}},\\
(i, \lceil\frac{\log(W_\eta)}{\log(1 +\balloss)}\rceil+\lfloor\frac{\log{\eta_e}}{\log (1 +\balloss)}\rfloor, & \text{if } e \in E_{R^{(i)}\rightarrow L^{(i)}}.\\
\end{cases}\,
\end{align}
where $i$ is the first bit where $u$ and $v$ differ. 
\end{proof}

\begin{table}[!ht]
\renewcommand{\arraystretch}{1.2}
\centering
    \begin{tabular}{|c|c|c|c|}
    \hline
    {\bf Param} & $\phi$ & $\balloss$ & $\epsad$ \\ \hline
    {\bf Value} & $\log^{-3}(m)$ & $\log^{-11} (m)$ & $m^{-1} W_\eta^{-2} W_g^{-2} \lambda_1 \log^{-1}(W_{\eta}) \log^{-1}(W_{g}) \log^{-12}(m)$\\ \hline
    \end{tabular}
    \caption{Choice of parameters for \Cref{thm:heavy_hitter_general}, where $W_{\eta}, W_g, \lambda_1$ are defined in the theorem statement.}
    \label{tab:parameters}
\end{table}

Now we are ready to prove \Cref{thm:heavy_hitter_general}.
\begin{proof}[Proof of \Cref{thm:heavy_hitter_general}]
In the proof we show how to perform the operations \textsc{Initialize}, \textsc{Scale}, \textsc{ScaleTau}, \textsc{QueryHeavy}, and \textsc{Sample} using the heavy hitter data structure for balanced lossy graphs in \Cref{thm:heavy_hitter_balanced}. In this proof we will use parameters $\balloss = 1/\log^{11} m$, $\phi = 1/\log^3 m$, and $\epsad = m^{-1} W_\eta^{-2} W_g^{-2} \lambda_1 \log^{-1}(W_{\eta}) \log^{-1}(W_{g}) \log^{-12}(m)$, as summarized in Table~\ref{tab:parameters}. We also denote $g_{\min} = \min_{e' \in E} g_{e'}$ and $g_{\max} = \max_{e' \in E} g_{e'}$. 

{\bf Initialization.} Let $\ma \in \R^{m \times n}$ and $g \in \R_{>0}^m$ denote the initial inputs, where $\ma$ is the incidence matrix of a lossy graph $G = (V, E, \eta)$. We first use \Cref{lem:decomposition_balanced} with parameter $\balloss$ to decompose the edges of $G$ into $E = \bigcup_{j=1}^T E^{(j)}$, where 
\[
T = O\big(\log(m) \log(W_{\eta}) \balloss^{-1}\big) = O\big(\log^{12}(m) \log(W_{\eta}) \big).
\]
Let $\ma^{(j)}$ denote the submatrix of $\ma$ with rows in $E^{(j)}$, let $\ms^{(j)}$ be the vertex scaling matrix for $\ma^{(j)}$ given by \Cref{lem:decomposition_balanced}. %

Next for every edge set $E^{(j)}$, we decompose it into $E^{(j)} = \bigcup_{k=1}^{\lceil\log(W_g)\rceil} E^{(j,k)}$, where $E^{(j,k)}$ include all $e \in E^{(j)}$ such that 
\[
g_{\min} \cdot 2^{k-1} \leq g_e \leq g_{\min} \cdot 2^{k}.
\]
Note that the decomposition covers all $e \in E^{(j)}$ since every $g_e$ satisfies $g_{\min} \leq g_e \leq g_{\min} \cdot 2^{\lceil\log(W_g)\rceil} = g_{\max}$ because we defined $W_g = \frac{g_{\max}}{g_{\min}}$. Let $\ma^{(j,k)}$ denote the submatrix of $\ma$ with rows in $E^{(j,k)}$. Let $G^{(j,k)}$ denote the subgraph of $G$ with edges in $E^{(j,k)}$ and multipliers $\eta^{(j,k)}_e = (1+\eta_e) \cdot \ms^{(j)}_{a_e,a_e} - 1$ for all $e = (a_e,b_e) \in E^{(j,k)}$, which is equivalent to $G^{(j,k)}$ being an unweighted lossy graph with incidence matrix $\ma^{(j,k)} \cdot \ms^{(j)}$. Let $\md^{(j,k)} = \md_{G^{(j,k)}}$. Let $g^{(j,k)}$ and $\ov{\tau}^{(j,k)}$ be the subvector of $g$ and $\ov{\tau}$ with entries in $E^{(j,k)}$.

For every $j \in [T]$ and every $k \in [\lceil\log(W_g)\rceil]$, we initialize a heavy hitter data structure $\DS^{(j,k)}$ of \Cref{thm:heavy_hitter_balanced} for the lossy graph $G^{(j,k)}$, with parameters $\epsad$ and $\phi$, and vector $\ov{\tau}^{(j,k)}$.

By \Cref{lem:decomposition_balanced}, computing the first decomposition takes $O(m \log(m))$ time. Computing the second decomposition takes $O(n)$ time. Initializing the $T \cdot \lceil\log(W_g)\rceil$ data structures takes $T \log(W_g) \cdot O(m \phi^{-2}\log^2{(n/\epsad)})$ time by \Cref{thm:heavy_hitter_balanced}. So the total initialization time is 
\begin{align*}
P = &~ O\Big(m \phi^{-2}\log^2{(n/\epsad)}\log{(m)} \log(W_{\eta}) \log(W_g) \balloss^{-1} \Big) \\
= &~ O\Big(m \log^{20}(m) \log^2(\lambda_1^{-1}) \log(W_{\eta})^3 \log(W_g)^3 \Big),
\end{align*}
where we used the choice of parameters $\epsad = m^{-1} W_\eta^{-2} W_g^{-2} \lambda_1\log^{-1}(W_{\eta}) \log^{-1}(W_{g}) \log^{-12}(m)$, and $\balloss = 1/\log^{11} m$.

{\bf Scale.} Given $i \in [m]$ and $b \in \R$, let $a_i$ and $b_i$ denote the two endpoints of edge $i$. We first find the set $E^{(j,k)}$ such that $i \in E^{(j,k)}$, and update $g_i$ and $g^{(j,k)}_i$ to be $b$. Next we find $k' \in [1,\lceil\log(W_g)\rceil]$ such that the new scalar $b$ satisfies 
\[
g_{\min} \cdot 2^{k'-1} \leq b \leq g_{\min} \cdot 2^{k'}.
\]

If $k' = k$, then we don't change anything else. 

If $k' \neq k$, then we call $\mathrm{DS}^{(j,k)}.\textsc{Delete}(i)$, and $\mathrm{DS}^{(j,k')}.\textsc{Insert}(a_i, b_i, \eta^{(j,k)}_i)$, and we also delete $g^{(j,k)}_i$ from $g^{(j,k)}$, and add a new entry of value $b$ to $g^{(j,k')}$ that corresponds to the newly added edge. This takes $O\left(\phi^{-3}\log(m) \log^2{(n/\epsad)} \right) = O(\log^{12}(m) \log^2(\lambda_1^{-1}) \log(W_{\eta})^2 \log(W_g)^2)$ time by \Cref{thm:heavy_hitter_balanced}.

{\bf Insert.} Given a new edge $e = (u,v,\eta_e)$, and its edge weight $b$, we use the map $f$ in \cref{lem:decomposition_balanced} to find the bucket $E^{(j)}$ for $j = f(u,v,\eta_e)$ that $e$ belongs to. Next, we find $k \in [1,\lceil\log(W_g)\rceil]$ such that the new scalar $b$ satisfies 
\[
g_{\min} \cdot 2^{k-1} \leq b \leq g_{\min} \cdot 2^{k}.
\]
and we call $\DS^{(j,k)}.\textsc{Insert}(u,v,\eta_e)$. Note that the guarantee of \cref{lem:decomposition_balanced} gives us that adding this edge to $E^{(j)}$ maintains the property that the subset of edges $E^{(j)}$ form a $\balloss$-balanced lossy graph. 

Note here that as the graph evolves, $g_{\min}$ may change. We keep this notation for ease of understanding, but it is easy to implement this data structure efficiently without actually changing each set, simply by defining an offset term. 

This takes $O\left(\phi^{-2}\log(m) \log^2{(n/\epsad)} \right) = O(\log^{9}(m) \log^2(\lambda_1^{-1}) \log(W_{\eta})^2 \log(W_g)^2)$ time by \Cref{thm:heavy_hitter_balanced}.

{\bf Delete.} Given $e \in [m]$, we find the set $E^{(j,k)}$ that $e \in E^{(j,k)}$, then we call $\mathrm{DS}^{(j,k)}.\textsc{Delete}(i)$. 

This takes $O\left(\phi^{-3}\log(m) \log^2{(n/\epsad)} \right) = O(\log^{12}(m) \log^2(\lambda_1^{-1}) \log(W_{\eta})^2 \log(W_g)^2)$ time by \Cref{thm:heavy_hitter_balanced}.

{\bf ScaleTau.} Given $e\in [m]$ and $b \in \R$, we find the set $E^{(j,k)}$ such that $e \in E^{(j,k)}$, and call $\DS^{(j,k)}.\textsc{ScaleTau}(e, b)$. This takes $O(1)$ time by \Cref{thm:heavy_hitter_balanced}.

{\bf QueryHeavy.} Given query vector $h \in \R^n$ and error parameter $\epsilon \in (0,1)$, we compute $h^{(j)} = (\ms^{(j)})^{-1} \cdot h$ for all $j \in [T]$ and $\epsilon^{(k)} = \epsilon \cdot 2^{-k} g_{\min}^{-1}$ for all $k \in [\lceil\log(W_g)\rceil]$. We then call $\mathrm{DS}^{(j,k)}.\textsc{QueryHeavy}(h^{(j)}, \epsilon^{(k)})$ for all $j \in [T]$ and $k \in [\lceil\log(W_g)\rceil]$. By \Cref{thm:heavy_hitter_balanced}, the time of the function call $\mathrm{DS}^{(j,k)}.\textsc{QueryHeavy}$ is
\begin{align}\label{eq:query_heavy_bound}
&~ O\left(\phi^{-4} (\epsilon^{(k)})^{-2} \|(\ma^{(j,k)} \ms^{(j)}) h^{(j)}\|_2^2 + \phi^{-4} \epsad (\epsilon^{(k)})^{-2} \|(\md^{(j,k)})^{1/2}h^{(j)}\|_2^2 + n \log^2 m\right) \notag \\
\leq &~ O\left(\phi^{-4} \epsilon^{-2} \|\mg^{(j,k)} \ma^{(j,k)} h\|_2^2 + \phi^{-4} \epsad \epsilon^{-2} g_{\max}^2 m W_{\eta}^{2} \cdot \|h\|_2^2 + n \log^2 m\right),
\end{align}
where we bound the first term using the definition of $h^{(j)}$ and $\epsilon^{(k)}$, and that $g^{(j,k)}_e = \Theta(2^{k} g_{\min})$, and we bound the second term using that $G^{(j,k)}$ is $\balloss$-balanced, so each diagonal entry of $\md^{(j,k)}$ is at most $m(1+\balloss) \leq O(m)$, %
and $\ms^{(j)}_{i,i} \geq W_{\eta}^{-1}$ by \Cref{lem:decomposition_balanced}, and $\epsilon^{(k)} \geq \epsilon \cdot g_{\max}^{-1}$.

So the total runtime of this operation is
\begin{align}\label{eq:query_heavy_bound_total}
&~ \sum_{j=1}^T \sum_{k=1}^{\log(W_g)} O\left(\phi^{-4} \epsilon^{-2} \|\mg^{(j,k)} \ma^{(j,k)} h\|_2^2 + \phi^{-4} \epsad \epsilon^{-2} g_{\max}^2 m W_{\eta}^{2} \cdot \|h\|_2^2 + n \log^2 m\right) \notag \\
\leq &~ O\Big( \phi^{-4} \epsilon^{-2} \|\mg \ma h\|_2^2 + n \log^{14}(m) \log(W_{\eta}) \log(W_g) \Big),
\end{align}
where we bound the first term using the fact that $E^{(j,k)}$'s form an edge-disjoint decomposition of $E$, we bound the second term by $\|h\|_2^2 \leq \|\ma h\|_2^2 \cdot \lambda_1^{-1} \leq \|\mg \ma h\|_2^2 \cdot g_{\min}^{-2} \cdot \lambda_1^{-1}$ and our choice of $\epsad = m^{-1} W_{\eta}^{-2} W_g^{-2} \lambda_1 \cdot T^{-1} \log^{-1} (W_g)$, and lastly we bound the third term by $T = O\big(\log(m) \log(W_{\eta}) \balloss^{-1}\big)$, $\balloss = 1/\log^{11} m$.

{\bf Sample.} Given a vector $h \in \R^n$ and parameters $ C_1, C_2, C_3$, we first compute $h^{(j)} = (\ms^{(j)})^{-1} \cdot h$ call $L^{(j,k)} = \DS^{(j,k)}.\textsc{Norm}(h^{(j)})$ for all $j \in [T]$ and $k \in [\lceil\log(W_g)\rceil]$. We then compute
\[
L = \sum_{j=1}^T \sum_{k=1}^{\lceil\log(W_g)\rceil} 2^{2k} g_{\min}^{2} \cdot L^{(j,k)}.
\]
By \Cref{thm:heavy_hitter_balanced} each returned value $L^{(j,k)}$ satisfies
\begin{align*}
\|(\ma^{(j,k)} \ms^{(j)}) h^{(j)}\|_2^2 \leq L^{(j,k)} \leq O\left(\phi^{-4} \|(\ma^{(j,k)} \ms^{(j)}) h^{(j)}\|_2^2 + \phi^{-4} \epsad \|(\md^{(j,k)})^{1/2}h\|_2^2\right).
\end{align*}
Using a similar argument as how we proved Eq.~\eqref{eq:query_heavy_bound}, 
\begin{align*}
\|\mg \ma h\|_2^2 \leq L \leq O\left(\phi^{-4} \|\mg \ma h\|_2^2 + \phi^{-4} \epsad \cdot m g_{\max}^2 W_{\eta}^2 \|h\|_2^2\right),
\end{align*}
then using $\|h\|_2^2 \leq \|\mg \ma h\|_2^2 \cdot g_{\min}^{-2} \cdot \lambda_1^{-1}$ and our choice of $\epsad \leq m^{-1} W_{\eta}^{-2} W_g^{-2} \lambda_1$, 
\begin{align*}
\|\mg \ma h\|_2^2 \leq L \leq O\left(\phi^{-4} \|\mg \ma h\|_2^2\right).
\end{align*}

Next we define
\[
\ov{C}_1^{(k)} = C_1 \cdot \frac{m}{\sqrt{n}} \cdot \frac{2^{2k} g_{\min}^2}{\phi^4 L}, ~~ \ov{C}_2 = \frac{C_2}{\sqrt{n}},
\]
and for all $j \in [T]$ and $k \in [\lceil\log(W_g)\rceil]$, we call $\DS^{(j,k)}.\textsc{Sample}(h^{(j)}, \ov{C}_1^{(k)}, \ov{C}_2, C_3)$ to obtain a diagonal matrix $\mr^{(j,k)}$ of size $|E^{(j,k)}| \times |E^{(j,k)}|$. We concatenate all $\mr^{(j,k)}$ to form a diagonal matrix $\mr$ of size $m \times m$. By \Cref{thm:heavy_hitter_balanced}, every $e \in E^{(j,k)}$ is sampled with probability
\begin{align*}
p_e \geq &~ \min \left\{1, ~ \ov{C_1}^{(k)} \cdot (\ma^{(j,k)} \ms^{(j)} h^{(j)})_e^2 + \ov{C}_2 + C_3 \ov{\tau}_e \right\} \\
= &~ \min \left\{1, ~ C_1 \cdot \frac{m}{\sqrt{n}} \cdot \frac{2^{2k} g_{\min}^2}{\phi^4 L} \cdot (\ma^{(j,k)} h)_e^2 + \frac{C_2}{\sqrt{n}} + C_3 \ov{\tau}_e \right\} \\
\geq &~ \min \left\{1, ~ C_1 \frac{m}{\sqrt{n}} \cdot \frac{(\mg\ma h)_e^2}{\|\mg\ma h\|_2^2} + C_2 \frac{1}{\sqrt{n}} + C_3 \ov{\tau}_e \right\},
\end{align*}
where the second step follows from the definition that $h^{(j)} = (\ms^{(j)})^{-1} \cdot h$ and the definition of $\ov{C}_1^{(k)}$, and the third step follows from $L \leq O(\phi^{-4} \|\mg \ma h\|_2^2)$ and $g^{(j,k)}_e = \Theta(2^{k} g_{\min})$. 

Finally we bound the runtime of $\textsc{Sample}(\cdot)$. The dominating term of the runtime is the time to call $\DS^{(j,k)}.\textsc{Sample}(h^{(j)}, \ov{C}_1^{(k)}, \ov{C}_2, C_3)$, which by \Cref{thm:heavy_hitter_balanced} is bounded by 
\[
O\Big(\big(\ov{C}_1^{(k)} (\phi^{-4} \|\ma^{(j,k)} \ms^{(j)} h^{(j)}\|_2^2 + \phi^{-4} \epsad \|(\md^{(j,k)})^{1/2}h\|_2^2) + \ov{C}_2 m + C_3 \|\ov{\tau}^{(j,k)}\|_1\big) C_0 \log m + n \log^3(m) \Big). 
\]
Using the same argument as how we proved Eq.~\eqref{eq:query_heavy_bound} and \eqref{eq:query_heavy_bound_total},  
\begin{align*}
\sum_{j=1}^T \sum_{k=1}^{\lceil\log(W_g)\rceil}  2^{2k} g_{\min}^2 \cdot \left(\phi^{-4} \|\ma^{(j,k)} \ms^{(j)} h^{(j)}\|_2^2 + \phi^{-4} \epsad \|(\md^{(j,k)})^{1/2}h\|_2^2\right) 
\leq &~ O\left( \phi^{-4} \|\mg \ma h\|_2^2 \right).
\end{align*}
Using this equation and that $L \geq \|\mg \ma h\|_2^2$, $\ov{C}_1^{(k)} = C_1 \cdot \frac{m}{\sqrt{n}} \cdot \frac{2^{2k} g_{\min}^2}{\phi^4 L}$, $\ov{C}_2 = \frac{C_2}{\sqrt{n}}$, and $T = O\big(\log^{10}(m) \log(W_{\eta}) \big)$, the total time of all $\DS^{(j,k)}.\textsc{Sample}$ calls is at most
\begin{align*}
O\left(C_0 C_1 \frac{m}{\sqrt{n}} \phi^{-8} + ( C_0 C_2 \frac{m}{\sqrt{n}} + n) \cdot \log(W_{\eta}) \log(W_g) \log^{13}(m) + C_0 C_3 \|\ov{\tau}\|_1 \log(m) \right),
\end{align*}
the claimed time complexity in the theorem statement then follows from $\phi = \log^{-3}(m)$. 
\end{proof}

\subsection{Heavy Hitters on Two-Sparse Matrices (Reduction to General Lossy Graphs)}\label{sec:heavy_hitter_reduction_two_sparse_to_lossy}

In this section we prove our main theorem \Cref{thm:heavy_hitter_two_sparse_general}. We prove \Cref{thm:heavy_hitter_two_sparse_general} through a standard reduction from two-sparse matrices to lossy graphs. The proof of the following lemma closely follows the reduction of \cite{h04}, and we include it here for completeness. %
\begin{lemma}[Reduction from two-sparse matrices to lossy graphs]\label{lem:reduction_two_sparse_to_lossy_graph}
Given any matrix $\ma \in \R^{m \times n}$ that has at most two non-zero entries per row, there exists a lossy graph $G = (V,E,\eta)$ with $|E| = m$ and $|V| = 2n$, and a diagonal edge weight matrix $\mg \in \R^{m \times m}$ that satisfy the following properties:
\begin{enumerate}
\item For any vector $h \in \R^n$, $\ma \cdot h = \mg \ov{\ma} \cdot \begin{bmatrix}
h \\ -h
\end{bmatrix}$, where $\ov{\ma} \in \R^{m \times 2n}$ is the incidence matrix of $G$.
\item 
The flow multipliers satisfy $\max_{e \in [m]} \eta_e \leq W_{\ma}^2$ (\Cref{def:W_A}), and the edge weights satisfy $\max_{e \in [m]} \max(|\mg_{e,e}|, \frac{1}{|\mg_{e,e}|}) \leq W_{\ma}$.
\end{enumerate}
\end{lemma}
\begin{proof}%
We construct $\ov{\ma}$ in two steps. First build $\wh{\ma}$, a matrix with a positive and negative entry in each row and second we build $\ov{\ma}$, the adjacency matrix, where the positive entry is $1$. First, we define $\wh{\ma}$ to be a matrix of size $m \times 2n$, where every row of $\ma$ corresponds to a row of $\wh{\ma}$. There are three cases for every row of $\ma$: (1) The row contains two non-zero entries of the same sign. (2) The row contains one positive entry and one negative entry. (3) The row contains one non-zero entry. Next we show how the corresponding rows of $\wh{\ma}$ are constructed for each of these three cases:
\begin{itemize}
\item For every row $e \in [m]$ of $\ma$ that has two non-zero entries of the same sign, denote these two entries as $\ma_{e,i} = \alpha$ and $\ma_{e,j} = \beta$.  Construct the $e$-th row of $\wh{\ma}$ as follows: Let $\wh{\ma}_{e, i} = \alpha$, $\wh{\ma}_{e, n+j} = -\beta$, and let all other entries be $0$.
\item  For every row $e \in [m]$ of $\ma$ that has one positive entry and one negative entry, denote these two entries as $\ma_{e,i} = \alpha > 0$ and $\ma_{e,j} = -\beta < 0$. Construct the $e$-th row of $\ov{\ma}$ as follows: Let $\wh{\ma}_{e, i} = \alpha$, $\wh{\ma}_{e, j} = -\beta$, and let all other entries be $0$.
\item  For every row $e \in [m]$ of $\ma$ that has exactly one non-zero entry, denote this entry as $\ma_{e,i} = \alpha$, and let $\sigma = \mathrm{sign}(\alpha) \in \{1, -1\}$. Construct the $e$-th row of $\wh{\ma}$ as follows: Let $\wh{\ma}_{e, i} = \alpha/2$, $\wh{\ma}_{e, n+i} = -\alpha/2$, and let all other entries be $0$.
\end{itemize}

It is easy to see that $\ma \cdot h = \wh{\ma} \cdot \begin{bmatrix}
h \\ -h
\end{bmatrix}$, and that every row of $\wh{\ma} \in \R^{m \times 2n}$ has exactly one positive entry and one negative entry. See \Cref{fig:reduction_two_sparse_to_lossy_graph} for an illustration of the construction of $\wh{\ma}$.

\begin{figure}[!ht]
    \centering
    \begin{equation*}
    \renewcommand{\arraystretch}{1.6}
    \ma = \begin{bmatrix}
    \alpha_1 &  & \beta_1 \\
     & -\alpha_2 & -\beta_2 \\
     -\beta_3 &  & \alpha_3 \\
     \alpha_4 & &
    \end{bmatrix}, ~~~~
    \wh{\ma} = 
    \begin{bmatrix}
    \begin{array}{ccc|ccc}
    \alpha_1 & & & & & -\beta_1 \\
    & -\alpha_2 & & & ~~ & \beta_2 \\
    -\beta_3 & & \alpha_3 & & & \\
    \alpha_4 / 2 & & & -\alpha_4 / 2 & &
    \end{array}
    \end{bmatrix}
    \end{equation*}
    \caption{An example of a two-sparse matrix $\ma \in \R^{m \times n}$ with $\alpha_1, \alpha_2, \alpha_3, \alpha_4, \beta_1, \beta_2, \beta_3 \geq 0$, and the matrix $\wh{\ma} \in \R^{m \times 2n}$ corresponding to it.}
    \label{fig:reduction_two_sparse_to_lossy_graph}
\end{figure}

We will next build $\ov{\ma}$ from $\wh{\ma}$.  For every $e \in [m]$, denote these two entries as $\wh{\ma}_{e,i} = \alpha > 0$ and $\wh{\ma}_{e,j} = -\beta < 0$, and define a diagonal weight matrix $\mg \in \R^{m \times m}$ and a vector $\eta \in \R^{m}$ as follows:
\begin{equation*}
\mg_{e,e} = \begin{cases}
\alpha & \text{ if } \beta \geq \alpha, \\
-\beta & \text{ if } \beta < \alpha.
\end{cases} ~~~
\eta_e = 
\begin{cases}
-\frac{\beta}{\alpha} & \text{ if } \beta \geq \alpha, \\
-\frac{\alpha}{\beta} & \text{ if } \beta < \alpha.
\end{cases}
\end{equation*}
We then define $\ov{\ma} \in \R^{m \times 2n}$ as $\ov{\ma} \defeq \mg^{-1} \cdot \wh{\ma}$, and it's straightforward to see each row of $\ov{\ma}$ contains exactly two non-zero entries: one is $1$, and the other is $-\eta_e \in (-W_\ma^2, -1]$. Also note that by the definition of $W_{\ma}$ in \Cref{def:W_A}, $\alpha, \beta, \alpha^{-1}, \beta^{-1} \leq W_{\ma}$, so $\eta_e \leq W_{\ma}^{2}$, $\|\mg\|_{\infty} \leq W_{\ma}$, and $\|\mg^{-1}\|_{\infty} \leq W_{\ma}$. Lastly, since $\ma \cdot h = \wh{\ma} \cdot \begin{bmatrix}
h \\ -h
\end{bmatrix}$, and $\ov{\ma} = \mg^{-1}\wh{\ma}$, we have that $\ma \cdot h = \mg\ov{\ma} \cdot \begin{bmatrix}
h \\ -h
\end{bmatrix}$ as desired. 
\end{proof}

Now we prove \Cref{thm:heavy_hitter_two_sparse_general} using \Cref{thm:heavy_hitter_general} and \Cref{lem:reduction_two_sparse_to_lossy_graph}. 

\begin{proof}[Proof of \Cref{thm:heavy_hitter_two_sparse_general}]
    Using \cref{lem:reduction_two_sparse_to_lossy_graph}, whenever we are given a two-sparse matrix $\ma$, we instead consider the problem on the lossy incidence matrix $\ov{\ma}$, with query vectors $\begin{bmatrix}
    h \\ -h
    \end{bmatrix}$, and edge weights $\mg$. We now bound the parameters of the new problem:
    \begin{itemize}
        \item $\log{(W_\eta)} \leq O(\log(W_\ma))$
        and $\log{(W_g)} \leq O(\log(W_\ma))$
        by \Cref{lem:reduction_two_sparse_to_lossy_graph}. 
        \item $\lambda_1$ is  Note that
        \begin{align*}
            \|\ma h\|_2 \geq \sqrt{\lambda_1}\|h\|_2 \implies \left\lVert\mg \ov{\ma} \cdot 
            \begin{bmatrix}
                h \\ -h
            \end{bmatrix}\right\rVert_2
            \geq \sqrt{\frac{1}{2}\lambda_1}
            \left\lVert\begin{bmatrix}
                h \\ -h
            \end{bmatrix}\right\rVert_2
        \end{align*}. 
    \end{itemize}
    It is easy to see that the operations of the two data structuress correspond exactly, since $\ma \cdot h = \mg \ov{\ma} \cdot \begin{bmatrix}
h \\ -h
\end{bmatrix}$. 
\end{proof}

\section{Two-Sparse Linear Programs and Generalized Min-Cost Flows}
\label{sec:IPM_using_data_structs}

The goal of this section is to prove our main results, \Cref{thm:two_sparse_LP} and \Cref{thm:generlize_min_cost_flow}, by implementing the IPM algorithm of \cite{bll+21} using the data structure we developed in \Cref{sec:general_heavy_hitters}. To make this section self-contained, we restate the relevant definitions and theorems from \cite{bll+21} that are used in our proofs. In particular, we restate several of their intermediate theorems to carefully track the dependence on the magnitude parameter $W$ and the error parameter $\delta$ in our specific settings of two-sparse LPs and generalized min-cost flows. 

This section is structured as follows. In \Cref{sec:other_data_structures}, we first restate the formal definitions of the three data structures required by the IPM algorithm of \cite{bll+21}, and show how to implement them for two-sparse matrices using our data structure from \Cref{thm:heavy_hitter_two_sparse_general}. 
In \Cref{sec:path_following}, we prove that using these three data structures, we can implement the algorithm of \cite{bll+21} to approximately follow the central path for two-sparse LPs. %
In \Cref{sec:initialization}, we describe the standard procedures for constructing an initial central point, and for rounding a final central point into an approximate LP solution. 
Finally in \Cref{sec:proof_main_theorems}, we combine all the results of this section to prove our main results, \Cref{thm:two_sparse_LP} and \Cref{thm:generlize_min_cost_flow}.

\subsection{Data Structures for the IPM}\label{sec:other_data_structures}

In this section we show how to use the data structure we developed in \Cref{thm:heavy_hitter_two_sparse_general} to build the the heavy hitter, sampler and inverse maintenance data structures required by the IPM algorithm of \cite{bll+21}.

\paragraph{Heavy Hitter.} We first restate the formal definition of the heavy hitter data structure of \cite{bll+21}.
\begin{definition}[Heavy hitter, Definition~3.1 of \cite{bln+20}]\label{def:heavyhitter}
For $c\in\R^m$ and $P,Q \in \R_{>0}$ with $nP \ge \|c\|_1 \ge P$ ,
we call a data structure with the following procedures a $(P,c,Q)$-\textsc{HeavyHitter} data structure:
\begin{itemize}
\item \textsc{Initialize}$(\ma \in \R^{m \times n}, g \in \R_{>0}^m)$:
	Let $\ma$ be a matrix with $c_i \ge \nnz(a_i)$, $\forall i \in [m]$ and $P \ge \nnz(\ma)$.
	The data structure initializes in $O(P)$ time.
\item \textsc{Scale}$(i \in [m], b \in \R)$:
	Sets $g_i \leftarrow b$ in $O(c_i)$ time.
\item \textsc{QueryHeavy}$(h \in \R^n, \epsilon \in (0,1) )$:
	Returns $I \subset [m]$ containing exactly those $i$ with $|(\mg \ma h)_i| \ge \epsilon$
	in $O( \epsilon^{-2} \| \mg \ma h \|_c^2 + Q )$ time. %
\end{itemize}
\end{definition}

Using \Cref{thm:heavy_hitter_two_sparse_general} we have the following $(P,c,Q)$-\textsc{HeavyHitter} data structure for two-sparse matrices.
\begin{corollary}[Heavy hitters on two-sparse matrices]\label{cor:heavy_hitter_two_sparse}
There exists a $(P,c,Q)$-\textsc{HeavyHitter} data structure for any two-sparse matrix $\ma \in \R^{m\times n}$ that succeeds with high probability with 
\begin{align*}
P = &~ \wt{O}\Big(m \log^2(W_g \lambda_{\min}(\ma^{\top} \ma)^{-1}) \log^{8}(W_g W_{\ma}) \Big), \\
c = &~ \wt{O}\left(\log^2(W_g \lambda_{\min}(\ma^{\top} \ma)^{-1}) \log^6(W_g W_{\ma})\right) \cdot \allone_m, \\
Q = &~ \wt{O}\left(n \log^2(W_g W_{\ma})\right),
\end{align*}
where $W_{\ma}$ is defined in \Cref{def:W_A}, and $W_g$ is the ratio of the largest to smallest non-zero entry in $g$.
\end{corollary}
\begin{proof}
We let $W'_g$ denote the ratio of the largest non-zero entry in any $g$ throughout the algorithm to smallest non-zero entry in any $g$ throughout the algorithm. We can without loss of generality assume that $W'_g \leq W_g^2$, because otherwise the largest non-zero entry of $g$ at some time $t$ is less than the smallest non-zero entry of $g$ at some other time $t'$, and we can simply re-build the data structure when this happens.

Let $g_{\min}$ be the minimum value of the vector $g$ throughout the algorithm. We implement the three operations using \Cref{thm:heavy_hitter_two_sparse_general} as follows: 
\begin{itemize}
\item \textsc{Initialize}: Given $\ma$ and $g$, we let $\ov{g} \defeq \frac{g}{g_{\min}}$, and initialize a data structure of \Cref{thm:heavy_hitter_two_sparse_general} with $\ov{\ma} \defeq \ov{\mg} \cdot \ma$.
\item \textsc{Scale}: Given $i$ and $b$, we call the data structure of \Cref{thm:heavy_hitter_two_sparse_general} to first delete row $i$ from $\ov{\ma}$, then insert a new row $\frac{b}{g_{\min}} \cdot a_i$, where $a_i$ denotes the $i$-th row of $\ma$.
\item \textsc{QueryHeavy}: Given query vector $h$ and parameter $\epsilon$, we call the \textsc{QueryHeavy} operation of the data structure of \Cref{thm:heavy_hitter_two_sparse_general} with $h$ and $\ov{\epsilon} \defeq \frac{\epsilon}{g_{\min}}$. Note that this is correct because $|(\mg \ma h)_i| \ge \epsilon$ if and only if $|(\ov{\ma} h)_i| \ge \ov{\epsilon}$, and $\ov{\epsilon}^{-2} \|\ov{\ma} h \|_2^2 = \epsilon^{-2} \|\mg \ma h \|_2^2$.
\end{itemize}
Next we bound the parameters $\lambda_1$ and $W_{\ov{\ma}}$ that show up in the time complexity of \Cref{thm:heavy_hitter_two_sparse_general}. First note that during the algorithm we always have $\ov{g} \geq 1$ and $\max_i g_i \leq W_g'$. At any time a query is performed we have that
\begin{align*}
\lambda_1(\ov{\ma}^\top \ov{\ma}) = &~ \lambda_1(\ma^\top \ov{\mg}^2 \ma) 
\leq W_g' \cdot \lambda_1(\ma^\top \ma).
\end{align*}
We also have that
\begin{align*}
W_{\ov{\ma}} = &~ \max_{i\in [n],e\in[m]:\ov{\ma}_{e,i} \neq 0}\left(\max(|\ov{\ma}_{e,i}|, \frac{1}{|\ov{\ma}_{e,i}|})\right) 
\leq W_g' \cdot W_A.
\end{align*}
The claimed time bounds of this Corollary then directly follows from \Cref{thm:heavy_hitter_two_sparse_general}.
\end{proof}

\paragraph{Sampler.} Next we restate the formal definition of the sampler data structure of \cite{bll+21}, which is defined using the following definition of a valid sampling distribution.

\begin{definition}[Valid sampling distribution, definition 4.13 in \cite{bll+21}]\label{def:valid_sampling_distribution}
    Given vector $\delta_r,\ma,g,\overline{\tau}$, %
    we say that a random diagonal matrix $\mr \in \R^{m\times m}$ is $C_{\valid}$-valid if it satisfies the following properties, for $\overline{\ma} = \overline{\mt}^{1/2} \mg \ma$. %
    We assume that $C_{\valid} \geq C_{\normrm}$. 
    \begin{itemize}
        \item (Expectation) We have that $\E[\mr] = \mI$. 
        \item (Variance) For all $i \in [m]$, we have that $\mathrm{Var}[\mr_{ii}(\delta_r)_i]\leq\frac{\eta|(\delta_r)_i|}{C^2_{\valid}}$ and $\E[\mr_{ii}^2] \leq 2\sigma(\overline{\ma})^{-1}_i$.
        \item (Covariance) For all $i\neq j$, we have that $\E[\mr_{ii}\mr_{jj}]\leq 2$. 
        \item (Maximum) With probability at least $1 - n^{-10}$ we have that $\|\mr\delta_r - \delta_r\|_{\infty}\leq\frac{\eta}{C^{2}_{\valid}}$.
        \item (Matrix approximation) We have that $\overline{\ma}^\top\mr\overline{\ma} \approx_\eta \overline{\ma}^\top\overline{\ma}$ with probability at least $1 - n^{-10}$. 
    \end{itemize}
\end{definition}

Using the notion of a valid sampling distribution, we now restate the definition of the sampler of \cite{bll+21}.

\begin{definition}[Sampler data structure, Definition 6.3 in \cite{bll+21}]\label{def:heavy_sampler}
    We call a data structure a $(P,c,Q)$-\textsc{HeavySampler} data structure if it supports the following operations:
    \begin{itemize}
        \item \textsc{Initialize}$(\ma\in \R^{m\times n},g\in \R^m_{>0},\overline{\tau}\in \R^m_{>0})$: Let $\ma$ be a matrix with $c_i \geq \nnz(a_i)$, $\forall i \in [m]$. The data structure initializes in $O(P)$ time. 
        \item \textsc{Scale}$(i,a,b):$ Sets $g_i \leftarrow a$ and $\overline{\tau}_i \leftarrow b$ in $O(c_i)$ amortized time. 
        \item \textsc{Sample}$(h \in \R^{m}):$ Returns a random diagonal matrix $\mr\in \R^{m\times m}$ that satisfies \Cref{def:valid_sampling_distribution} for $\delta_r = \mg\ma h$ with $\|\delta_r\|_2 \leq m/n$ and $\overline{\tau}\approx_{1/2}\sigma(\ov{\mt}^{1/2} \mg \ma)$ %
        in $O(Q)$ expected time. Furthermore, $\E[\nnz(\mr\ma)] = O(Q)$.
    \end{itemize}
\end{definition}

Our data structure of \Cref{thm:heavy_hitter_two_sparse_general} supports the proportional sampling of Lemma~4.42 of \cite{bll+21}, and the output $\mr$ satisfies \Cref{def:valid_sampling_distribution} when choosing $C_0 = 100 C_{\valid}^4 \gamma^{-2} \log m$, where $\gamma = \frac{\epsilon^2}{C^2 \log(Cm/\epsilon^2)}$, $\epsilon=\frac{1}{4C\log(m/n)}$, and $C  \geq 100$ is a constant.
Similar to how \cite{bll+21} proved its sampler for graphs in Corollary~F.4, we also use Corollary~4.42 of \cite{bll+21} and the \textsc{ScaleTau} and \textsc{Sample} operations of the data structure of \Cref{thm:heavy_hitter_two_sparse_general} to obtain the following sampler for two-sparse matrices.
\begin{corollary}[Sampler for two-sparse matrices]\label{cor:sampler}
There exists a $(P,c,Q)$-\textsc{HeavySampler} data structure for any two-sparse matrix $\ma \in \R^{m\times n}$ where 
\begin{align*}
P = &~ \wt{O}\Big(m \log^2(W_g \lambda_{\min}(\ma^{\top} \ma)^{-1}) \log^8(W_g W_{\ma}) \Big), \\
c = &~ \wt{O}\left(\log^2(W_g \lambda_{\min}(\ma^{\top} \ma)^{-1}) \log^6(W_g W_{\ma})\right) \cdot \allone_m, \\
Q = &~ \wt{O}\left((\frac{m}{\sqrt{n}} + n) \cdot \log^2(W_g W_{\ma})\right),
\end{align*}
where $W_{\ma}$ is defined in \Cref{def:W_A}, and $W_g$ is the ratio of the largest to smallest non-zero entry in $g$.
\end{corollary}

\paragraph{Inverse Maintenance.} Finally we restate the formal definition of the inverse maintenance data structure of \cite{bll+21}.

\begin{definition}[Inverse maintenance, Definition~6.2 of \cite{bln+20}]\label{def:inverse_maintenance}
We call a data structure a $(P,c,Q)$-\textsc{InverseMaintenance} data structure, if it supports the following operations:
\begin{itemize}
\item \textsc{Initialize}$(\ma \in \R^{m \times n}, v\in \R^m, \bar{\sigma}\in \R^m)$ The data structure initializes in $O(P)$ time for $\bar{\sigma} \geq \frac{1}{2}\sigma(\mv^{1/2}\ma)$ and $\|\bar{\sigma}\|_1 = O(n)$. 
\item \textsc{Scale}$(i \in [m],a,b)$:
        Set $v_i \leftarrow a$ and $\bar{\sigma}_i \leftarrow b$ in $O(c_i)$ time.
\item \textsc{Solve}$(\bar{v} \in \R^m, b,\epsilon \in (0,1) )$: Assume $\bar{\sigma} \geq \frac{1}{2}\sigma(\mv^{1/2}\ma)$ and the given $\bar{v}$ satisfies $\ma^\top\mv\ma \approx_{1/2}\ma^\top\overline{\mv}\ma$. Then \textsc{Solve} returns $\mh^{-1} b$ for $\mh \approx_{\epsilon} \ma^\top \overline{\mv}\ma$ in $O\big( (Q + \nnz(\ov{\mv} \ma)) \cdot \log\epsilon^{-1}\big)$ time. Furthermore, for the same $\bar{v}$ and $\epsilon$, the algorithm uses the same $\mh$. %
\end{itemize}
For the complexity bounds one may further assume the following stability assumption: Let $\bar{v}^{(1)},\bar{v}^{(1)},...$ be the sequence of inputs given to \textsc{Solve}, then there exists a sequence $\tilde{v}^{(1)},\tilde{v}^{(2)},...$ such that for all $t > 0$:
\begin{align*}
    \bar{v}^{(t)} \in (1 \pm 1/(100\log n))\tilde{v}^{(t)} \text{ and }\|(\tilde{v}^{(t)})^{-1}(\tilde{v}^{(t)}-\tilde{v}^{(t+1)}\|_{\bar{\sigma}}=O(1).
\end{align*}
\end{definition}

We show that there exists such an inverse maintenance data structure for two-sparse matrices.
\begin{theorem}[Inverse maintenance data structure for two-sparse matrices]\label{thm:inverse_maintenance}
There exists a $(P,c,Q)$-\textsc{InverseMaintenance} data structure for any two-sparse matrix $\ma \in \R^{m\times n}$ where $c_i = O(1)$, $P = \wt{O}(m)$, and $Q = \wt{O}(n\log(\kappa))$ where $\kappa = \frac{\lambda_{\max}(\ma^{\top} \mv \ma)}{\lambda_{\min}(\ma^{\top} \mv \ma)}$.
\end{theorem}

Our proof uses the following fast Laplacian solver by \cite{st04} (see \Cref{thm:fast_Laplacian_solver}), where we use the stronger version that holds for any symmetric diagonally dominant (SDD) matrix. We will also use the M-matrix scaling theorem by \cite{ds08} (see \Cref{thm:fastmsolver}), which shows the stronger statement that given any matrix $\mm = \ma^\top \ma$ where $\ma$ has at most two non-zeros entries per row, there is an algorithm that can output a scaling matrix $\md$ such that $\md \mm \md$ is SDD. Note that the off-diagonal entries of $\mm$ can be positive.
\begin{theorem}[Fast Laplacian solver, \cite{st04}]\label{thm:fast_Laplacian_solver}
There is an algorithm that takes as input a two-sparse matrix $\ma \in \R^{m \times n}$, a non-negative diagonal matrix $\md \in \R^{m \times m}$, a vector $b \in \R^n$, and a parameter $\epsilon > 0$, such that $\ma^{\top} \md \ma$ is a symmetric diagonally dominant (SDD) matrix, and the algorithm returns a vector $\ov{x} = \mz \cdot b$ such that $\mz \approx_{\epsilon} (\ma^{\top} \md \ma)^{\dagger}$. The algorithm runs in $\wt{O}(m \log(\epsilon^{-1}))$ time, and succeeds with high probability. Furthermore, for the same $\ma$, $\md$, and $\epsilon$, the algorithm uses the same $\mz$.
\end{theorem}

\begin{theorem}[M-Matrix Scaling, Theorem 4.5 of \cite{ds08}]\label{thm:fastmsolver}
    There is an algorithm that takes as input any M-matrix $\mm \in \mathbb{R}^{n\times n}$ and its factorization $\mm = \ma^\top \ma$, where $\ma \in \mathbb{R}^{m\times n}$ is two-sparse, along with upper and lower bounds $\lambda_{\max}$ and $\lambda_{\min}$ for the eigenvalues of the matrix $\ma$, and the algorithm returns a positive diagonal matrix $\md \in \R^{n \times n}$ such that the matrix $\md \mm \md$ is a symmetric diagonally dominant (SDD) matrix.
    This algorithm runs in expected time $\wt{O}(m\log \kappa)$, where $\kappa \defeq \lambda_{\max}/\lambda_{\min}$.
\end{theorem}

Now we are ready to prove \Cref{thm:inverse_maintenance} using \Cref{thm:fast_Laplacian_solver} and \Cref{thm:fastmsolver}.
\begin{proof}[Proof of \Cref{thm:inverse_maintenance}]
In the \textsc{Initialize} operation, the algorithm computes a random diagonal matrix $\ms \in \R^{m \times m}$ where for each $ \in [m]$, $\ms_{i,i}$ is independently set to $\frac{1}{p_i}$ with probability $p_i = \min\{1, \Theta(\ov{\sigma}_i \log n)\}$, and it is set to $0$ otherwise. The guarantees of leverage score sampling implies that $\ma^{\top} \mv^{1/2} \ms \mv^{1/2} \ma \approx_{0.1} \ma^{\top} \mv \ma$, and with high probability $\ms$ has $O(n \log n)$ non-zero entries on its diagonal. This can be computed in $\wt{O}(m)$ time.

In the \textsc{Scale} operation with input $i \in [m]$ and $a,b \in \R$, we re-sample the $i$-th entry of $\ms$: We again set $\ms_{i,i} = \frac{1}{p_i}$ with probability $p_i = \min\{1, \Theta(\ov{\sigma}_i \log n)\}$, and set it to $0$ otherwise. This step takes $O(1)$ time.

In a \textsc{Solve} operation, since we maintained $\ms$, we can use $(\ma^{\top} \mv^{1/2} \ms \mv^{1/2} \ma)^{-1}$ as a preconditioner for $\ma^{\top} \ov{\mv} \ma$.  First note that using \Cref{thm:fastmsolver} and \Cref{thm:fast_Laplacian_solver} we can solve any linear system with $\ma^{\top} \mv^{1/2} \ms \mv^{1/2} \ma$ to constant accuracy in $\wt{O}(n \log \kappa)$ time, where $\kappa = \frac{\lambda_{\max}(\ma^{\top} \mv \ma)}{\lambda_{\min}(\ma^{\top} \mv \ma)} = \Theta(\frac{\lambda_{\max}(\ma^{\top} \mv^{1/2} \ms \mv^{1/2} \ma)}{\lambda_{\min}(\ma^{\top} \mv^{1/2} \ms \mv^{1/2} \ma)})$. We then use preconditioned Richardson iterations 
\[
x^{(k+1)} = x^{(k)} + (\ma^{\top} \mv^{1/2} \ms \mv^{1/2} \ma)^{-1} (b - \ma^{\top} \ov{V} \ma x^{(k)}).
\]
It takes $O(\log \epsilon^{-1})$ iterations to solve to $\epsilon$ accuracy. So the total runtime of \textsc{Solve} is $\wt{O}\big( (n \log \kappa + \nnz(\ov{\mv} \ma)) \log \epsilon^{-1} \big)$.
\end{proof}

\subsection{Approximately Following the Central Path}\label{sec:path_following} %
In this section we prove that using the heavy hitter (\Cref{cor:heavy_hitter_two_sparse}), sampler (\Cref{cor:sampler}), and inverse maintenance (\Cref{thm:inverse_maintenance}) data structures, we can implement the algorithm of \cite{bll+21} to approximately follow the central path for two-sparse LPs. This path following algorithm takes an initial central point as input and produces another central point with improved parameters. We leave the problems of finding the initial central point, and rounding the final central point to an LP solution to next sections.

Recall that we consider an LP of the following form, where $\ma \in \R^{m \times n}$ is a two-sparse matrix: 
\begin{equation*}
\begin{aligned}
\min_{x \in \R^m} &~ c^{\top} x \\
\text{s.t.} &~ \ma^{\top} x = b \\
&~ \ell \leq x \leq u
\end{aligned}
\end{equation*}

\paragraph{Path following algorithm from \cite{bll+21}.}
We first restate the guarantees of the \textsc{PathFollowing} algorithm of \cite{bll+21}. We use the following definitions from Section~4 of \cite{bll+21}. For the two-sided constraints $\ell_i \leq x_i \leq u_i$ for all $i \in [m]$, barrier function is defined as
\begin{align*}
\phi(x) = \sum_{i \in [m]} \phi_i(x_i), ~ \text{where} ~ \phi_i(x_i) = - \log(x_i - \ell_i) - \log(u_i - x_i).
\end{align*}
We also use the definition of regularized Lewis weights for a matrix. Define $p = 1 - \frac{1}{4 \log(4m/n)}$. For every matrix $\ma \in \R^{m \times n}$, its regularized $\ell_p$-Lewis weights $w(\ma) \in \R^m_{>0}$ is defined as the solution of
\[
w(\ma) = \sigma( \mw^{\frac{1}{2} - \frac{1}{p}} \ma ) + \frac{n}{m}, \text{ where } \mw \defeq \mdiag(w(\ma)) .
\]
Given any $x \in \R^m$, define the central path weights as
\[
\tau(x) = w\Big(\mdiag(\phi''(x)^{-\frac{1}{2}}) \ma\Big).
\]
We use the following centrality definition from \cite{bll+21}.
\begin{definition}[$\epsilon$-centered point, Definition~4.7 of \cite{bll+21}]\label{def:eps_centered}
We say that $(x,s,\mu) \in \R^m \times \R^m \times \R^m_{>0}$ is $\epsilon$-centered for $\epsilon \in (0,1/80]$ if the following properties hold, where $C$ is a constant such that $C \ge 100$, and $C_{\normrm} \defeq C / (1-p)$, $\gamma \defeq \frac{\epsilon^2}{C^2 \log(Cm/\epsilon^2)}$.
\begin{enumerate}
\item (Approximate centrality) $\left\| \frac{s+\mu \tau(x)\phi'(x)}{\mu\tau(x)\sqrt{\phi''(x)}}\right\|_\infty \le \epsilon.$
\item (Dual Feasibility) There exists a vector $z \in \R^n$ with $\ma z+s=c$.
\item (Approximate Feasibility) $\| \ma^\top x - b \|_{(\ma^\top(\mt(x)\Phi''(x))^{-1}\ma)^{-1}} \le \epsilon\gamma/C_{\normrm}$.
\end{enumerate}
\end{definition}

The \textsc{PathFollowing} algorithm of \cite{bll+21} satisfies the following guarantees.
\begin{theorem}[\textsc{PathFollowing}, Lemma~4.12 and Theorem~6.1 of \cite{bll+21}]\label{thm:path_following}
Consider a linear program with $\ma\in\R^{m\times n}$, $\ell,u,c \in \R^m$, and $b \in \R^n$:
\[
\Pi:\min_{\substack{\ma^{\top}x=b\\
        \ell\le x\le u
    }
}c^{\top}x.
\]
Let $\epsilon=\frac{1}{4C\log(m/n)}$ for a large enough constant $C$. Given an $\epsilon$-centered initial point $(x^{(\init)},s^{(\init)},\mu^{(\init)})$ for $\Pi$ and a target $\mu^{(\final)}$, there exists an algorithm \textsc{PathFollowing}$(\ma, b, \ell, u, c, \mu, \mu^{(\final)})$ that returns an $\epsilon$-centered point $(x^{(\final)},s^{(\final)},\mu^{(\final)})$.

Assume there exists a $(P,c,Q)$-\textsc{HeavyHitter} (\Cref{def:heavyhitter}), a $(P,c,Q)$-\textsc{InverseMaintenance} (\Cref{def:inverse_maintenance}), and a $(P,c,Q)$-\textsc{HeavySampler} (\Cref{def:heavy_sampler}), then the total time of \textsc{PathFollowing} is
\begin{align*}
\widetilde{O}\left(\left(\sqrt{P\|c\|_1} + \sqrt{n} \Big(Q + n\cdot\max_{i}\nnz(a_i)\Big)\right) \log\frac{\mu^{\mathrm{(init)}}}{\mu^{\mathrm{(final)}}}\right).
\end{align*}
\end{theorem}

Furthermore, the algorithm \textsc{PathFollowing} also guarantees that the parameters along the central path are bounded.
\begin{lemma}[Parameter changes of \textsc{PathFollowing}, Lemma 4.46 of \cite{bll+21}]\label{lem:parameter_changes}
For $\ma \in \R^{m\times n}, b \in \R^n, c \in \R^m$ and $\ell,u \in \R^m$, assume that the point $x^{(\init)}=(\ell+u)/2$ is feasible, i.e. $\ma^\top x^{(\init)} = b$. Let $W$ be the ratio of the largest to smallest entry of $\phi''(x^{(\init)})^{1/2}$, and let $W'$ be the ratio of the largest to smallest entry of $\phi''(x)^{1/2}$ encountered in \textsc{PathFollowing}. Then:
\begin{align*}
    \log W' = \widetilde{O}\left(\log W + \log(1/\mu^{\mathrm{(final)}}) + \log \|c\|_\infty\right).
\end{align*}
\end{lemma}

\paragraph{Approximately following the central path for two-sparse LPs.} Next we use the data structures for two-sparse matrices to efficiently implement the \textsc{PathFollowing} algorithm.
\begin{theorem}[Path following for two-sparse LPs]\label{thm:path_following_two_sparse}
Consider a linear program with $\ma \in \R^{m\times n}$ that has at most two non zero entries per row, %
$\ell,u,c \in \R^m$, and $b \in \R^n$:
\[
\Pi:\min_{\substack{\ma^{\top}x=b\\
        \ell\le x\le u
    }
}c^{\top}x.
\]
Let $\epsilon=\frac{1}{4C\log(m/n)}$ for a large enough constant $C$. Given an $\epsilon$-centered initial point $(x^{(\init)},s^{(\init)},\mu^{(\init)})$ for $\Pi$ and a target $\mu^{(\final)}$, the algorithm \textsc{PathFollowing}$(\ma, b, \ell, u, c, \mu^{(\init)}, \mu^{(\final)})$ returns an $\epsilon$-centered point $(x^{(\final)},s^{(\final)},\mu^{(\final)})$ with high probability in time
\begin{align*}
&~ \widetilde{O}\left( m \cdot  \left(\log(\lambda_{\min}(\ma^{\top} \ma)^{-1}) + \log(W_{\phi}) + \log(\frac{1}{\mu^{\mathrm{(final)}}}) + \log \|c\|_\infty + \log(W_{\ma})\right)^{10} \log\frac{\mu^{\mathrm{(init)}}}{\mu^{\mathrm{(final)}}} \right) \\
&~ + \widetilde{O} \left( n^{1.5} \cdot  \left(\log(\frac{\lambda_{\max}(\ma^{\top} \ma)}{\lambda_{\min}(\ma^{\top} \ma)}) + \log(W_{\phi}) + \log(\frac{1}{\mu^{\mathrm{(final)}}}) + \log \|c\|_\infty + \log(W_{\ma})\right)^2 \log\frac{\mu^{\mathrm{(init)}}}{\mu^{\mathrm{(final)}}}\right).
\end{align*}
where $W_{\ma}$ is defined in \Cref{def:W_A}, and $W$ is the ratio of the largest to smallest entry of $\phi''(x^{(\init)})^{1/2}$. 
\end{theorem}
\begin{proof}
Using \Cref{thm:path_following}, we can implement the \textsc{PathFollowing} algorithm using the three data structures for two-sparse matrices: heavy hitter (\Cref{thm:heavy_hitter_general}), sampler (\Cref{cor:sampler}), and inverse maintenance (\Cref{thm:inverse_maintenance}), and this gives a runtime of
\begin{align*}
\widetilde{O}\left( \Big(m \cdot \log^2(W_g \lambda_{\min}(\ma^{\top} \ma)^{-1}) \log^{8}(W_g W_{\ma}) + n^{1.5} \cdot \big(\log^2(W_g W_{\ma}) + \log(\kappa) \big) \Big) \log\frac{\mu^{\mathrm{(init)}}}{\mu^{\mathrm{(final)}}}\right),
\end{align*}
where $W_g$ is the ratio of the largest to smallest non-zero entry in all scaling vectors $g$ of the data structures, and $\kappa = \frac{\lambda_{\max}(\ma^{\top} \mv \ma)}{\lambda_{\min}(\ma^{\top} \mv \ma)}$, where $\mv$ is the scaling matrix of \textsc{InverseMaintenance}.

Let $W'$ be the ratio of the largest to smallest entry of $\phi''(x)^{1/2}$ encountered in \textsc{PathFollowing}, note that $W_g \leq O(W')$, and the ratio of the largest to smallest entry of $\mv$ of \textsc{InverseMaintenance} is also at most $O(W')$, so 
\begin{align*}
\kappa = \frac{\lambda_{\max}(\ma^{\top} \mv \ma)}{\lambda_{\min}(\ma^{\top} \mv \ma)} \leq O\left(\frac{\lambda_{\max}(\ma^{\top} \ma) \max_i v_i}{\lambda_{\max}(\ma^{\top} \ma) \min_i v_i}\right) \leq O\left(\frac{\lambda_{\max}(\ma^{\top} \ma)}{\lambda_{\min}(\ma^{\top} \ma)}\cdot W'\right).
\end{align*}
\Cref{lem:parameter_changes} implies that $\log W' = \widetilde{O}\left(\log(W_{\phi}) + \log(1/\mu^{\mathrm{(final)}}) + \log \|c\|_\infty\right)$, where $W_{\phi}$ is the ratio of the largest to smallest entry of $\phi''(x^{(\init)})^{1/2}$.  

So the total runtime is bounded by
\begin{align*}
&~ \widetilde{O}\left( m \cdot  \left(\log(\lambda_{\min}(\ma^{\top} \ma)^{-1}) + \log(W_{\phi}) + \log(\frac{1}{\mu^{\mathrm{(final)}}}) + \log \|c\|_\infty + \log(W_{\ma})\right)^{10} \log\frac{\mu^{\mathrm{(init)}}}{\mu^{\mathrm{(final)}}} \right) \\
&~ + \widetilde{O} \left( n^{1.5} \cdot  \left(\log(\frac{\lambda_{\max}(\ma^{\top} \ma)}{\lambda_{\min}(\ma^{\top} \ma)}) + \log(W_{\phi}) + \log(\frac{1}{\mu^{\mathrm{(final)}}}) + \log \|c\|_\infty + \log(W_{\ma})\right)^2 \log\frac{\mu^{\mathrm{(init)}}}{\mu^{\mathrm{(final)}}}\right). \qedhere
\end{align*}
\end{proof}

\subsection{Initial and Final Point of LP}\label{sec:initialization}
In this section we first show how to obtain an initial central point as required by the \textsc{PathFollowing} algorithm  of \Cref{thm:path_following_two_sparse}.

\begin{lemma}[LP initialization, Section 8.2 of \cite{bll+21}]\label{lem:LP_initialization}
Given any matrix $\ma \in \R^{m \times n}$, any vector $b \in \R^n$, any vector $c \in \R^n$, any lower bound and upper bound vectors $\ell \leq u \in \R^m$, and any accuracy parameter $\delta \in (0,0.1)$, consider the following linear program:
\begin{align}\label{eq:LP}
\min_{\substack{\ma^{\top} x = b \\ \ell \leq x \leq u}} c^{\top} x.
\end{align}
Let $W \defeq \max\left\{\|c\|_{\infty}, \|\ma\|_{\infty}, \|b\|_{\infty}, \|u\|_{\infty}, \|\ell\|_{\infty}, \frac{\max_{i \in [m]} (u_i - \ell_i)}{\min_{i \in [m]} (u_i - \ell_i)}\right\}$, $\delta' \defeq \frac{\delta}{10 m W^2}\cdot \min\{1, \|c\|_1\}$, $\Xi \defeq \max_{i \in [m]} |u_i - \ell_i|$, and define the initial points by %
\begin{align*}
x^{(\init)} \defeq (\ell + u) / 2, ~~ \beta \defeq \|b - \ma^{\top} x^{(\init)}\|_{\infty} / \Xi + 1, ~~ \wt{x}^{(\init)} \defeq \frac{1}{\beta} |b - \ma^{\top} x^{(\init)}|.
\end{align*}
Finally, define a vector $\sigma = \mathrm{sign}(b - \ma^{\top} x^{(\init)}) \in \R^n$.

There exists a modified linear program $\min_{\substack{\ov{\ma}^{\top} x = \ov{b} \\ \ov{\ell} \leq x \leq \ov{u}}} \ov{c}^{\top} x$ with %
\begin{align*}
\ov{\ma} = 
\begin{bmatrix}
\ma  \\
\beta \cdot \mdiag(\sigma) 
\end{bmatrix},~ 
\ov{b} = b,~ 
\ov{c} = 
\begin{bmatrix}
c \\
\frac{2 \|c\|_1}{\delta'} \cdot \allone_n
\end{bmatrix},~ 
\ov{u} = 
\begin{bmatrix}
u \\
2 \wt{x}^{(\init)}
\end{bmatrix}, ~
\ov{\ell} = 
\begin{bmatrix}
\ell \\
0
\end{bmatrix},
\end{align*}
that satisfies the following guarantees:
\begin{enumerate}
\item Define $\ov{x}^{(\init)} \defeq \begin{bmatrix}
x^{(\init)} \\ \wt{x}^{(\init)}
\end{bmatrix}$. The point $(\ov{x}^{(\init)}, \ov{c}, \mu)$ is $\epsilon$-centered for $\mu = \frac{4 m \|c\|_1 \Xi}{\epsilon \delta'}$.
\item Assume the linear program \eqref{eq:LP} has a feasible solution. For any $\ov{x}^{(\final)} \defeq \begin{bmatrix}
x^{(\final)} \\ \wt{x}^{(\final)}
\end{bmatrix}$ that satisfies $\ov{\ma}^{\top} x^{(\final)} = \ov{b}$, $\ov{\ell} \leq \ov{x}^{(\final)} \leq \ov{u}$, and $\ov{c}^{\top} \ov{x}^{(\final)} \leq \min_{\substack{\ov{\ma}^{\top} x = \ov{b} \\ \ov{\ell} \leq x \leq \ov{u}}} \ov{c}^{\top} x + \delta$, the point $x^{(\final)}$ satisfies
\begin{align*}
c^{\top} x^{(\final)} \leq \min_{\substack{\ma^{\top} x = b \\ \ell \leq x \leq u}} c^{\top} x + \delta, ~~ \text{and} ~~ \|\ma^{\top} x^{(\final)} - b\|_{\infty} \leq \delta.
\end{align*}
\item $\ov{\ma}$ satisfies $\sigma_{\min}(\ov{\ma}) \geq 1$.
\end{enumerate}
\end{lemma}
Our lemma differs from the initialization techniques of \cite{bll+21} in that we also need to lower bound the least singular value of $\ov{\ma}$. For completeness we include a proof in \Cref{sec:missing_proofs}. %

Next we restate the following lemma from \cite{bll+21} that shows how to round an $\epsilon$-centered point produced by the $\textsc{PathFollowing}$ algorithm into an approximate LP solution.
\begin{lemma}[Final point, Lemma~4.11 of \cite{bll+21}]\label{lem:final_point} 
Given an $\epsilon$-centered point $(x,s,\mu)$ where $\epsilon\le1/80$, we can compute a point $(x^{(\final)},s^{(\final)})$ satisfying 
\begin{enumerate}
	\item $\ma^{\top}x^{(\final)}=b$, $s^{(\final)}=\ma y+c$ for some $y$. 
	\item $c^{\top}x^{(\final)}-\min_{\substack{\ma^{\top}x=b\\
			\ell \le x \le u
		}
	}c^{\top}x\lesssim n\mu.$ 
\end{enumerate}
The algorithm takes $O(\nnz(\ma))$ time plus the time for solving
a linear system on $\ma^{\top}\md\ma$ where $\md$ is a diagonal
matrix.
\end{lemma}

\subsection{Proof of Main Theorems}\label{sec:proof_main_theorems}

In this section we prove our main results \Cref{thm:two_sparse_LP}, \Cref{thm:generlize_min_cost_flow}, and \Cref{thm:generlize_maxflow}. %
\begin{proof}[Proof of \Cref{thm:two_sparse_LP}]

Given any LP $\min_{\substack{\ma^{\top} x = b \\ \ell \leq x \leq u}} c^{\top} x$, we use \Cref{lem:LP_initialization} to construct a modified LP instance with $(\ov{\ma}, \ov{b}, \ov{c}, \ov{\ell}, \ov{u})$ where $\ov{\ma} = \begin{bmatrix}
\ma \\ \beta \mdiag(\sigma)
\end{bmatrix}$ for $\beta > 1$ and $\mdiag(\sigma)$ a diagonal matrix with $1$s and $-1$s on the diagonal, together with an initial point $(\ov{x}^{(\init)}, \ov{c}, \mu^{(\init)})$ that is $\epsilon$-centered. The new inputs are bounded by
\begin{align*}
& \mu^{(\init)} \leq \poly(m) \cdot \poly(W) \cdot \delta^{-1}, ~~~
\beta \leq \poly(m) \cdot \poly(W), \\
& \lambda_{\min}(\ov{\ma}^{\top} \ov{\ma}) \geq 1, ~~~
\lambda_{\max}(\ov{\ma}^{\top} \ov{\ma}) \leq \beta^2 \|\ma\|_2^2 \leq \poly(m) \cdot \poly(W). 
\end{align*}
And since $x^{(\init)} = (\ell+u) / 2$ and $\phi_i(x_i) = - \log(x_i - \ell_i) - \log(u_i - x_i)$, the entries of $\phi''(x^{(\init)})^{1/2}$ are $\frac{8}{(u_i - \ell_i)^2}$ $W_{\phi}$, so 
\[
W_{\phi} \leq \poly\left(\frac{\max_i(u_i-\ell_i)}{\min_i(u_i-\ell_i)}\right) \leq \poly(W).
\]

We set $\mu^{(\final)} = \frac{\delta}{\poly(m)}$, and using \Cref{thm:path_following_two_sparse}, we can solve the modified LP $(\ov{\ma}, \ov{b}, \ov{c}, \ov{\ell}, \ov{u})$ to obtain a $\epsilon$-centered point $(x^{(\final)}, s^{(\final)}, \mu^{(\final)})$ with high probability in time
\[
\wt{O}\left( (m + n^{1.5}) \cdot \log^{10}(\frac{W}{\delta})\right).
\]
Then using \Cref{lem:final_point} we convert it into a solution $x^{(\final)}$ that satisfies 
\begin{equation*}
\|\ma^{\top}x^{(\final)}-b\|_{\infty}\le\delta\enspace\text{ and }\enspace\ell\le x^{(\final)}\le u \enspace\text{ and }\enspace c^{\top}x^{(\final)}\le\min_{\substack{\ma^{\top}x=b\\
        \ell \le x \le u
    }
}c^{\top}x+\delta. \qedhere
\end{equation*}
\end{proof}

Next we prove \Cref{thm:generlize_min_cost_flow}, which directly follows from \Cref{thm:two_sparse_LP}.
\begin{proof}[Proof of \Cref{thm:generlize_min_cost_flow}]
Given any lossy graph $G = (V,E,\gamma)$, costs $c \in \R^E$, capacities $u\in\R_{\geq 0}^{E}$, demands $d\in\R^{V}$, and $\delta>0$, our goal is to solve the LP
\begin{align*}
\min_{\substack{\mb_G^{\top} f = d \\ \allzero \leq f \leq u}} c^{\top} f.
\end{align*}
As a corollary of \Cref{thm:two_sparse_LP}, we immediately have that we can solve this LP to $\delta$ error in time 
\begin{align*}
\wt{O}\left( (m + n^{1.5}) \cdot \log^{10}(\frac{W}{\delta})\right),
\end{align*}
where $W \defeq \max(W_{\mb_G}, \|c\|_{\infty},\|d\|_{\infty},\|u\|_{\infty},\frac{\max_{e}u_{e}}{\min_{e}u_{e}})$, $W_{\mb_G} \defeq \max_{e,i:(\mb_G)_{e,i} \neq 0}(\max(|(\mb_G)_{e,i}|, \frac{1}{|(\mb_G)_{e,i}|}))$. 

Note that by definition $W_{\mb_G} \leq \max(\|\gamma\|_{\infty}, \|\gamma^{-1}\|_{\infty})$, and this upper bound together with $\mb_G^{\top} f = d$ implies that $\|d\|_{\infty} \leq O(\min\{1,\|\gamma\|_{\infty}\} \cdot \|u\|_{\infty})$. So we can define 
\[
W' \defeq \max\left(\|\gamma\|_{\infty}, \|\gamma^{-1}\|_{\infty}, \|c\|_{\infty},\|u\|_{\infty},\frac{\max_e u_e}{\min_e u_e}\right)
\]
where $W' \leq O(W^2)$, and the algorithm runs in time $\wt{O}\left( (m + n^{1.5}) \cdot \log^{10}(\frac{W'}{\delta})\right)$.
\end{proof}

Next we prove \Cref{thm:generlize_maxflow}, which also directly follows from \Cref{thm:two_sparse_LP}, and we use the technique of \cite{ds08} to make the flow feasible for the special case of lossy maxflow.
\begin{proof}[Proof of \Cref{thm:generlize_maxflow}]
Given any lossy graph $G = (V,E,\gamma)$, a source vertex $s \in V$ and a sink vertex $t \in V$, capacities $u\in\R_{\geq 0}^{E}$, and $\delta>0$, our goal is to solve the LP
\[
\min_{\substack{\mb_{G \backslash \{s,t\}}^{\top} f = 0 \\ \allzero \leq f \leq u}} (\mb_G^{\top} f)_t.
\]
As a corollary of \Cref{thm:two_sparse_LP}, this LP can be solved to $\delta$ error in time 
\begin{align*}
\wt{O}\left( (m + n^{1.5}) \cdot \log^{10}(\frac{W}{\delta})\right),
\end{align*}
where $W \defeq \max(W_{\mb_G},\|u\|_{\infty},\frac{\max_{e}u_{e}}{\min_{e}u_{e}})$, and $W_{\mb_G} \defeq \max_{e,i:(\mb_G)_{e,i} \neq 0}(\max(|(\mb_G)_{e,i}|, \frac{1}{|(\mb_G)_{e,i}|}))$, and by definition $W_{\mb_G} \leq \max(\|\gamma\|_{\infty}, \|\gamma^{-1}\|_{\infty})$.

In the special case of lossy maxflow with $\gamma \leq \allone_m$, we further apply Lemma~3.2 of \cite{ds08} to convert this $\delta$-approximately feasible flow into an exactly feasible flow in $\wt{O}(m)$ additional time. 
\end{proof}

\section{Conclusion and Open Problems}

In this paper we provide a randomized $\wt{O}( (m + n^{1.5}) \cdot \polylog(\frac{W}{\delta}))$ time algorithm for solving two-sparse LPs. As a corollary, we obtain nearly-linear time algorithms for the generalized maximum flow and generalized minimum cost flow problems on moderately dense graphs where $m \geq n^{1.5}$. 

Perhaps the most immediate open problem is whether our running times can be further improved. A key open problem is to close the gap between our running times for lossy flow and state-of-the-art running times for maximum flow, namely, improving the $\wt{O}( (m + n^{1.5}) \cdot \polylog(\frac{W}{\delta}))$ runtime to an almost-linear runtime of $m^{1+o(1)} \cdot \polylog(\frac{W}{\delta})$. 

Another natural direction for future research is to further improve our dependence on the conditioning of the problem, e.g., improving the exponent in the $\polylog(\frac{W}{\delta})$ factors, and expressing the dependence to scale-invariant condition measures rather than the current $W_{\ma} = \max_{i,j:\ma_{i,j} \neq 0}(\max(|\ma_{i,j}|, \frac{1}{|\ma_{i,j}|}))$.

Finally, investigating further spectral characterizations of lossy graphs and finding further applications of the spectral results proved in this paper are interesting directions for future work.

\section*{Acknowledgements}

Thank you to the reviewers for their helpful anonymous feedback. 
This research was done in part during the authors' visit to the Simons Institute for the Theory of Computing at UC Berkeley. 
Shunhua Jiang is supported by the ERC Starting Grant \#101039914. Lawrence Li is supported by grants awarded to Sushant Sachdeva --- an NSERC Discovery grant RGPIN-2025-06976 and an Ontario Early Researcher Award (ERA) ER21-16-283. 
Aaron Sidford was funded in part by a Microsoft Research Faculty Fellowship, NSF CAREER Award CCF-1844855, NSF Grant CCF1955039, and a PayPal research award. %

\addcontentsline{toc}{section}{References}
\bibliographystyle{alpha}
\bibliography{ref}

\appendix
\section{Missing Proofs}\label{sec:missing_proofs}

\subsection{Proof of LP Initialization}
\begin{proof}[Proof of \Cref{lem:LP_initialization}]
{\bf Part 1.} By definition we have $\ov{\ma}^{\top} \ov{x}^{(\init)} = \ov{b}$, $\ov{\ma} 0_n + \ov{c} = \ov{c}$, $\ell \leq x^{(\init)} \leq u$, and $0 \leq \wt{x}^{(\init)} \leq 2 \wt{x}^{(\init)}$, so the point is feasible.

Next we bound centrality. Recall that the barrier for the interval $[\ell, u]$ is $\phi(x) = - \log(x - \ell) - \log(u - x)$, and its derivatives are $\phi'(x) = -\frac{1}{x-\ell} + \frac{1}{u-x}$, and $\phi''(x) = \frac{1}{(x-\ell)^2} + \frac{1}{(u-x)^2} \geq \frac{1}{(u-\ell)^2}$. For the old constraints, we have $\frac{1}{(u_i-\ell_i)^2} \geq \frac{1}{\Xi^2}$. For the new constraints, we have
\begin{align*}
\frac{1}{(2\wt{x}_i^{(\init)})^2} \geq \frac{1}{(2\|b - \ma^{\top} x^{(\init)}\|_{\infty} / \beta)^2} 
\geq \frac{1}{4 \Xi^2}.
\end{align*}

Since $\tau(x) \geq \frac{n}{m}$, and since $\phi'(\ov{x}^{(\init)}) = 0$, we can bound centrality by
\begin{align*}
\left\|\frac{\ov{c} + \mu \tau(\ov{x}^{(\init)}) \phi'(\ov{x}^{(\init)})}{\mu \tau(\ov{x}^{(\init)}) \sqrt{\phi''(\ov{x}^{(\init)})}}\right\|_{\infty} 
\leq &~ \max\left\{ \frac{\|c\|_{\infty}}{\mu \frac{n}{m} \frac{1}{\Xi}}, \frac{\frac{2 \|c\|_1}{\delta'}}{\mu \frac{n}{m} \frac{1}{2 \Xi}} \right\} \\
\leq &~ \frac{4 \|c\|_1 \Xi m}{\mu n \delta'} \leq \epsilon,
\end{align*}
where the last step follows from $\mu = \frac{4 m \|c\|_1 \Xi}{\epsilon \delta'}$.

{\bf Part 2.} First note that for any vector $x \in \R^m$ that is feasible for the original LP, the vector $\ov{x} = \begin{bmatrix}x \\ 0_n \end{bmatrix} \in \R^{m+n}$ is feasible for the modified LP. This means
\begin{align*}
\min_{\substack{\ov{\ma}^{\top} \ov{x} = \ov{b} \\ \ov{\ell} \leq \ov{x} \leq \ov{u}}} \ov{c}^{\top} \ov{x} \leq \min_{\substack{\ma^{\top} x = b \\ \ell \leq x \leq u}} c^{\top} x.
\end{align*}
Using this we can bound the objective value as follows:
\begin{align*}
c^{\top} x^{(\final)} \leq &~ \min_{\substack{\ov{\ma}^{\top} \ov{x} = \ov{b} \\ \ov{\ell} \leq \ov{x} \leq \ov{u}}} \ov{c}^{\top} \ov{x} + \delta - \frac{2\|c\|_1}{\delta'} \cdot 1_n^{\top} \wt{x}^{(\final)} \\
\leq &~ \min_{\substack{\ma^{\top} x = b \\ \ell \leq x \leq u}} c^{\top} x + \delta,
\end{align*}
where the second step follows from the equation above, and that $\wt{x}^{(\final)} \geq 0$ since $\ov{x}^{(\final)}$ is feasible.

Next we bound the error in feasibility. First note that
\begin{align*}
c^{\top} x^{(\final)} + \frac{2\|c\|_1}{\delta'} \cdot 1_n^{\top} \wt{x}^{(\final)} \leq \min_{\substack{\ov{\ma}^{\top} x = \ov{b} \\ \ov{\ell} \leq x \leq \ov{u}}} \ov{c}^{\top} x + \delta 
\leq \min_{\substack{\ma^{\top} x = b \\ \ell \leq x \leq u}} c^{\top} x + \delta.
\end{align*}
Since $\ell \leq x^{(\final)} \leq u$, and the same bounds hold for the $x$ in $\min_{\substack{\ma^{\top} x = b \\ \ell \leq x \leq u}} c^{\top} x$, 
\begin{align*}
c^{\top} x^{(\final)} \geq \sum_{i=1}^m \min\{c_i \ell_i, c_i u_i\}, ~~~ \min_{\substack{\ma^{\top} x = b \\ \ell \leq x \leq u}} c^{\top} x \leq \sum_{i=1}^m \max\{c_i \ell_i, c_i u_i\}.
\end{align*}
Combining this and the inequality above yields
\begin{align*}
\sum_{i=1}^m \min\{c_i \ell_i, c_i u_i\} + \frac{2\|c\|_1}{\delta'} \cdot 1_n^{\top} \wt{x}^{(\final)} \leq \sum_{i=1}^m \max\{c_i \ell_i, c_i u_i\} + \delta.
\end{align*}
Next note that for an $i \in [m]$, $\max\{c_i \ell_i, c_i u_i\} - \min\{c_i \ell_i, c_i u_i\} = |c_i| \cdot (u_i - \ell_i) \leq |c_i| \cdot \Xi$, which implies
\begin{align*}
&~ \frac{2\|c\|_1}{\delta'} \cdot 1_n^{\top} \wt{x}^{(\final)} \leq \sum_{i=1}^m |c_i| \cdot \Xi + \delta \\
\implies &~ 1_n^{\top} \wt{x}^{(\final)} \leq \delta' \cdot (\frac{\Xi}{2} + \frac{\delta}{2 \|c\|_1}).
\end{align*}

Finally, since $\ov{x}^{(\final)}$ is feasible,  $\ma^{\top} x^{(\final)} + \beta \cdot \mdiag(\sigma) \cdot \wt{x}^{(\final)} = b$, and so 
\begin{align*}
\|\ma^{\top} x^{(\final)} - b\|_{\infty} = &~ \beta \cdot \|\wt{x}^{(\final)}\|_{\infty} \\
\leq &~ (\|b - \ma^{\top} x^{(\init)}\|_{\infty} / \Xi + 1) \cdot \delta' \cdot (\frac{\Xi}{2} + \frac{\delta}{2 \|c\|_1}) \\
\leq &~ \Big(\|b\|_{\infty} \Xi^{-1} + m\|\ma\|_{\infty} \max\{\|u\|_{\infty}, \|\ell\|_{\infty}\} \Xi^{-1} + 1\Big) \cdot \delta' \cdot (\frac{\Xi}{2} + \frac{\delta}{2 \|c\|_1}) \\
\leq &~ \delta,
\end{align*}
where the last step follows from $\delta' \defeq \frac{\delta}{10 m W^2} \cdot \min\{1, \|c\|_1\}$.

{\bf Part 3.} By definition, $\ov{\ma}^{\top} \ov{\ma} = \ma^{\top} \ma + \beta^2 I \succeq I.$
\end{proof}

\subsection{Equivalent Statement of Small Rayleigh Quotient Condition}
\begin{lemma}\label{lem:bound_on_v_topv_other_direction}
Let $\mm$ be any PSD matrix. For any $c \geq 0$, if $v$ is a unit vector that satisfies $1 - \left(v^{\top} v_{1}(\mm)\right)^2 \leq \frac{c \cdot \lambda_1(\mm)}{\lambda_n(\mm)}$, then $v^{\top} \mm v \leq (1 + c) \cdot \lambda_1(\mm)$.
\end{lemma}
\begin{proof}
\begin{align*}
v^\top \mm v 
= &~ \sum_{i \in [n]} \lambda_i(\mm) (v^\top v_i(\mm))^2 \\
\leq &~ \lambda_1(\mm) + \sum_{i = 2}^{n}  \lambda_n(\mm) (v^\top v_i(\mm))^2
= \lambda_1(\mm) + \lambda_n(\mm) \left(1 - (v^\top v_1(\mm))^2\right)
\leq (1+ c) \lambda_1(\mm)\,. \qedhere
\end{align*}

\end{proof}

\end{document}